\documentclass{article}
\usepackage[utf8]{inputenc}
\usepackage{enumitem}
\usepackage{pdfpages}
\usepackage[margin=1in]{geometry} 
\usepackage{amsmath,amsthm,amssymb}
\usepackage[utf8]{inputenc}
\usepackage{graphicx}
\usepackage[justification=centering]{caption}
\usepackage{subfigure}
\usepackage{subcaption}

\newcommand{\R}{\mathbb{R}}
\usepackage[colorlinks=true, allcolors=blue]{hyperref}
\usepackage{float}
\usepackage{mathtools,cases}
\usepackage{bbm}
\usepackage{calrsfs}
\usepackage{breqn}
\usepackage{multirow}
\usepackage{xcolor}
\usepackage{comment}
\usepackage{soul}
\usepackage{bbold}
\usepackage{bm}
\usepackage{pifont}
\usepackage{parskip}
\usepackage{algorithm}
\usepackage{algpseudocode}
\usepackage{array}
\usepackage{float}    
\usepackage{makecell} 
\usepackage{adjustbox} 
\usepackage{caption}   

\usepackage[round]{natbib}
\newtheorem{theorem}{Theorem}[section]
\newtheorem{lemma}[theorem]{Lemma}

\newtheorem{proposition}[theorem]{Proposition}

\newtheorem{remark}[theorem]{Remark}
\newtheorem{example}[theorem]{Example}

\newtheorem{definition}[theorem]{Definition}

\numberwithin{equation}{section}

\def\E{{\mathbb{E}}}

\def \d{\mathrm{d}}

\def\Cov#1#2{{\rm Cov}\left(#1,\, #2\right)}

\renewcommand{\d}{{\rm d}}

\usepackage{enumitem}

\usepackage{mathtools}
\usepackage{authblk}
\mathtoolsset{showonlyrefs}
\title{Heath–Jarrow–Morton meet lifted Heston in energy markets for joint historical and implied calibration}

\author[1]{Eduardo Abi Jaber\thanks{eduardo.abi-jaber@polytechnique.edu. I am grateful for the financial support from the Chaires FiME-FDD, Financial Risks, Deep Finance \& Statistics and Machine Learning and systematic methods in finance at Ecole Polytechnique.}}
\author[1]{Soukaïna Bruneau}
\author[2,3]{Nathan De Carvalho\thanks{nathan.decarvalho@engie.com. I am grateful for the financial support provided by Engie Global Markets.}}
\author[1,2]{Dimitri Sotnikov\thanks{dimitri.sotnikov@gmail.com. I am grateful for the financial support provided by Engie Global Markets.}}
\author[2]{Laurent Tur}
\affil[1]{École Polytechnique, CMAP}
\affil[2]{Engie Global Markets}
\affil[3]{Université Paris Cité, LPSM}

\definecolor{h_color}{rgb}{0.118, 0.827, 0.722}
\definecolor{g_color}{HTML}{FF5500}
\definecolor{sigma_color}{rgb}{0.00, 0.666, 1}
\definecolor{v_color}{HTML}{D400FF}

\usepackage{tikz}
\usetikzlibrary{positioning, arrows.meta}
\usepackage{todonotes}

\usepackage[normalem]{ulem}

\begin{document}

\maketitle

\begin{abstract}
In energy markets, joint historical and implied calibration is of paramount importance for practitioners, yet notoriously challenging due to the need to align historical correlations of futures contracts with implied volatility smiles from the option market. We address this crucial problem with a multiplicative multi-factor Heath-Jarrow-Morton (HJM) model for forward curves, combined with a stochastic volatility factor coming from the lifted Heston model. We develop a sequential fast calibration procedure leveraging the Kemna-Vorst approximation of futures contracts: (i) historical correlations and the Variance Swap (VS) volatility term structure are captured through Level, Slope, and Curvature factors, (ii) the VS volatility term structure can then be corrected for a  perfect match via a fixed-point algorithm, (iii) implied volatility smiles are calibrated using Fourier-based techniques. The main advantage of the proposed calibration framework is the decoupling of the calibration steps: each step tackles a simpler calibration subproblem and guaranties that the previously optimized parameters remain unchanged. Our model displays remarkable joint historical and implied calibration fits on the German power market and enables realistic interpolation within the implied volatility hypercube.
\end{abstract}

\begin{description}
\item[Mathematics Subject Classification (2010): 91G20, 91G60] 
\item[JEL Classification: C02, Q41, C63] 
\item[Keywords:] Energy markets, Multi-factor HJM, Nelson-Siegel, Stochastic Volatility, Lifted Heston, Kemna-Vorst, Calibration
\end{description}

\newpage

\tableofcontents

\newpage

\section*{Introduction}

In power markets, futures contracts deliver electricity continuously over a fixed period, rather than on a fixed delivery date as is typical for commodities like oil. These contracts are settled either physically or financially with respect to the average spot price of electricity over the delivery period. Due to the inability to efficiently store electricity in large quantities, these instruments are often referred to as \textit{swaps}, as the holder effectively exchanges the forward rate against the spot price of the commodity.

The unique characteristics of electricity markets, such as the need to maintain equilibrium between real-time production and consumption, distinguish them from other commodity markets. Real-time delivery is typically managed through intra-day and imbalance markets, which align production and consumption levels and can exhibit frequent price spikes or even negative prices during periods of significant imbalance. Moreover, financial futures with very short-term deliveries often show low correlation with long-term futures, which are highly correlated among themselves. This behavior is accompanied by exponentially increasing realized volatility as the time to delivery reduces, the so-called \cite{samuelson2016proof} effect.

Following the deregulation of European electricity markets, an extensive body of literature has emerged on the modeling of electricity markets. We refer to the survey by \cite*{deschatre2021survey} and the book by \cite*{benth2008stochastic} for an overview of key modeling approaches and market dynamics.

Several approaches to energy market modeling are considered in the literature, differing in the choice of the initial stochastic quantity to be modeled. The first class of models focuses on the spot price process, as seen in works like \cite*{Mishura2023GaussianVP, Schmeck2021TheEO, cortazar2017multifactor}. Another approach, inspired by the LIBOR market model , see \cite*{Mercurio2006}, models futures contracts with specific delivery periods (e.g., monthly contracts, as in \cite*{Kiesel2009}, \cite*{gardini2023heath}), using these contracts as building blocks to derive the prices of other contracts under no-arbitrage conditions. The third class of models, inspired by the well-known \cite*{HJM1992} (HJM) interest rate model, takes infinitesimal futures contracts as a starting point and uses them to reconstruct futures contracts with any delivery period. 
This type of model, adapted to the energy-markets setting by \cite{Clewlow1999}, ensures consistent dynamics across all futures contracts while also guaranteeing the 
absence of arbitrage between futures contracts with overlapping delivery periods.
The ability to model the dynamics of the entire futures curve, that is, the joint evolution of all futures contracts, is essential for pricing and hedging derivatives that are sensitive to cross-contract correlations, such as swing and spread options \citep{carmona2003pricing}.

Since the HJM approach naturally leads to modeling the entire futures curve, it is common to work with stochastic processes taking values in infinite-dimensional Hilbert spaces. In particular, several classes of infinite-dimensional stochastic volatility models have been developed; see, for instance, \citet*{cox2022infinite}. \citet{benth2018heston} consider an infinite-dimensional Heston-type model, \citet{benth2018space} employ an NIG Lévy process as the stochastic driver, and \citet{benth2021infinite} analyze infinite-dimensional Volterra-type dynamics. For a more systematic treatment of infinite-dimensional modeling, we refer the reader to the monograph by \citet{benth2023stochastic}.

Early commodity models primarily focus on historical calibration, which involves calibrating the covariance structure of traded futures contracts' returns; see for example \cite{Andersen2010} for the calibration of a HJM model to gas prices, \cite*{edoli2013calibration} for a multi-underlyings calibration using the quadratic variations of two-factor models for each market, \cite*{gardini2023heath} where they calibrate a LIBOR market Black-Scholes-type factor model using historical swap prices, and \cite*{feron2024estimation} for a historical calibration of a multi-factor HJM model using maximum likelihood and a Kalman filter where they also comment on the number of factors to capture the historical covariance of futures' returns.

Since April 2024, brokers have started quoting smiles for vanilla options on German power. These vanilla options are written on monthly contracts, quarterly contracts, and calendar (yearly) contracts, which can overlap. For example, the first quarter of 2025 and the calendar 2025 can have quoted smiles. Another specificity of the power market is that for one calendar underlying, three or four smiles can be quoted. With the increasing liquidity of the electricity derivatives market in Europe, there has been significant growth in research in option pricing in power market models, see \cite*{Schmeck2021TheEO, cortazar2017multifactor, benth2017additive}. However, few studies as in \cite*{piccirilli2021capturing, Fanelli2016path-dep} address the challenge of implied calibration, which involves calibrating model parameters to fit available option prices. Note that implied calibration in the energy market is exceptionally challenging, as it requires the simultaneous calibration of multiple volatility surfaces associated with different underlyings, which are interconnected through the futures price curve.

In energy markets, the limited liquidity of options and the relatively small set of underlyings for which they are traded make it impossible to calibrate pricing models using option prices alone, a marked contrast to more developed derivatives markets, such as equities or interest rates. Consequently, a crucial requirement for any pricing model in this low-liquidity environment is the ability to produce reasonable interpolations and extrapolations of option prices, not only across strikes and maturities but also across different futures contracts, in order to fill the missing regions of the implied volatility hypercube.
Moreover, as discussed above, calibrating the correlation structure is essential, yet cannot be inferred from option prices. Historical calibration therefore becomes indispensable for practical implementation. At the same time, given the gradually increasing liquidity of options, the options market can no longer be neglected, both because of regulatory requirements and due to the arbitrage opportunities that arise under purely historical models.
The most significant limitation of the approaches mentioned above is that they allow for either historical calibration or implied calibration, but not both simultaneously.
Hence the following questions:
\begin{center}
    \textit{Is joint historical and implied calibration possible in energy markets? 
    \\ If so, can we do it with a tractable model that provides a consistent interpolation of option prices?
    }
\end{center}

We answer both questions \textit{affirmatively} with a multiplicative multi-factor \cite*{HJM1992} (HJM) model for forward curves, combined with a stochastic volatility factor coming from the lifted Heston model of \cite{lifted2019}. To the best of our knowledge, this paper is the first to address both calibration problems with a single model, taking into account the entire implied volatility smile.

On the one hand, by capturing historical covariances and implied volatility levels, our model provides insights on the futures correlation structure between liquid and non-liquid futures contracts. On the other hand, implied calibration allows the entire implied volatility smile, not just the At-The-Money (ATM) volatilities, to be represented by the model. Moreover, the proposed implied calibration takes into account potential multiple maturities for a given underlying futures contract. Thus, the joint historical and implied calibration matches both the correlation between futures and quoted smiles --— a problem of paramount importance for market practitioners --- leading to more accurate pricing of exotic contracts like Asian options or swing options. One potential difficulty is that implied calibration may affect the historical one, but we will see that this influence is negligible. To the best of our knowledge, this paper is the first to address both calibration problems with a single model, taking into account the entire implied volatility smile.

\paragraph{Contributions.} More precisely, to solve the joint historical and implied calibration problem, we introduce in Section \ref{S:model} an HJM model with
\begin{itemize}
    \item[(i)] parsimonious parametric Level, Slope and Curvature risk-factors in order to capture both historical covariances of rolling futures' log returns and implied volatility levels with a few factors,

    \item [(ii)] two piece-wise constant functions to perfectly match the implied volatility term structure, including early maturities,

    \item[(iii)] a  stochastic volatility component coming from the lifted Heston model with three time-scales to match the implied volatility skews.
\end{itemize}
Such model is by construction arbitrage-free with respect to futures contracts with overlapping delivery periods.

In Section \ref{s:market_data}, we present our Market data. We recall the typical liquid ``absolute'' futures quoting on power markets, we detail a stripping algorithm to construct rolling futures contracts satisfying absence of overlapping arbitrage and used to estimate historical covariances. Furthermore, we propose a novel multi-contract SSVI parametrization to extract Variance Swap (VS) volatilities from listed options.

Then,  we detail a novel three-step sequential calibration methodology relying on the \cite{KEMNA1990113} (KV) approximation  of futures contracts in order to
\begin{itemize}
    \item[1)] jointly capture the historical covariances of rolling futures' daily log returns and the implied VS volatility levels of absolute futures via a non-linear -- linear cone program,  see Section \ref{s:joint_cov_vs_calib},
    \item[2)] correct and fit perfectly the VS volatility term structure via a fast fixed-point algorithm based on VS prices, see Section \ref{section:gh_calib},
    \item[3)] fit the volatility smile's shapes via Fourier inversion techniques as in \cite{Lewis2001} for fast and efficient vanilla option pricing, see Section \ref{S:fastpricing}. 
\end{itemize}
We emphasize that the calibration procedure is fully decoupled: the three steps are performed sequentially, with each step preserving the calibration results achieved in the previous ones. 
Moreover, thanks to both the KV approximation and the decoupled structure, each step remains highly tractable. Decomposing the calibration into simpler subproblems yields a procedure that is a priori faster and more efficient than calibrating all parameters simultaneously. Our model displays remarkable joint historical and implied calibration fits  to both German power and TTF gas markets.  In order to validate our calibration methodology, we show a posteriori how close the KV approximated futures are to the exact arbitrage-free futures in terms of sample trajectories, implied volatility smiles and correlations between futures contracts. Finally, we show that such a fully calibrated model can be used to interpolate the implied volatility hypercube in a consistent manner. The main calibration results are collected in Section \ref{s:numerical_results}. 
Additional model  algorithmic insights and  calibration results  are postponed to the appendices.

\paragraph{Related literature.} The paper of \cite*{piccirilli2021capturing} is the work most closely related in spirit to our approach, although there are several important differences. The authors propose a two-factor model with Normal Inverse Gaussian Lévy factors, while our model is a stochastic volatility model with a continuous process as the variance. Their calibration procedure considers only one smile per contract, ignoring multiple maturities for a given futures contract. Our model accounts for these ``early maturities'', enabling a more refined calibration of the volatility term structure. 

\paragraph{Notations.}
For $N \in \mathbb{N}^{*}$, we denote by $\mathbb{S}_{++}^{N}$ (resp. $\mathbb{S}_{+}^{N}$) the set of $N \times N$ definite (resp. semi-definite) positive matrices, and $\|.\|_{w}$ denotes a weighted Euclidean norm with weights $w := \left( w_{i} \right)_{i \in \{1, \ldots, N \}^{2}} \in \mathbb{R}_{+}^{N}$ such that $\|u\|_{w} := \sqrt{\sum_{i=1}^{N} w_{i} u_{i}^{2}}, \; u \in \mathbb{R}^{N}$. We define similarly the weighted Frobenius norm $\|A\|_{\Gamma} := \sqrt{\sum_{i=1}^{N} \sum_{j=i}^{N} \Gamma_{i,j} A_{i,j}^{2}}, \; A \in \mathbb{R}^{N \times N}$, with matrix weights $\Gamma := \left( \Gamma_{i,j} \right)_{i,j \in \{1, \ldots, N \}^{2}} \in \mathbb{R}_{+}^{N \times N}$. We will omit the indices $w$ and $\Gamma$ for standard Euclidean and Frobenius norms $\|u\|$ and $\|A\|$ respectively. We will also denote by $\|\cdot\|_{\infty}$ the infinity norm $\|u\|_{\infty} := \max_{i \in \{1, \ldots, N\}}|u_{i}|$ and $\|A\|_{\infty} := \max_{i, j \in \{1, \ldots, N\}}|A_{ij}|$.

\section{The model: HJM with lifted Heston}\label{S:model}

{Fix a filtered probability space $\left( \Omega, \mathcal F, \left( \mathcal{F}_t \right)_{t \geq 0}, \mathbb Q \right)$ satisfying the usual conditions, where $\mathbb Q$ represents the risk-neutral probability.} We model the futures price curve {under $\mathbb Q$}, i.e.~the infinitesimal futures contract prices $\left( f(t, T) \right)_{0 \leq t \leq T < \infty}$, à la \citet*{HJM1992} (HJM), enhanced with a stochastic volatility component coming from the lifted Heston model of \citet{lifted2019} in the form
\begin{equation}\label{eq:HJM_def}
     \frac{\d f(t,T)}{f(t,T)} =  \textcolor{black}{g(T)}\textcolor{black}{h(t)}\textcolor{black}{\sqrt{V_t}}\sum_{i=1}^N \color{black} \sigma_{i}(t,T) \color{black} \d W_t^{i}, \quad f(0,T) \in \mathbb R_+, \quad 0 \leq t \leq T,
\end{equation}
where 
\begin{itemize}
    \item $W := (W_t)_{t \geq 0}$ is an $N$-dimensional Brownian motion  with a correlation matrix $R \in \mathbb{S}_{++}^{N}$.
    
    \item Each $\sigma_i:\R^2_+ \to \mathbb R, \; i \in \{1 \ldots, N \}$ is a deterministic continuous and bounded function, capturing the \textit{realized futures contracts' covariances} and the \textit{implied volatility levels}. We set  $\sigma = (\sigma_1, \ldots, \sigma_N)^\top$.
    
    \item $\color{black} V$ is a nonnegative stochastic variance process responsible for the \textit{implied smile} of  the form
    \begin{equation}\label{eq:V_def}
        \textcolor{black}{V_t} = 1 + \sum_{i=1}^M c_i U_t^i, \quad t \geq 0
    \end{equation}
    where the pseudo-factors $(U^i)_{i \in \{1,\ldots,M\}}$, weighted by $c_i \geq 0$,  are driven by the same Brownian motion $B$, but mean-revert at different speeds $ 0< x_1< x_2< \ldots < x_M$ such that
\begin{equation}\label{eq:U_def}
        \d U_t^i = -x_i U_t^i\, \d t + \sqrt{\textcolor{black}{V_t}}\, \d B_t, \quad U_0^{i} = 0, \quad i \in \{1, \ldots, M\}.
    \end{equation}
    Here  $B = (B_t)_{t \geq 0}$ is a one-dimensional Brownian motion correlated with $W$, {via $\left( \hat\rho_i \right)_{i \in \{1,\ldots, N\}} \in [-1,1]^{N} $}, to take into account \textit{the leverage effect} such that
    \begin{equation}\label{eq:B_def}
        B_t = \sum_{i=1}^N\hat\rho_i \widehat W_t^i + \sqrt{1 - \sum_{i=1}^N \hat\rho_i^2}\, W_t^\perp, \quad t \geq 0,
    \end{equation}
    where $W^\perp$ is a scalar Brownian motion independent from $\widehat W$ which is a standard $N$-dimensional Brownian motion constructed from $W$ via the Cholesky decomposition (see, for example, \cite{horn2013matrix})
    \begin{equation}\label{eq:Cholesky}
        R = L^\top L, \quad \widehat W = L^{-1} W, \quad \|\hat\rho\| \leq 1,
    \end{equation}
    so that
    \begin{equation} \label{eq:get_rho_tilde_from_rho_hat}
        \d \langle B, W^i\rangle_t = \tilde\rho_i\,\d t, \quad \tilde\rho = L\hat\rho.
    \end{equation}

    \item $T \mapsto g(T)$ and $t \mapsto h(t)$ are deterministic bounded positive functions \textit{correcting the implied volatility levels} to match the volatility term structure for the contracts with several maturity dates.
\end{itemize}

For a fixed $t \in [0,T]$, $f(t, T)$ represents the quote observed at date $t$ of the contract that delivers a unit amount of commodity between dates $T$ and $T+\d T$, where $\d T>0$ denotes an infinitesimal amount of time. The special case $t=T$ is seen as the spot price $S$ of such a commodity: $S_{t} := f(t,t)$, which is well-defined as soon as $\lim_{t \to T} f(t,T) < \infty$.

We prove in the next theorem that the model \eqref{eq:HJM_def} is indeed well-defined leveraging results from \cite{lifted2019}. In particular, we note that although the different factors $U^i$ can become negative, the variance process $V$ is always nonnegative \cite[Theorem~7.1]{jaber2019affinevolterraprocesses}, as illustrated on  Figure \ref{fig:var_traj}. Furthermore, the variance process \eqref{eq:V_def} is Markovian in the state variables $U := (U^i)_{i \in \{1, \ldots, M\}}$ with a state space corresponding to the set of $ u \in \mathbb{R}^M$ such that
$$
{1 + }\sum_{j=1}^M c_j u_j \geq 0 \quad \text{and} \quad {\mu\sum_{j=1}^i \frac{c_j}{x_j} +} \sum_{j=1}^i c_j u_j \geq {\mu\sum_{j=1}^i\frac{c_j}{x_{i + 1}} +} \sum_{j=1}^i c_j u_{i+1}  \quad \text{for } i \in \{1, \ldots, M-1\},
$$
where $\mu = \left(\sum_{k=1}^M\frac{c_k}{x_k}\right)^{-1}$ see \cite*{abijaber2024state}.

\begin{remark}
    We chose the lifted Heston model for two reasons. First, its affine structure allows for fast vanilla option pricing using Fourier techniques, which is highly beneficial for calibration speed. Second, it offers significantly more flexibility in the shapes of implied volatility surfaces that can be calibrated, compared to the simpler standard Heston model, as we demonstrate in Appendix~\ref{section:attainable}. Thus, the lifted Heston model strikes a good balance between numerical tractability and universality. However, we emphasize that the modeling framework we propose is not limited to this model and is flexible enough to accommodate any stochastic variance process $V_t$ satisfying $\E[V_t] \equiv 1$.
\end{remark}

\begin{figure}[H]
\begin{center}
\includegraphics[width=1\linewidth]{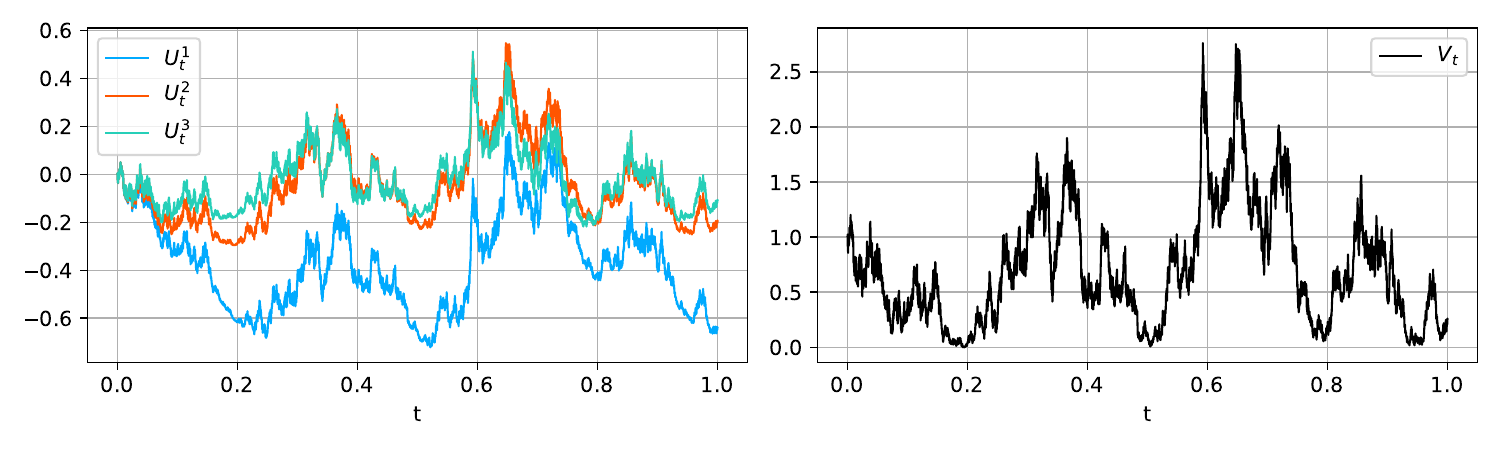}
\caption{Trajectory of the stochastic variance factors $(U^i)_{i=1,2,3}$ (on the left) and the resulting stochastic variance $V$ (on the right) corresponding to the parameters $x = (4.6\cdot 10^{-6},\, 9.712,\, 20.249)$ and $c = (0.492,\, 0.68,\, 2.79)$. The simulation scheme is described in Appendix \ref{section:monte_carlo_scheme}.}
\label{fig:var_traj}
\end{center}
\end{figure}

\begin{theorem}\label{Thm:lifted_heston_existence}
    Fix $T > 0$. Let  $g,h:[0,T]\to \mathbb R$ and $\sigma(\cdot,T):[0,T]\to \mathbb R^N$ be bounded and measurable functions. Then, there exists a unique strong solution $(U^i)_{i \in \{1,\ldots, M\}}$ to \eqref{eq:U_def} such that $V$ given by  \eqref{eq:V_def} remains non-negative. Furthermore, the process $f(\cdot,T)$ defined by 
    \begin{align}\label{eq:f_explicit}
        f(t,T) = f(0,T) \exp\left( -\frac{g^2(T)}{2} \int_0^t h^2(s)V_s \sigma^\top (s,T)  R \sigma (s,T)\, \d s + g(T) \int_0^t h(s)\sqrt{V_s}  \sigma^\top (s,T)\, \d W_s \right), \quad t \in [0,\, T],
    \end{align}
    is the unique strong solution to \eqref{eq:HJM_def}. In particular, $f(\cdot, T)$ is a true martingale.
\end{theorem}

\begin{proof}
The equation \eqref{eq:HJM_def} can be rewritten as a one-dimensional diffusion,
\begin{equation}\label{eq:f_one_dim_diff}
    \dfrac{\d f(t, T)}{f(t, T)} = \textcolor{black}{h(t)}\textcolor{black}{\sqrt{V_t}}\sqrt{\sigma(t, T)^\top R \sigma(t, T)}\d \tilde W_t, 
    \quad d\langle B, \tilde W\rangle_t = \dfrac{\sigma(t, T)^\top \tilde\rho\,\d t}{{\sqrt{\sigma(t, T)^\top R \sigma(t, T)}}},
    \quad t \in [0, T],
\end{equation}
where the Brownian motion $\tilde W$ is given by 
\begin{equation}
    \tilde W_t = \sum_{i=1}^N\dfrac{\sigma_i(t, T)}{\sqrt{\sigma(t, T)^\top R \sigma(t, T)}}W_t^i, \quad 0 \leq t \leq T.
\end{equation}

The factor processes $(U^i)_{i \in \{1,\ldots, M\}}$ can be written as
\begin{equation}
    U_t^i = \int_0^t e^{-x_i (t - s)}\sqrt{V_s}\, \d B_s, \quad i \in \{1, \ldots, M\}, \quad 0 \leq t \leq T,
\end{equation}
so that the variance dynamics \eqref{eq:V_def} reads
\begin{equation}\label{eq:V_volterra_formulation}
    V_t = 1 + \int_0^t \left(\sum_{i=1}^M c_i e^{-x_i (t - s)}\right)\sqrt{V_s}\, \d B_s = 1 + \int_0^t K(t - s)\sqrt{V_s}\, \d B_s, \quad 0 \leq t \leq T,
\end{equation}
where $K(t) := \sum_{i=1}^M c_i e^{-x_i t}$. The equation \eqref{eq:V_volterra_formulation} means that $V$ is a Volterra square-root process understood in the sense of \cite*[Section 6]{jaber2019affinevolterraprocesses}. An application of \cite[Theorem A.1]{lifted2019} yields the existence and uniqueness of strong solution $(U^i)_{i \in \{1,\ldots, M\}}$, such that the variance process $V \geq 0$, as well as the existence and uniqueness of the strong solution to \eqref{eq:f_one_dim_diff} given by \eqref{eq:f_explicit}.  The proof of the  martingality of $f(\cdot, T)$ follows exactly the proof of \cite*[Lemma 7.3]{jaber2019affinevolterraprocesses} taking into account that the correlation coefficient in \eqref{eq:f_one_dim_diff} is time-dependent.
\end{proof}

In practice, the infinitesimal futures contracts given by \eqref{eq:HJM_def} are not observed in the market and therefore must be related to market traded futures contracts. Using practitioners' vocabulary, ``absolute'' futures contracts delivering electricity on a fixed calendar delivery period $[T_{s}, T_{e}]$ are typically quoted in organized markets until a few days before their first delivery date i.e.~for dates $t$ such that $0 \leq t \leq T_{s}-\delta$, $\delta$ equal to a few days. We recall in Section \ref{ss:typical_futures} the typical futures quoting in power markets. From the model perspective, the unitary absolute futures contract delivering continuously a unitary power unit of electricity over $[T_s, T_e]$ is given by\footnote{We assume a zero discount rate ($r = 0$) for simplicity.}
\begin{equation} \label{eq:futures_contract_def}
    F_t(T_s,T_e) := \frac{1}{T_e-T_s} \int_{T_s}^{T_e} f(t,T) \d T, \quad 0 \leq t < T_s.
\end{equation}
ensuring the absence of arbitrage opportunity for futures contracts with overlapping delivery periods.

On the other hand, a ``rolling'' futures contract is a contract that depends on the observation date, $t \geq 0$, and maintains a constant time to delivery, $T_{s} > 0$, with a fixed and contiguous delivery period of duration $T_{e}-T_{s} > 0$. As a result, the contract’s delivery period adjusts as the observation date changes, ensuring that the contract always quotes. However, its quote is typically not directly observable in the market. Instead, it must be derived using no-arbitrage principles based on market quotes available at each observation date. Its quote is expressed in our model by the formula
\begin{equation} \label{eq:rolling_forward_model_definition}
    F_t(t+T_s,t+T_e) = \frac{1}{T_e-T_s} \int_{t+T_s}^{t+T_e} f(t,T) \d T, \quad t \geq 
    0.
\end{equation}

\begin{example}[Distinction between absolute and rolling futures]
    For example, on the $20^{th}$ of August 2024, the next absolute monthly futures contract quoting on the market corresponds to the contract September 2024 delivering electricity between the $1^{st}$ to the $30^{th}$ of September 2024. Fixing $T_{s}$ to seven days, and $T_{e}-T_{s}$ to thirty days, one can define a month-ahead rolling contract delivering electricity from the $27^{th}$ of August to the $25^{th}$ of September 2024. On the $21^{st}$ of August 2024, such rolling contract becomes the one delivering electricity from the $28^{th}$ of August to the $26^{th}$ of September 2024, while the forward September 2024 keeps the same delivery period.
\end{example}

\begin{remark}
In the discussion above, we assumed that the risk-free rate is $r = 0$. In the general case, however, a similar no-arbitrage argument \citep[Section 4]{benth2008modeling} yields
\begin{equation*}
    F_t(T_s, T_e) := \int_{T_s}^{T_e} w(T)\, f(t,T)\, \mathrm{d}T, 
    \qquad 0 \le t < T_s,
\end{equation*}
where
$w(T) := \frac{r e^{-rT}}{e^{-rT_s} - e^{-rT_e}}.$
Thus, the extension of our approach to a non-zero constant interest rate is straightforward, although it makes the resulting expressions and computations more cumbersome. For notational convenience, we will therefore assume throughout the paper that $r = 0$.
\end{remark}

\subsection{Correlation structure in the model}

One of the main objectives of the HJM modeling framework is the consistent joint simulation of multiple futures contracts. This capability is essential for pricing derivatives that depend on several underlyings and are sensitive to their correlations, such as spread options. Consequently, the correlation structure between futures contracts plays a crucial role in both pricing and hedging such derivatives and should therefore be carefully calibrated.  

We focus on the calibration of \emph{rolling covariances}, that is, the covariances between returns of rolling futures contracts, defined as
\begin{equation}\label{eq:rolling_cov_def}
    \mathrm{Cov}\left(r_{t}^{k_1}(\tau_{d}),\, r_{t}^{k_2}(\tau_{d})\right), 
    \quad \text{where} \quad 
    r_{t}^k(\tau_{d}) := \log \frac{{F}_{t} (t+T_{s}^k, t+T_{e}^k)}{{F}_{t-\tau_{d}} (t+T_{s}^k, t+T_{e}^k)},
\end{equation}
where the indices $k_1$ and $k_2$ correspond to the rolling contracts under consideration, and $\tau_{d}$ denotes the time step used for computing returns, typically set to one day.  

The main motivation for calibrating rolling covariances is that they are assumed to be approximately constant over time and can therefore be estimated from historical data. The computation of rolling covariances in the model will be discussed in Section~\ref{s:joint_cov_vs_calib}, while the historical estimation procedure is outlined in Appendix~\ref{ss:covariance_estimation}.  

However, although rolling covariances are convenient for calibrating the covariance structure of historical factors, pricing problems require dealing with \emph{instantaneous covariances} of \emph{absolute contracts}, given by
\begin{equation}\label{eq:absolute_inst_correl_def}
    \rho_{ij}(t) = 
    \frac{\left\langle \d \log {F}(T_{s}^i, T_{e}^i),\, \d \log {F}(T_{s}^j, T_{e}^j)\right\rangle_t}
    {\sqrt{\langle \d \log {F}(T_{s}^i, T_{e}^i)\rangle_t}\,
     \sqrt{\langle \d \log {F}(T_{s}^j, T_{e}^j)\rangle_t}},
\end{equation}
which have a direct impact on option prices, but cannot be estimated historically. The numerical precision of computing these correlations is discussed in Subsection~\ref{ss:validation_by_MC_of_kv_futures}.  

We now turn to an approximation that renders the model tractable and, in particular, allows us to derive closed-form expressions for \eqref{eq:rolling_cov_def} and \eqref{eq:absolute_inst_correl_def}.

\subsection{The Kemna--Vorst approximation} \label{section:KV_approx_presentation}
The dynamics of the futures contract $F_{.}(T_s, T_e)$ cannot be written explicitly in our model, since \eqref{eq:futures_contract_def} involves an arithmetic mean, and not a geometric one. For this reason, we will use the \citet{KEMNA1990113} approximation such that
\begin{align}
     \frac{\d F_t(T_s,T_e)}{F_t(T_s,T_e)} &\approx \dfrac{1}{T_e - T_s}\int_{T_s}^{T_e}\dfrac{\d f(t, T)}{f(t, T)}\,\d T \\
     &= h(t) \sqrt{V_t}\sum_{i=1}^N \left(\frac{1}{T_s-T_e}\int_{T_s}^{T_e} g(T) \sigma_i(t, T) \,\d T\right)\,\d W_t^i \\
     &= h(t) \sqrt{V_t} \Sigma_{t}(T_s, T_e)^\top \d W_t, \label{eq:kv_approx}
\end{align}
where
\begin{equation} \label{eq:futures_contracts_volatility}
    \Sigma_{.}(T_s, T_e) := \frac{1}{T_e-T_s} \int_{T_s}^{T_e} g(T)\sigma (., T) \d T.
\end{equation}
Consequently, the approximated dynamics $\widetilde F_{.}(T_s,T_e)$ of the futures contract is given by
\begin{equation}\label{eq:KV_def}
    \dfrac{\d \widetilde F_t(T_s,T_e)}{\widetilde F_t(T_s,T_e)} = \textcolor{black}{h(t)}\textcolor{black}{\sqrt{V_t}}\Sigma_t(T_s,T_e)^\top \d W_t, \quad \widetilde F_0(T_s,T_e)=F_0(T_s,T_e).
\end{equation}
We stress that $\Sigma_{.}(T_s, T_e) = \Sigma_{.}(T_s, T_e; g, \sigma)$ depends essentially on $g$ and $\sigma$, but for the sake of brevity, we omit these arguments. The equation \eqref{eq:KV_def} admits
\begin{align} \label{eq:kv_approximated_futures}
    \widetilde F_t(T_s,T_e) =  F_0(T_s,T_e) \exp\Bigg( &-\frac{1}{2} \int_0^t h^2(r)V_r \Sigma^\top_r(T_s,T_e)  R \Sigma_r(T_s,T_e) \, \d r \\&\quad\quad\quad\quad\quad \quad  + \int_0^t h(r)\sqrt{V_r}  \Sigma^\top_r(T_s,T_e)\, \d W_r \Bigg), \quad t \in [0,\, T_s],
\end{align}
as the unique strong solution, recall Theorem~\ref{Thm:lifted_heston_existence} for the existence and uniqueness of $V$.

Recall that futures contracts' quotes $F := F_{.}(T_s,T_e)$ defined by \eqref{eq:futures_contract_def} are free from arbitrage opportunities in the case of overlapping delivery periods. In contrast, arbitrage may arise in the Kemna-Vorst approximated futures prices $\widetilde{F} := \widetilde{F}_{.}(T_s, T_e)$ from \eqref{eq:kv_approximated_futures}, as they do not strictly satisfy \eqref{eq:futures_contract_def}. However, such approximated futures' quotes $\widetilde{F}$ offer a tractable approach for calibration including 
\begin{itemize}
    \item the explicit computation of their variance swaps' volatilities as shown in the next Section \ref{ss:model_vs_swap_and_vol},
    \item the explicit computation of their covariances as detailed in Section \ref{ss:model_covariance_vs_volatility},
    \item the fast pricing of call and put options written on such futures contracts using a Fourier inversion technique as shown in Subsection \ref{S:fastpricing}.
\end{itemize}
Once the model \eqref{eq:HJM_def} is calibrated using the Kemna-Vorst approximated futures $\widetilde{F}$, we will illustrate numerically a posteriori that, not only are the trajectories of the arbitrage-free futures' quotes' $F$ similar to the ones of $\widetilde{F}$, but also that their instantaneous futures correlation structure and associated vanilla option prices are nearly identical. Indeed, all these quantities with respect to futures with quotes given by \eqref{eq:futures_contract_def} can be approximated by Monte-Carlo techniques as detailed in Section \ref{ss:validation_by_MC_of_kv_futures}.

Since the numerical analysis of the Kemna–Vorst approximation for our model, presented in Subsection \ref{ss:validation_by_MC_of_kv_futures}, demonstrates that the approximated HJM model provides accurate pathwise and weak approximations and preserves the correlation structure well for contracts with the delivery periods of interest (i.e., one month, three months, and one year), we can use the approximated model interchangeably with the postulated model, thereby leveraging its tractability and analytical properties.

From now on, we fix a delivery period $[T_s,\, T_e]$ and use the approximation $\widetilde{F}$ everywhere instead of the exact dynamics. Accordingly, we also use the notation $\Sigma_{.} := \Sigma_{.}(T_s, T_e)$, leaving the delivery period implicit when there is no ambiguity.

\begin{remark}
    One may notice that there was no need for the approximation of futures contracts if the initial model \eqref{eq:HJM_def} was additive and not multiplicative. However, the additive specification has several dynamics' disadvantages, which we comment briefly in Appendix \ref{section:additive_model}.
\end{remark}

\subsection{Variance swap price and volatility} \label{ss:model_vs_swap_and_vol}

By definition, the \textit{variance swap price} with maturity $T>0$, denoted by $\mathrm{VS}_T$, is the expected integrated quadratic variation of the futures contract following the dynamics \eqref{eq:KV_def}
given by
\begin{equation} \label{eq:vs_price_model}
    \mathrm{VS}_T := \E \langle \log {F} \rangle_T \approx \E \langle \log \tilde{F} \rangle_T = \int_0^T{h(t)^2}\Sigma_t^\top R\Sigma_t \d t \geq 0,
\end{equation}
and the \textit{variance swap volatility} $\sigma_{\mathrm{VS}, T}$ satisfies the equation $\sigma_{\mathrm{VS}, T}^2T = \mathrm{VS}_T$, so that 
\begin{equation} \label{eq:vs_vol_model}
    \sigma_{\mathrm{VS}, T} = \sqrt{\dfrac{1}{T}\int_0^T{h(t)^2}\Sigma_t^\top R\Sigma_t\,\d t}.
\end{equation}

For more profound discussion of variance swaps and their properties, we refer reader to \citet{demeterfi1999more}.

\paragraph{Why considering variance swaps for calibration?}

The variance swap volatility is a zero-order approximation of the ATM volatility level as shown by \cite{BergomiGuyon2011}, allowing us to formalize the idea of the \textit{implied volatility smile level}. Moreover, formula \eqref{eq:vs_vol_model} coincides exactly with the implied volatility for maturity $T$ in the HJM model under deterministic volatility, i.e., when $V \equiv 1$.

In fact, in our model \eqref{eq:HJM_def}, variance swap prices \eqref{eq:vs_price_model} depend the functions $h$ and also $\sigma$, $g$ via the deterministic futures' volatilities $\Sigma_{.}(T_s, T_e)$ \eqref{eq:futures_contracts_volatility}, but not on the parameters of $V$ from \eqref{eq:V_def} since we conveniently set $\mathbb{E}V_{t} = 1, \; t \geq 0$. This allows to first calibrate the variance swap prices, i.e.~the smile level with the functions $\sigma, h$ and $g$, and then solve independently the smile shape calibration problem with the parameters of $V$.

Furthermore, the evaluation of the variance swap prices \eqref{eq:vs_price_model} turn out to be in practice less time-consuming than the evaluation of ATM volatilities since they are given by the closed formula \eqref{eq:vs_price_model}, while the ATM volatilities are extracted via the Lewis formula \eqref{eq:lewis_formula} which require the computation of the log-price characteristic function, see Section \ref{S:fastpricing} for more details.

Although Variance Swaps (VS) are not traded in power markets, their market prices $\mathrm{VS}_{T}^{\mathrm{mkt}}$ can be extracted efficiently from the prices of vanilla option quotes as shown in Section \ref{sect:vs_quote_extr}. Thus, given a family of futures contracts whose smiles are quoting, these VS volatilities \eqref{eq:vs_vol_model} define the implied volatility term structure that we will aim to calibrate by fitting a priori $\sigma$, and then perfecting the match with the functions $h$ and $g$.

We emphasize that our interest is not in variance swap contracts themselves, nor in hedging volatility risk; rather, we use the quantity \eqref{eq:vs_vol_model}, which naturally generalizes the implied volatility in the deterministic volatility case, only to calibrate the implied volatility term structure.

Indeed, under deterministic volatility, the quantity \eqref{eq:vs_price_model} reduces to the integrated variance up to maturity $T$, which was used, for instance, in \citet[Section 4]{Kiesel2009} to calibrate the implied volatility term structure. This calibration procedure corresponds to the first step of our joint calibration. However, unlike \citet{Kiesel2009} and other classical HJM calibration approaches, in the second step we also calibrate the functions 
$g$ and $h$ to achieve an exact fit of the term structure, and in the third step we calibrate the implied volatility smile while keeping the variance swap volatility fixed.

For now, let us introduce a flexible parametrization for the functions $\sigma$.

\subsection{A Nelson-Siegel parametrization for the volatility functions $\sigma$} \label{S:paramcorrel}

Back in 1985, Charles B. Nelson and Andrew F. Siegel proposed in \cite{Nelson1987} the following ``parsimonious'' parametrization for the instantaneous forward rate $r$ with $m \geq 0$ days to maturity
\begin{equation} \label{eq:Nelson_Siegel_parametrization}
    r(m) := \beta_{0} + \beta_{1} e^{-\frac{m}{\tau}} + \beta_{2} \frac{m}{\tau} e^{-\frac{m}{\tau}}, \quad \beta_{0}, \beta_{1}, \beta_{2} \in \mathbb{R}, \quad \tau > 0,
\end{equation}
as an alternative to the previously used polynomial fitting techniques to match the yields of US Treasury bills of maturity $m$, obtained in their formulation by integrating from zero to $m$ the forward rate \eqref{eq:Nelson_Siegel_parametrization} and dividing by $m$. They showed that such expression of the forward rate has a desirable non-explosive asymptotic behavior, and is capable of reproducing humps, $S$-shapes, and
monotonic curves. {We refer to \citet{Sepp2023} as an example of a recent application of the HJM approach with Nelson–Siegel parametrization and stochastic volatility to the interest rate market.}

In our case, the volatility functions $\sigma$ in the HJM model \eqref{eq:HJM_def} aim at capturing with parsimony both the realized covariance structure of a family of rolling futures contracts' daily log returns as well as the VS volatility term structure of the futures contracts whose smiles quote in the market.

Inspired by the Nelson--Siegel parametrization \eqref{eq:Nelson_Siegel_parametrization}, we specify three possible forms for the volatility shape functions $\left( \sigma_i \right)_{i \in \{1, \ldots, N \}}$, which yield respectively three distinct types of factors, named \textit{Level} ($L$), \textit{Slope} ($S$) and \textit{Curvature} ($C$), including
\begin{itemize}
    \item one $L$-factor which has a constant volatility function in order to capture \textit{the long-term volatility level of the curve}
    \begin{equation} \label{eq:l_shape}
    \sigma_{1}(t,T) = \sigma_{L}, \quad 0 \leq t \leq T,
    \end{equation}
    \item $S$-factors having an exponentially increasing volatility as time to maturity decreases to capture \textit{the \cite{samuelson2016proof} effect} 
    \begin{equation} \label{eq:s_shape}
    \sigma_{i+1}(t,T) = \sigma_{S,i} e^{-\frac{T-t}{\tau_{S,i}}}, \quad 0 \leq t \leq T, \quad i \in \{1,\ldots, N_{s}\},
    \end{equation}
    \item $C$-factors aiming at capturing \textit{``humps'' in the volatility term structure}, thereby capturing potential \textit{``anti-Samuelson effect''} for long-term deliveries, i.e.~a decreasing volatility with time to maturity, such that
    \begin{equation} \label{eq:c_shape}
    \sigma_{j+N_s+1}(t,T) = \sigma_{C,j} \frac{T-t}{\tau_{C,j}} e^{-\frac{T-t}{\tau_{C,j}}}, \quad 0 \leq t \leq T, \quad j \in \{1,\ldots, N_{c}\},
    \end{equation}
\end{itemize}
with $N_{s}, N_{c} \in \mathbb{N}$ such that $N := 1 + N_{s} + N_{c}$. We display in Figure \ref{F:l_s_c_volatility_shapes} the three distinct volatility shapes in terms of time to maturity.

\begin{figure}[H]
\begin{center}
    \includegraphics[width=0.8\linewidth]{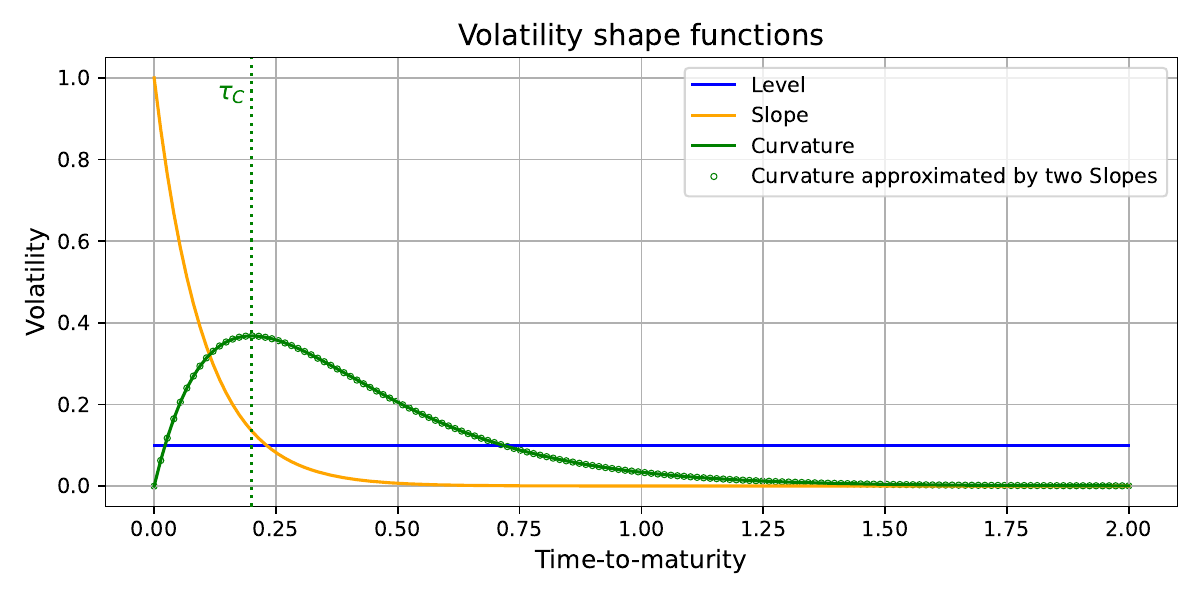}
    \caption{Illustrations of the respective Level (plain blue), Slope (plain orange) and Curvature (plain green) volatility shapes with parameters $\sigma_{L} = 0.1$, $\sigma_{S} = \sigma_{C} = 1$, $\tau_{S} = 0.1$ and $\tau_{C} = 0.2$. We also plot the approximation of curvature by two opposed slope shapes (dotted green) with $\sigma_{S_1} = \sigma_{S_2} = 8.18$, $\tau_{S_1} = 0.212$ and $\tau_{S_2} = 0.188$.}
    \label{F:l_s_c_volatility_shapes}
\end{center}
\end{figure}

The use of both $L$ and $S$ factors is not new, see for example \cite*{Kiesel2009}, \cite{gardini2023heath}. By contrast, we are not aware of previous works trying to calibrate $C$ factors to power markets. Notice in particular that both $S$- and $C$-factors vanish when time to maturity goes to infinity,  i.e.~$T-t \to \infty$, hence only the $L$-factor remains in such regime and therefore captures the long-term volatility level of the energy curve. 

Consequently, such $L$-$S$-$C$ parametrization writes explicitly as
\begin{equation} \label{eq:lsc_deterministic_vol_factors}
    \sigma(t, T)^\top \d W_t = \sigma_{L} \d W_t^{1} + \sum_{i=1}^{N_{s}} \sigma_{S,i} e^{-\frac{T-t}{\tau_{S,i}}} \d W_t^{i+1} + \sum_{j=1}^{N_{c}} \sigma_{C,j} \frac{T-t}{\tau_{C,j}} e^{-\frac{T-t}{\tau_{C,j}}} \d W_t^{j+N_{s}+1}.
\end{equation}
In a similar spirit as \citet{Nelson1987} who noted that their parametrization \eqref{eq:Nelson_Siegel_parametrization} can be easily fitted to market data by least-squares, given a provisional $\tau > 0$, fixing the parameters $(\tau_{S, i})_{i=1, \ldots, N_{s}}$, $(\tau_{C, i})_{j=1, \ldots, N_{c}}$ in \eqref{eq:lsc_deterministic_vol_factors} leads to an efficient calibration problem for the remaining parameters $\sigma_L, (\sigma_{S, i})_{i=1, \ldots, N_{s}}, (\sigma_{C, i})_{j=1, \ldots, N_{c}}$ of the $L$-$S$-$C$ factors, and their correlation matrix $R$ formulated as a linear cone program ensuring that $R$ remains non-negative definite as detailed in Section \ref{s:joint_cov_vs_calib}.

\begin{remark}[Are $L$-$S$-$C$ volatility shapes redundant?]
    On the one hand, notice that taking $\left. \left( \sigma_{S}, \tau_{S} \right)\right|_{\tau_{S} \to \infty}$ in a $S$-factor parametrization \eqref{eq:s_shape} yields a $L$-factor parametrization \eqref{eq:l_shape}. Moreover, by considering two perfectly anti-correlated $S$-factors with identical parameters $\sigma_{S,1}, \sigma_{S,2}$ and distinct mean reversion rates, it is possible to approximate the behavior of a $C$-factor parametrization \eqref{eq:c_shape} as illustrated in green dots in Figure \ref{F:l_s_c_volatility_shapes}. Thus, it is indeed reasonable to restrict oneself to $S$-factors only to identify systematically in the market those distinct level, slope and curvature volatility behaviors but to the price of losing parsimony, as was done for example in \cite{feron2024estimation}. On the other hand, a single curvature shape cannot approximate properly either a level or a slope volatility shape, nor a level shape can approximate neither a slope or a curvature shape.
\end{remark}

\subsection{Joint calibration: overview and snapshots}

The main practical advantage of our model is the possibility to decouple the joint calibration problem into three independent optimization problems to be solved consecutively, without the need to modify the parameters calibrated in previous steps. 
\begin{enumerate}
    \item[1)] A combined historical and implied calibration of the parameters ($\sigma, R$) to capture with parsimony i.e.~with a minimum number of deterministic risk factors, the historical correlation of rolling futures contracts' daily log returns as well as the overall term structure of the implied VS volatility deduced from the market quotes of vanilla options.
    \item[2)] An exact calibration correction of the term structure of the implied VS volatility not captured in step 1), thanks to the functions $g$ and $h$.
    \item[3)] A calibration of the entire shape of the smile using the parameters of the stochastic variance process $V$ in \eqref{eq:V_def}, that is $(c,x)=(c_i,x_i)_{i \in \{1,\ldots, M\}}$ and $\tilde \rho = (\tilde \rho_i)_{i \in \{1,\ldots, N\}}$.
\end{enumerate}

The flow chart in Figure \ref{F:calibration_methodology} illustrates the successive calibration procedure with the data used in each step as well as the resulting calibrated model parameters.

\begin{figure}[H]
\centering
\begin{tikzpicture}[
    font=\footnotesize, 
    box/.style = {rectangle, draw, rounded corners, minimum height=2em, minimum width=9em, align=center, fill=blue!10}, 
    arrow/.style = {thick,->,>=Stealth},
    label/.style = {align=right},
    node distance=1cm 
]

\node (box1) [box] {Step 1: Joint historical and implied VS volatilities term structure calibration};

\node (data1) [left=1cm of box1, label] {Data: Historical covariances \\
and variance swap volatilities};
\draw [arrow] (data1.east) -- (box1.west);

\node (box2) [box, below=0.8cm of box1] {Step 2: VS volatilities term structure fit correction};

\draw [arrow] (box1.south) -- node[right, align=center, xshift=0.2em] {Calibrated $(\sigma, R)$} (box2.north);

\node (data2) [left=1cm of box2, label] {Data: Variance swap volatilities};
\draw [arrow] (data2.east) -- (box2.west);

\node (box3) [box, below=0.8cm of box2] {Step 3: Implied smile calibration};

\draw [arrow] (box2.south) -- node[right, xshift=0.2em] {Calibrated $(g, h)$} (box3.north);

\node (data3) [left=1cm of box3, label] {Data: Call and put options};
\draw [arrow] (data3.east) -- (box3.west);

\node (params3) [below=0.8cm of box3, align=center] {Output: Calibrated $(\sigma, R, g, h, c, x, \tilde{\rho})$};
\draw [arrow] (box3.south) -- node[right, align=center, xshift=0.2em] {Calibrated $(c, x, \tilde{\rho})$} (params3.north);

\end{tikzpicture}
\caption{Overview of the calibration methodology in three iterative steps, including input data and calibrated parameters obtained at each step.}
\label{F:calibration_methodology}
\end{figure}
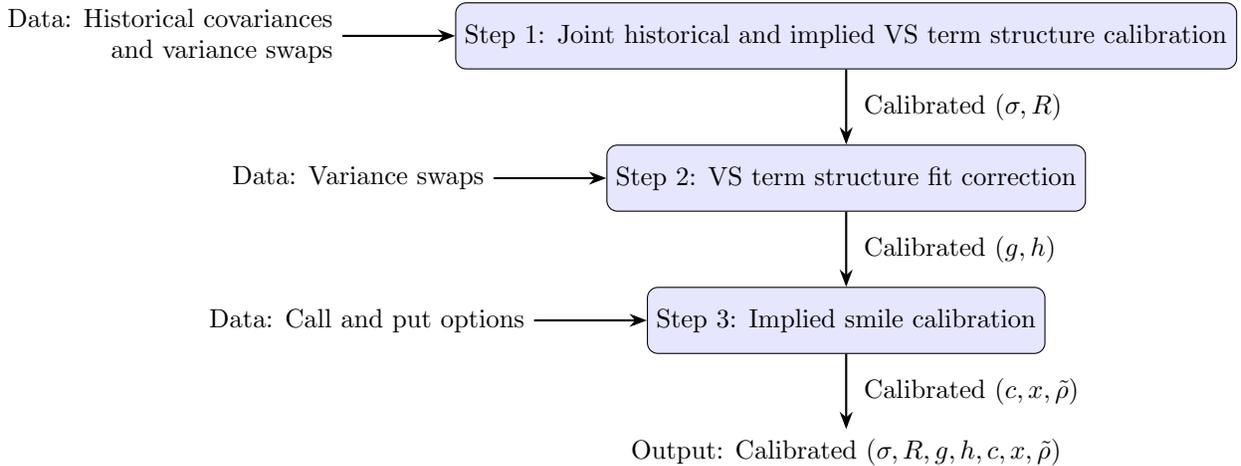

In Figure~\ref{fig:snapshot}, we provide a calibration snapshot that summarizes the results of all three calibration steps   for the case of the German power market as of the $1^{\text{st}}$ of July 2024. For this date, $N = 5$ $L$-$S$-$C$ risk factors with $N_s = 3$, $N_c = 1$, and $M = 3$ stochastic volatility factors were used, which yields $19$\footnote{Our choice of 19 parameters (5 factors) remains conservative compared to  the literature on historical calibration in power markets: \cite{feron2024estimation} consider 5 $S$-factors (20 parameters), \cite{gardini2023heath} use 10 factors with $100$ parameters to fit the realized volatility of $144$ futures from six different markets, and a PCA study in \cite{koekebakker2005forward} supports a similar order of factors.
} (resp.~$11$) calibrated parameters for step 1 (resp.~step 3).  Our model achieves an excellent fit of the historical volatilities and correlations in step 1, a perfect fit of the VS term structures in step 2 and a remarkable fit of the whole implied volatility surface in step 3, see also Figure~\ref{fig:iv_smiles_results}. More detailed specification of model factors and calibrated parameters will be described further in Section \ref{S:calib_res}.

\begin{figure}[H]
    \centering

    \adjustbox{trim={0cm 0cm 0cm 0cm},clip}{%
        \includegraphics[width=0.33\textwidth]{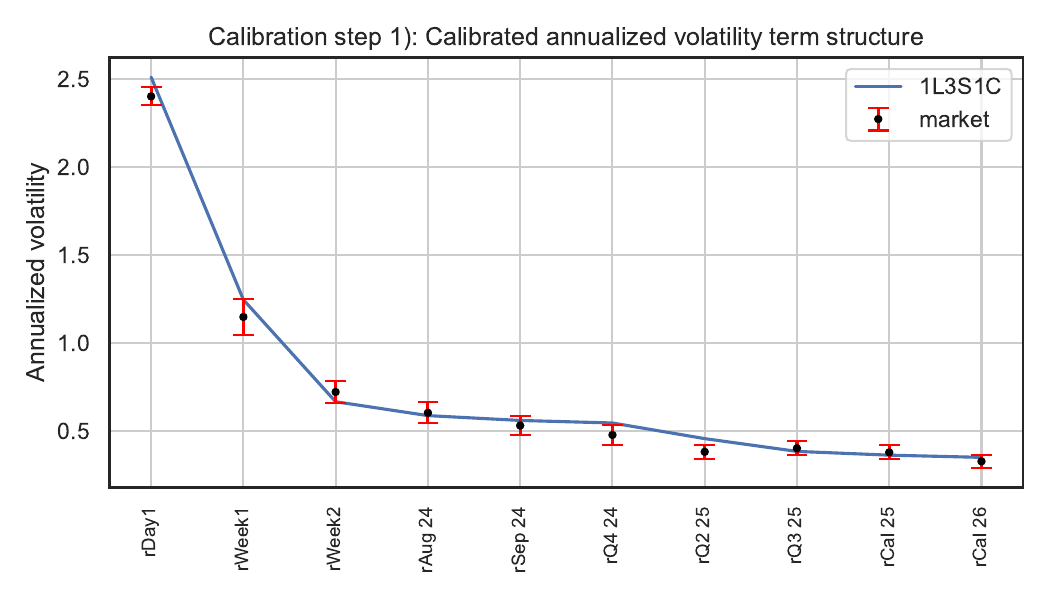}
    }
    \hspace{-0.3cm} 
    \adjustbox{trim={0cm 0cm 0cm 0cm},clip}{%
        \includegraphics[width=0.33\textwidth]{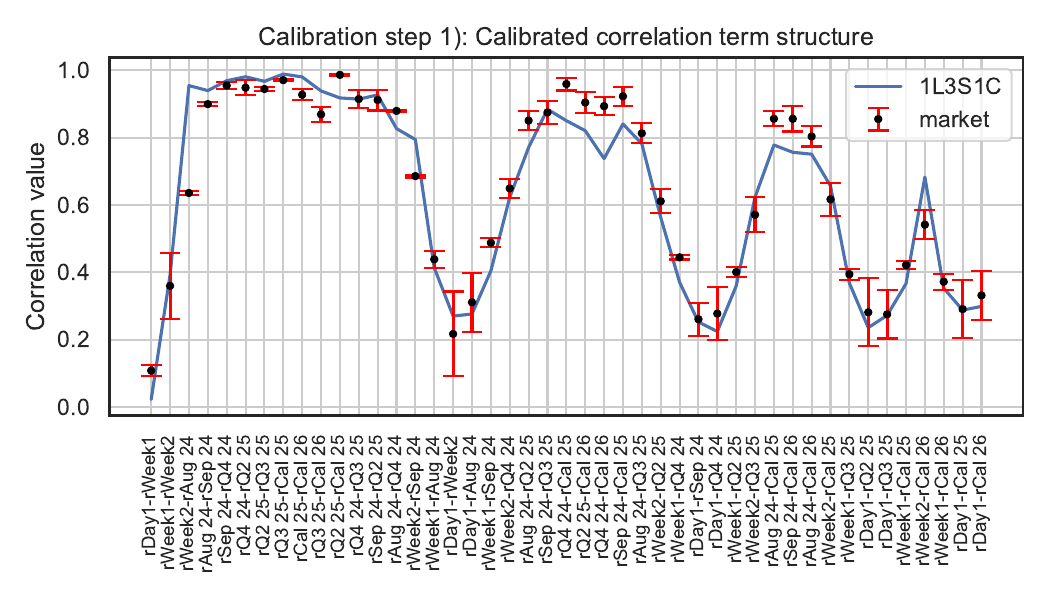}
    }
    \hspace{-0.3cm} 
    \adjustbox{trim={0cm 0cm 0cm 0cm},clip}{%
        \includegraphics[width=0.33\textwidth]{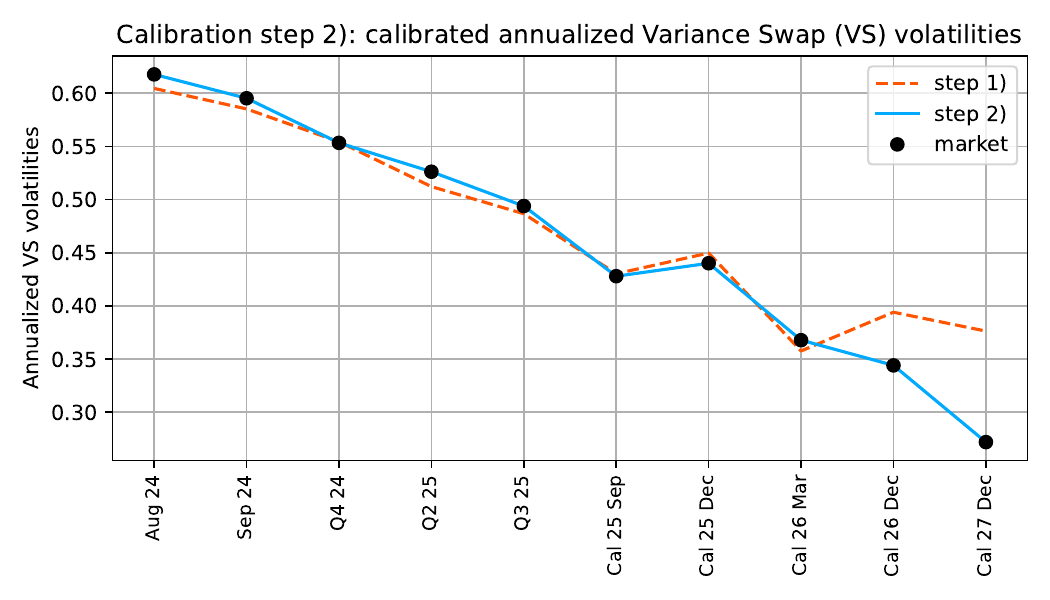}
    }

    \vspace{-0.1cm} 

    \adjustbox{trim={0cm 0cm 0cm 0cm},clip}{%
        \includegraphics[width=0.33\textwidth]{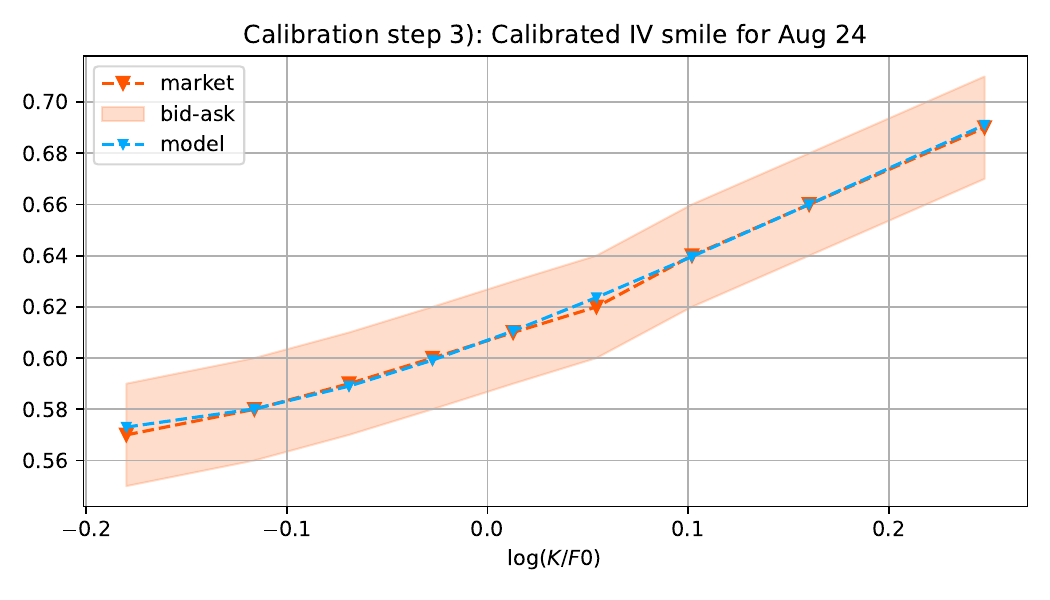}
    }
    \hspace{-0.3cm} 
    \adjustbox{trim={0cm 0cm 0cm 0cm},clip}{%
        \includegraphics[width=0.33\textwidth]{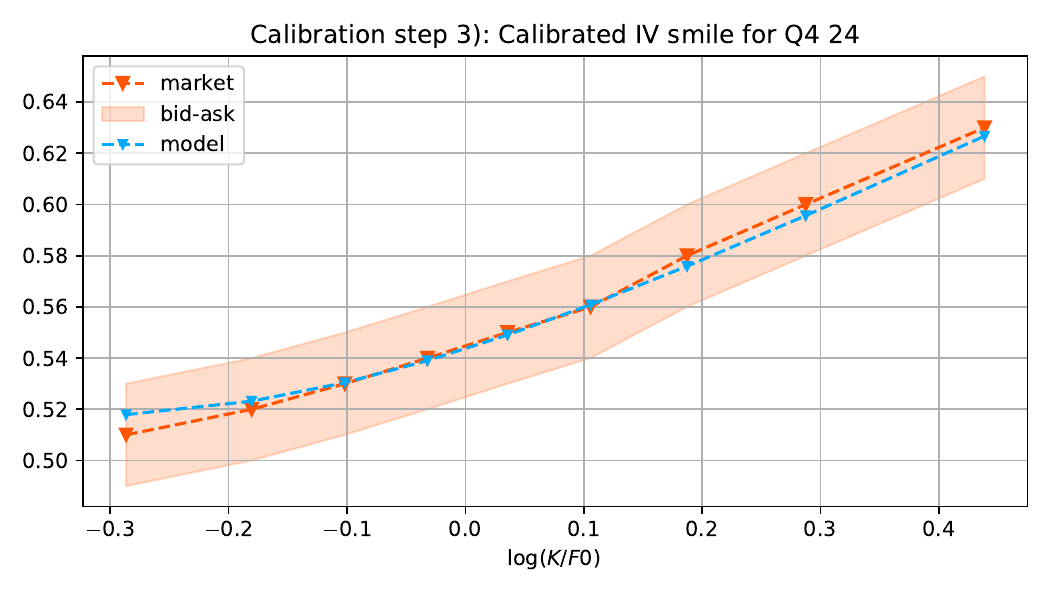}
    }
    \hspace{-0.3cm} 
    \adjustbox{trim={0cm 0cm 0cm 0cm},clip}{%
        \includegraphics[width=0.33\textwidth]{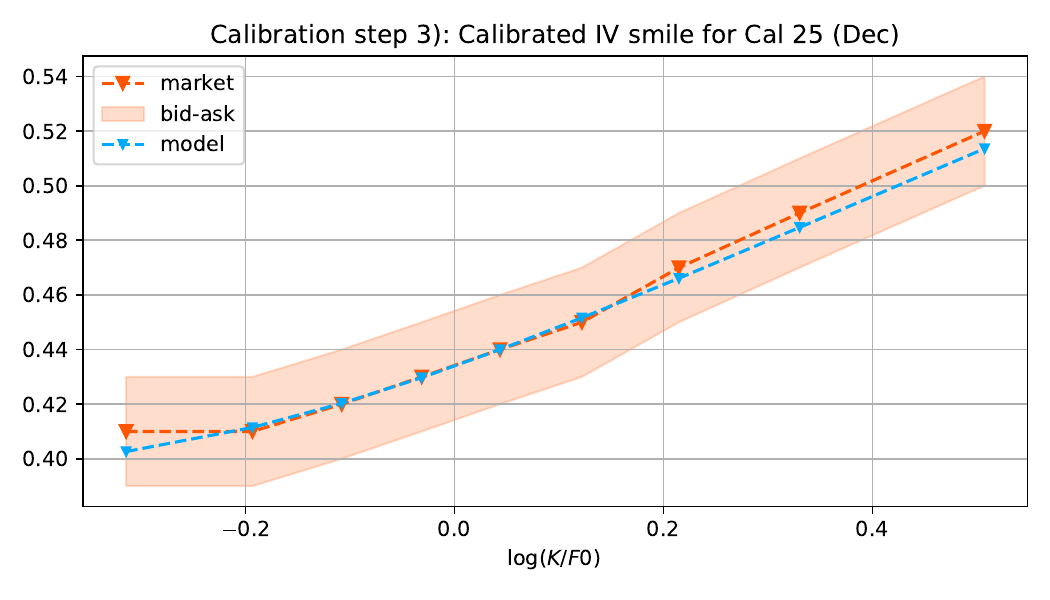}
    }

    \caption{Snapshot of the calibration results on the German power market on the $1^{\text{st}}$ July 2024: fit of the realized volatility term structure (upper left) and correlations (upper middle) of rolling futures in step 1 over one year and a half of historical data; fit of the VS volatility (upper right) term structure in step 1 (orange) and the corrected fit obtained in step 2 (blue); fit of the smile slices in step 3 of the respective futures contracts delivering electricity during August 24 (lower left), the forth Quarter of 2024 (lower middle) and the whole calendar year 2025 (with maturity December 2024; lower right).}
    \label{fig:snapshot}
\end{figure}

\section{Market data} \label{s:market_data}

\subsection{Futures contracts' quotes} \label{ss:typical_futures}

The market data used to calibrate the futures contracts' instantaneous correlation structure are the realized closing market quotes of unitary base-load\footnote{The other main traded profile, although less liquid, is the peak-load delivering electricity during working hours of the week. We chose to focus on the base-load data for illustrating our approach.} futures contracts i.e.~typically the price per MWh for the delivery of electricity every hour of a given delivery period\footnote{A precise specification of the contract can be found at \href{https://www.eex.com/en/markets/power/power-futures}{https://www.eex.com/en/markets/power/power-futures} for the German power market and at \href{https://www.theice.com/products/27996665/Dutch-TTF-Gas-Futures}{https://www.theice.com/products/27996665/Dutch-TTF-Gas-Futures} for futures traded on the TTF gas market.}. We present here several types of liquid futures contracts on power markets, these notations that will be used further in Section \ref{s:implied_calib}.
\begin{itemize}
    \item Day-ahead contracts, denoted by DayX with $X \in \mathbb{N}^{*}$, typically the day-ahead contract (DA1) whose price is set by auction for the delivery of electricity for the following day, with $T_{s}$ set to the first hour of delivery the next day and $T_{e}$ equals $T_{s}$ plus one day.
    \item Week-ends, denoted by WeekX with $X \in \mathbb{N}^{*}$, working days or week-ahead contracts, and also balance of month contracts which quotes the price of delivery of electricity such that, in theses cases, $T_{s}$ is set accordingly depending on the observation date, and $T_e$ equals $T_{s}$ plus, respectively, two days, five days and the number of days until the end of the running month.
    \item Monthly contracts denoted by Mon YY: delivering during the fixed month ``Mon'' of the year ``20YY''. For these contracts, $T_s$ an $T_e$ are equal to the first and the last day of the delivery month.
    \item Quarterly contracts denoted by QX YY: delivery period is the ``X''-th quarter of the year ``20YY'' with X $\in \{1, 2, 3, 4\}$. Here, the first quarter corresponds to the first three months of the year, the second one corresponds to the forth-sixth months, and so on. In this case, $T_s$ is equal to the first date of the first month of the quarter, and $T_e$ is the last day of the last quarter month.
    \item Calendar contracts denoted by Cal YY: contracts delivering during the whole year ``20YY''. $T_s$ and $T_e$ correspond to the first and the last days of the year.
\end{itemize}

Similarly, we use the notation rContract to refer to the rolling future contract with same time to delivery and delivery duration as Contract.

\subsubsection{Stripping optimization to construct rolling futures contracts}\label{S:hist_md}

Recall that, except for the first day-ahead contract, a rolling futures contract with time to delivery $T_{s}$ and delivery duration $T_{e} - T_{s}$ is typically not exchanged in standard markets\footnote{Although at some specific observation dates $t_{0}$, a rolling future with time to delivery $\delta>0$ and delivery duration $\Delta_{T}>0$ may coincide with the absolute futures contract delivering on $[T_{s}, T_{e}]$ if
\begin{equation}
    \begin{cases}
        \delta = T_{s} - t_{0} \\
        \Delta_{T} = T_{e} - T_{s}
    \end{cases}.
\end{equation}
}, and therefore its quote cannot be observed directly.
Yet, it is possible to construct the quotes of such $P_{\mathrm{hist}} \in \mathbb{N}^{*}$ rolling futures contracts from the ones of the absolute futures contracts observed on the market via a stripped forward curve defined as follows.

\begin{definition}[Stripped forward curve] \label{def:stripped_curve}
    Fix an observation date $t_{0} \geq 0$ and consider a family of $P' \in \mathbb{N}^{*}$ absolute futures contracts' quotes 
    $$
    \left( F_{t_{0}}^{\mathrm{mkt}} (T_{s}^{j}, T_{e}^{j})\right)_{j \in \{ 1, \ldots, P' \}},
    $$
    and denote $\bar T := \max_{j \in \{ 1, \ldots, P' \}} T_{e}^{j}$. Then, a \textit{stripped forward curve}, at observation date $t_{0}$, is a smooth\footnote{$f^{\mathrm{mkt}}$ can be uniquely defined when some additional regularizing criterion is imposed.} price curve $\left( f^{\mathrm{mkt}}(t_{0},T) \right)_{T \in [t_{0}, \bar T]}$ such that the absence of superposition arbitrage opportunities is verified, i.e.~such that the equality constraints
    \begin{equation}\label{eq:no_superposition_arbitrage_constraints}
        F_{t_{0}}^{\mathrm{mkt}} (T_{s}^{j}, T_{e}^{j}) = \frac{1}{T_{e}^{j}-T_{s}^{j}} \int_{T_{s}^{j}}^{T_{e}^{j}} f^{\mathrm{mkt}}(t_{0},T) \d T, \quad j \in \{ 1, \ldots, P' \} 
    \end{equation}
    are all satisfied. In particular, the quote of any rolling contract at date $t_{0}$ can be obtained from the stripped forward curve $f^{\mathrm{mkt}}$ by the following formula
    \begin{equation} \label{eq:historical_rolling_future_contract}
        F_{t_{0}}^{\mathrm{mkt}} (t_{0}+T_{s}, t_{0}+T_{e}) := \frac{1}{T_{e}-T_{s}} \int_{t_{0}+T_{s}}^{t_{0}+T_{e}} f^{\mathrm{mkt}}(t_{0},T) \d T,
    \end{equation}
    and we define its realized, or historical, log returns over a period of $\tau_{d}$ days at date $t_{0}$ by
    \begin{equation} \label{eq:historical_rolling_future_contract_returns}
        r_{t_{0}}^{\mathrm{mkt}}(\tau_{d}) := \log \frac{F_{t_{0}}^{\mathrm{mkt}} (t_{0}+T_{s}, t_{0}+T_{e})}{F_{t_{0}-\tau_{d}}^{\mathrm{mkt}} (t_{0}+T_{s}, t_{0}+T_{e})},
    \end{equation}
    where the delivery period is held fixed, so that their respective sample Fourier spectra (resp. auto-correlations) don't display any significant frequency peak (resp. lag) i.e.~absence of weekly effect from working days to week-ends (resp. independent increments assumption) hold valid with such log returns construction, see for example \cite[Figure 4]{cartea2005pricing} or \cite[Figure 3]{gardini2023heath}.
\end{definition}

In order to formulate the stripping optimization problem, we use the same notations as in Definition \ref{def:stripped_curve}, and we set the daily time grids
\begin{align}
    \mathcal{T}_{n}(t_{0}) & := \left\{ t_{0} = T_{0} < T_{1} < \ldots < T_{n} = \bar{T} \right\}, \\
    \mathcal{T}_{n^{j}}^{j} & := \left\{ T_{s}^{j} = T_{0}^{j} < T_{1}^{j} < \ldots < T_{n^{j}}^{j} = T_{e}^{j} \right\}, \quad j \in \{ 1, \ldots, P' \},
\end{align}
where $n \in \mathbb{N}^{*}$ (resp. $n^{j} \in \mathbb{N}^{*}, \; j \in \{ 1, \ldots, P' \}$) denotes the number of days in the time interval $[t_{0}, \bar{T}]$ (resp. in $[T_{s}^{j}, T_{e}^{j}], \; j \in \{ 1, \ldots, P' \}$).

As detailed in Definition \ref{def:stripped_curve}, the stripping optimization problem aims to build a stripped daily forward curve from a family of absolute forwards' quotes $F^{\mathrm{mkt}}(t_{0}) := \left( F_{t_{0}}^{\mathrm{mkt}} (T_{s}^{j}, T_{e}^{j})\right)_{j \in \{ 1, \ldots, P' \}}$ observed on the markets such that 
\begin{itemize}
    \item[(i)] the absence of superposition arbitrage opportunities \eqref{eq:no_superposition_arbitrage_constraints} is satisfied;

    \item[(ii)] a smoothing criterion is applied ensuring the well-posedness of the algorithm: the sum of daily increments of the curve squared is minimized.
\end{itemize}
Then, the daily stripped forward curve at observation date $t_{0}$ is given by $\left( f^{\mathrm{mkt}}(t_{0},T_{i}) \right)_{i \in \{0, \ldots, n-1 \}} := \left( \hat{s}(t_{0}, T_{i}) \right)_{i \in \{0, \ldots, n-1 \}}$, where $\hat{s}$ is defined as the unique minimizer of
\begin{equation} \label{eq:optimization_problem_stripping}
    \min_{s \in A^{n}(F^{\mathrm{mkt}}(t_{0}))} \sum_{i = 1}^{n-1} \left( s_{i} - s_{i-1} \right)^2,
\end{equation}
with $A^{n}(F^{\mathrm{mkt}}(t_{0}))$ denoting the set of admissible curves defined by
\begin{equation}
    A^{n}(F^{\mathrm{mkt}}(t_{0})) := \left\{ \left( s_{l} \right)_{l \in \{0, \ldots, n-1 \}} \in \mathbb{R}^{n} \; \Bigg| \; F_{t_{0}}^{\mathrm{mkt}} (T_{s}^{j}, T_{e}^{j}) = \sum_{l=0}^{n^{j}-1} w_{j,T_{l}} s_{l}, \; j \in \{ 1, \ldots, P' \} \right\},
\end{equation}
and where, for $j \in \{ 1, \ldots, P' \}$, $\left( w_{j,T_{l}} \right)_{l \in \{ 0, \ldots, n^{j}-1 \}} \in [0,1]^{n^{j}}$ denotes the family of daily weights associated to the forward $F_{t_{0}}^{\mathrm{mkt}} (T_{s}^{j}, T_{e}^{j})$ depending on \textit{the profiling of the underlying delivery}, e.g.~$\frac{1}{365.25}$ for a calendar contract in base profile.

The optimization problem \eqref{eq:optimization_problem_stripping} can be reformulated as a Quadratic Program with linear equality constraints in the sense of \cite[Chapter~4, 4.4]{boyd2004convex} and admits a unique solution by strict convexity. {Market data should be without absence of arbitrage to ensure the existence of such solution. Some day-ahead futures contracts may display negative quotes, hence $s_{l} \in \mathbb{R}, \; l \in \{0, \ldots, n-1 \}$ in general, yet it is possible to look for a non-negative stripped curve by adding the constraints $s_{l} \in \mathbb{R}^{+}, \; l \in \{0, \ldots, n-1 \}$ if all the quoting market futures contracts are non-negative.} The interested reader may look at \cite[Chapter 7]{benth2008modeling} for additional insights on stripping including seasonal effects.


\subsubsection{Estimated covariance of rolling futures' log returns from past historical data} 

Given a family of stripped forward curves $\left(\left( f^{\mathrm{mkt}}(t_{h},T_{i}) \right)_{t_{h} \leq T_{i} \leq \bar{T}} \right)_{h \in \left\{ 0, \ldots, H \right\}}$ constructed at various observation dates $\left( t_{h} \right)_{h \in \left\{ 0, \ldots, H \right\}}$ by solving the stripping optimization problem \eqref{eq:optimization_problem_stripping}, we can reconstitute any family of $P_{\mathrm{hist}} \in \mathbb{N}^{*}$ rolling forward contracts $\left( F_{t_{h}}^{\mathrm{mkt}} (t_{h}+T_{s}^{k}, t_{h}+T_{e}^{k}) \right)_{k \in \{ 1, \ldots, P_{\mathrm{hist}}\}}$ satisfying $t_{h}+T_{e}^{k} \leq \bar{T}, \; k \in \{ 1, \ldots, P_{\mathrm{hist}}\}$ for any $h \in \left\{ 0, \ldots, H \right\}$ using \eqref{eq:historical_rolling_future_contract}, and then estimate a family of time series of their respective realized daily log returns $\left( \left( r_{t_{h}}^{\mathrm{mkt}, k} (\tau_{d}) \right)_{h \in \left\{ 0, \ldots, H \right\}} \right)_{k \in \{ 1, \ldots, P_{\mathrm{hist}}\}}$ by setting $\tau_{d} = 1$ day using \eqref{eq:historical_rolling_future_contract_returns}.\\

Without any surprise, the rolling futures' log returns display heavy tails across all deliveries, and particularly for short-term ones. Since the futures' log returns' historical covariances will be used to calibrate a log normal distribution under the KV approximation \eqref{eq:kv_approx} in the first calibration step, i.e.~calibrate $\left( \sigma, R\right)$, with $\sigma$ specified in \eqref{eq:lsc_deterministic_vol_factors}, we first filter out futures contracts' log returns' large values away from three standard deviations (computed on the whole time-series) and setting them equal to three standard deviations, in the same spirit as in \cite{cartea2005pricing}. 

Then, we construct a biased averaged covariance estimator of such daily log returns in order to assign higher weights to newer data points and diminishing weights to older data, and we extract some confidence bounds when weighting differently the past observations. We postpone to Appendix \ref{ss:covariance_estimation} the detailed estimation procedure. 

Finally, we extract the futures contracts' volatilities by taking the square root of the diagonal of the resulting covariance estimator, and their correlations by normalizing their respective covariances by the corresponding volatilities. Figure \ref{F:estimated_correlation_term_structure} displays the resulting realized correlation and volatility term structures on the German power market from January $1^{\text{st}}$ 2023 to July $1^{\text{st}}$ 2024. 


\begin{figure}[H]
    \centering
    \begin{minipage}{0.45\textwidth}
        \centering
        \includegraphics[width=\textwidth, trim=10 0 100 10, clip]{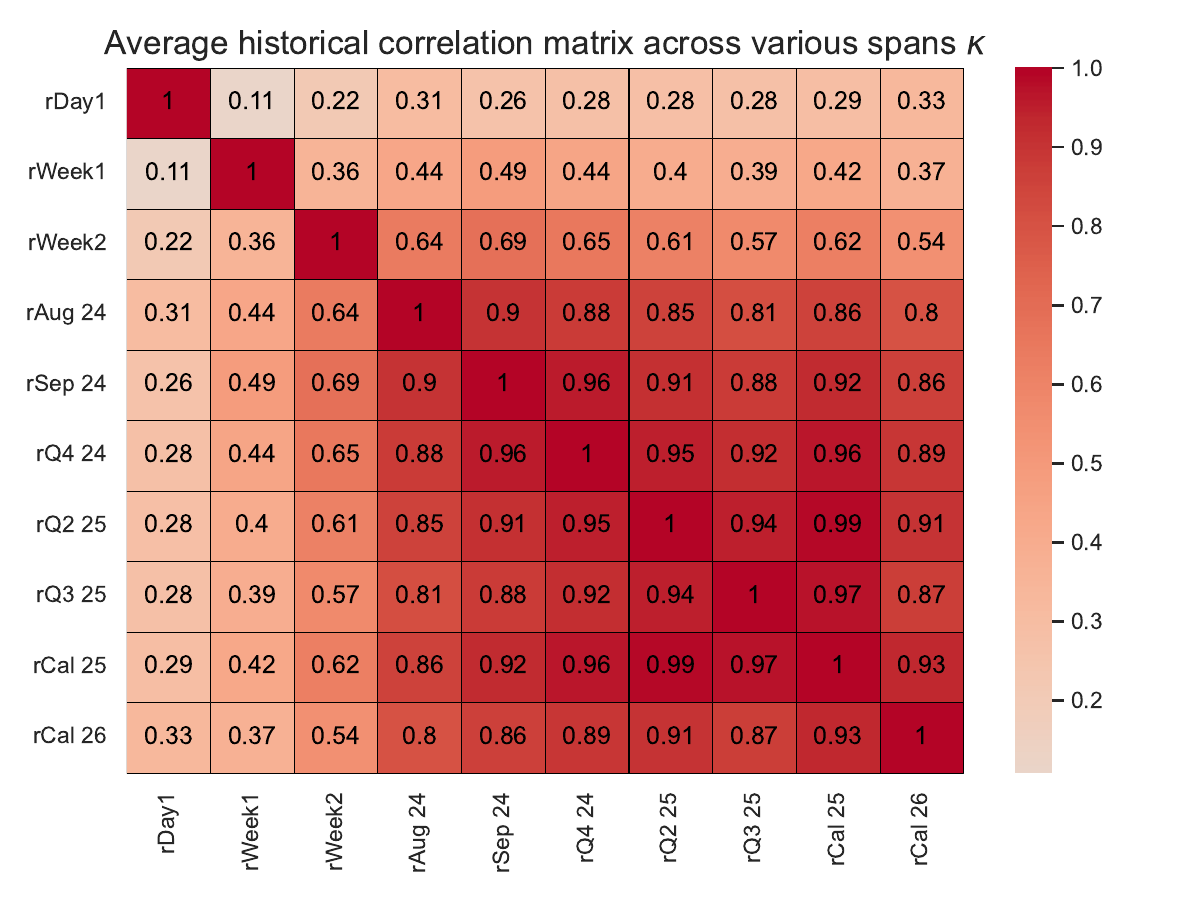} 
    \end{minipage}
    \hspace{0.1cm} 
    \begin{minipage}{0.45\textwidth}
        \centering
        \includegraphics[width=\textwidth, trim=10 0 10 10, clip]{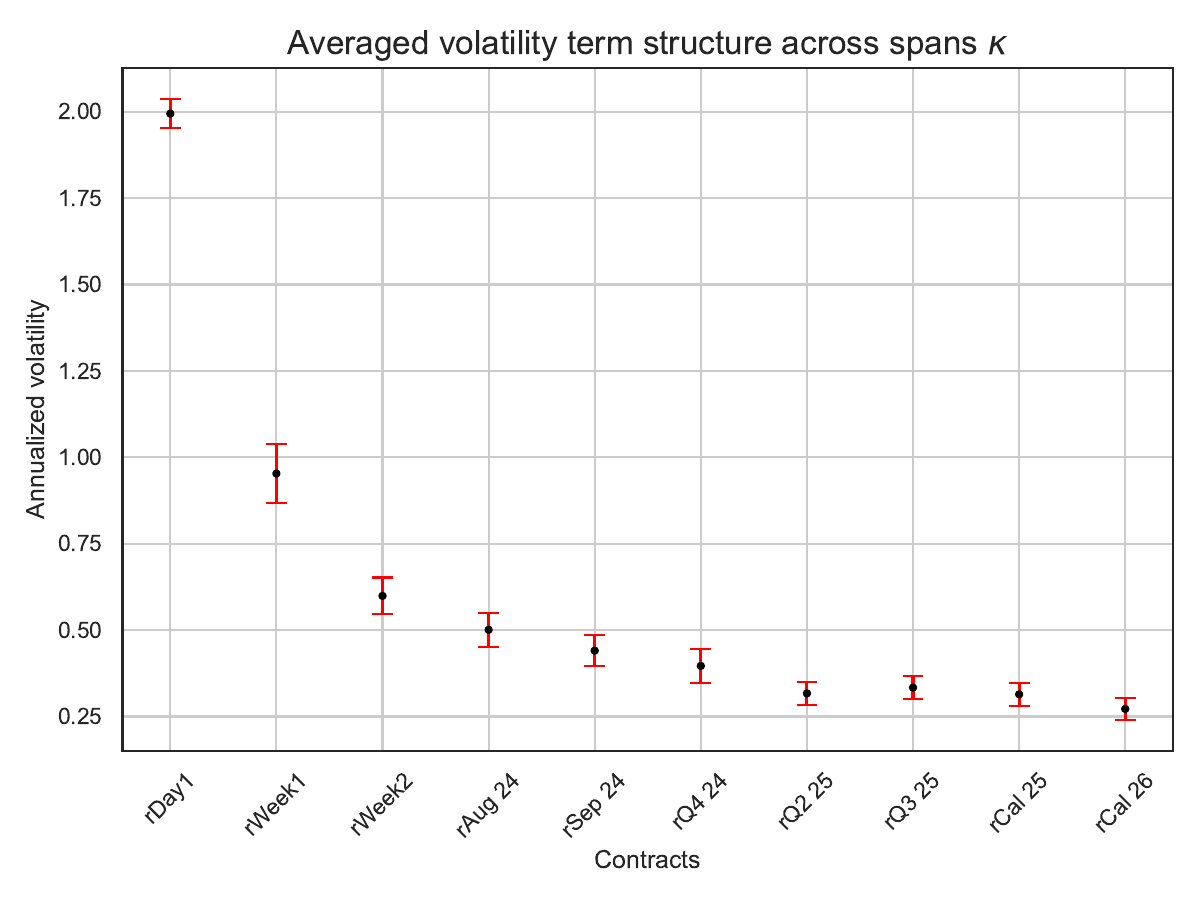} 
    \end{minipage}

    \caption{Left: averaged historical correlation matrix of rolling futures contracts' daily log returns on the German market; right: averaged historical volatility term structure of rolling futures contracts' daily log returns. The covariance estimator and confidence intervals are derived by averaging estimated covariances over a family of span parameters $\kappa$ going from $30$ to $365$ days to put more weights on recent observations, see Appendix \ref{ss:covariance_estimation}.}
    \label{F:estimated_correlation_term_structure}
\end{figure}

\subsection{Options data} \label{ss:option_market_description}

\subsubsection{Typical options quoted in power markets}

We are interested in calibrating a diverse set of volatility smiles corresponding to different underlying futures contracts. The most liquid options are typically written on three types of underlying futures contracts: the first few months, quarters and calendar contracts. These liquid options are linear combinations of vanilla options like call spreads, put spreads, butterflies or fences. The descriptions of vanilla options can be found  on the ice website (ICE).\footnote{https://www.ice.com/products/65898946/German-Power-Financial-Base-Options}

Vanilla options on monthly or quaterly futures have a single maturity: they expire several business day before the first day of the delivery period. However, vanilla options on calendar futures can have different maturities. The December maturity is generally the most liquid. The other maturities such as September, June and March --- referred to as \textit{early expiries} --- are also quoted. To avoid confusion, the maturity month (of the year preceding the delivery year) is specified in the option name. For example, ``Cal 26 Mar'' corresponds to the smile with maturity at the end of March 2025. The precise maturities can be found on the site of the European Energy Exchange (EEX)\footnote{https://www.eex.com/en/market-data/power/equity-styled-options}.

For our numerical experiments, we use the following ten implied volatility smiles observed on the on the $1^{st}$ July 2024 from the German power option market corresponding to eight different underlying futures contracts: Aug 24, Sep 24, Q4 24, Q2 25, Q3 25, Cal 25 Sep, Cal 25 Dec, Cal 26 Mar, Cal 26 Dec, Cal 27 Dec. Note that in energy markets, a large part of the volume is traded over the counter (OTC), unlike in equity markets, where options are typically traded via exchanges in a centralized manner. Consequently, there exist many possible sources of option price data. For our study, we use implied volatility data provided by the broker ICAP\footnote{https://icap.com/energy-commodities}.

\subsubsection{Variance swap volatilities deduced from option prices}\label{sect:vs_quote_extr}

Variance Swap (VS) volatility can be interpreted as an overall smile level and will be used for the first two steps of calibration, recall Figure \ref{F:calibration_methodology}, and can be deduced from VS contracts, recall Section \ref{ss:model_vs_swap_and_vol}. As already mentioned, such contracts are not quoted on power markets. However, we show how variance swap market prices can be extracted from vanilla option prices.

Indeed, these prices can be expressed as prices of European options with logarithmic payoff, further referred to as \textit{log-contracts}, that can be extracted from the implied volatility smile by the \cite{carr1998towards} formula:
\begin{equation}\label{eq:vs_carr_formula}
    \mathrm{VS}_T = -2\E\log \frac{{F}_T}{F_0} = 2\int_0^{F_0}\dfrac{P(T, K)}{K^2}\,\d K + 2\int_{F_0}^\infty\dfrac{C(T, K)}{K^2}\,\d K,
\end{equation}
where $P(T, K)$ and $C(T, K)$ denote respectively the prices of put and call vanilla options.

\paragraph{Notations.} The calibration will be performed on call and put options written on the futures contracts $F^{i} := F_{.}(T_s^i,\, T_e^i)$ for $i \in \{1, \ldots, P_{\mathrm{imp}}\}$, where the $i$-th futures contract is identified with its delivery period $[T_s^i,\, T_e^i]$. For $i \in \{1, \ldots, P_{\mathrm{imp}} \}$, we denote by $N_i$ the number of (sorted) maturities $(T_{j}^i)_{j \in \{1,...,N_i\}}$ i.e.~the number of smiles, associated to the underlying $F^i$.  The implied market data contains a family of European call and put option prices
$$
\left\{\mathrm{Call}^{\mathrm{mkt}, i}(T_{j}^i, K)\right\}_{K \in \mathcal{K}_j^i}, \quad \left\{\mathrm{Put}^{\mathrm{mkt}, i}(T_{j}^i, K)\right\}_{K \in \mathcal{K}_j^i} \quad j \in \{ 1, \ldots, N_i \}, \quad i \in \{1, \ldots, P_{\mathrm{imp}}\},
$$
where {$P_{\mathrm{imp}} \in \mathbb{N}^{*}$ denotes the number of underlyings,} $\mathcal{K}_j^i$ denotes the {set of} strikes corresponding to  maturity $T_j^i$ , and $\mathrm{Call}^{\mathrm{mkt}, i}\bigl(T_{j}^i, K)$ (resp. $\mathrm{Put}^{\mathrm{mkt}, i}\bigl(T_{j}^i, K)$) denotes the price of the European call option with maturity $T_{j}^i$, strike $K$, and underlying $F^i$. We denote by $\sigma_{\mathrm{BS}}^{\mathrm{mkt}, i}(T_{j}^i, K)$ the corresponding Black implied volatility. 

The market price of the variance swap with underlying $F^i$ and maturity $T_j^i$ is deduced from the call and put options by \eqref{eq:vs_carr_formula} as follows 
\begin{equation}\label{eq:vs_market}
    \mathrm{VS}_{T_j^i}^{\mathrm{mkt}, i} := 2\int_0^{F_0^i}\dfrac{\mathrm{Put}^{\mathrm{mkt}, i}(T_j^i, K)}{K^2}\,\d K + 2\int_{F_0^i}^\infty\dfrac{\mathrm{Call}^{\mathrm{mkt}, i}(T_j^i, K)}{K^2}\,\d K,
\end{equation}
where $F_0^i$ stands for the initial price of the $i$-th underlying futures contract. {Then the VS volatility is extracted by
\begin{equation} \label{eq:extract_vs_vol_from_vs_swap}
    \sigma_{\mathrm{VS}, T_j^i}^{\mathrm{mkt}, i} = \sqrt{\dfrac{1}{T} \mathrm{VS}_{T_j^i}^{\mathrm{mkt}, i}}.
\end{equation}}

Typically, the number of strikes per smile (no more than $9$) present in market data is not enough to apply the discretized formula \eqref{eq:vs_market} directly. Thus, a consistent smiles arbitrage-free extrapolation model is needed. An example of such a model for surface parametrization is given by the SSVI parametrization of \cite{GatheralSSVI}. However, in our case, multiple surfaces are present at the same time, and a more general model is needed. Thus, we propose a multi-contract SSVI parametrization of the implied volatility surfaces for all contracts which allows us to interpolate and extrapolate the smiles and compute numerically the integrals in \eqref{eq:vs_market}.

\paragraph{Multi-contract SSVI parametrization.}

Initially introduced by \cite{GatheralSSVI},  the Surface SVI (SSVI) parametrization specifies the total implied variance $w(T, k) = \sigma_{\mathrm{BS}}^2(T, k) T$, with corresponding time to maturity $T$ and log-moneyness $k := \log\frac{K}{F_0}$. Namely, the implied total variance of the $i$-th futures contract is given by
\begin{equation} \label{eq:implied_variance_SSVI_parametrization}
    w_i(T, k) = \dfrac{\theta_{T, i}}{2}\left(1 + \rho_{i}\phi(\theta_{T, i})k + \sqrt{(\phi(\theta_{T, i})k + \rho_{i})^2 + (1 - \rho_{i}^2)} \right),
\end{equation}
where
\begin{itemize}
    \item $\theta_{T, i} := \sigma_{\mathrm{BS}, i}^2(T, 0) T$ is the at-the-money (ATM) total variance of the underlying $F^i$; 
    \item $\rho_{i} \in [-1,\,1]$ is a skew parameter interpreted as a spot-vol correlation for the contract $F^i$;
    \item $\phi$ is a parametric function satisfying a list of conditions to guarantee the absence of static arbitrage.
\end{itemize}

Due to the presence of multiple surfaces to be parametrized together, we specify individual correlation parameters $\rho_i \in [-1, \, 1]$ for each contract $F^i$, and a common function $\phi$ since {we assume the same nature of volatility across all the underlying futures contracts. This choice appears consistent with observed market data and allows us to simultaneously parametrize all the volatility surfaces using a small number of parameters.}
We set
$$
\phi(\theta) := \dfrac{\eta}{\theta^\gamma(1 + \theta)^{1 - \gamma}}, \quad \gamma \in (0,\, 0.5],
$$ 
and 
$$
\eta \left( 1 + \max\limits_{i \in \{1, \ldots, P_{\mathrm{imp}}\}}|\rho_i| \right) \leq 2,
$$
which is consistent with the empirically-observed term structure of the volatility skew and such that \eqref{eq:implied_variance_SSVI_parametrization} generates {static-}arbitrage-free surfaces for the respective futures contracts \cite[Remark 4.4]{GatheralSSVI}.

A calibrated multi-contract SSVI model allows itself for an efficient smile interpolation and extrapolation for a fixed maturity $T$ as well as for a surface extrapolation for the contracts being in the calibration set. However, in contrast to the HJM stochastic volatility model \eqref{eq:HJM_def}, it cannot generate a volatility surface for a futures contract $F^{\tilde i}$ not present in the calibration set unless its ATM volatility $\sigma_{\mathrm{BS}, \tilde i}(T, 0)$ and correlation parameter $\rho_{\tilde i}$ are provided. Thus, it can interpolate individual volatility surfaces, but not the volatility hypercube.
That is why we limit our use of this parametrization to the evaluation of the implied variance swap prices \eqref{eq:vs_carr_formula}. In addition, the calibrated values $(\rho_i)_{i \in \{1, \ldots, P_{\mathrm{imp}}\}}$ can be further used to obtain an initial guess for the spot-vol correlations in the smile calibration procedure.

The results of the multi-contract SSVI calibration to the market data are shown in Figure~\ref{fig:SSVI_results}. We show a $5\%$ bis-ask spread as well to give an idea of the parametrization quality. The value of $5\%$ spread is consistent with actual market liquidity and will be further used as a benchmark for smile calibration.

\begin{figure}[H]
\begin{center}
    \includegraphics[width=1\linewidth]{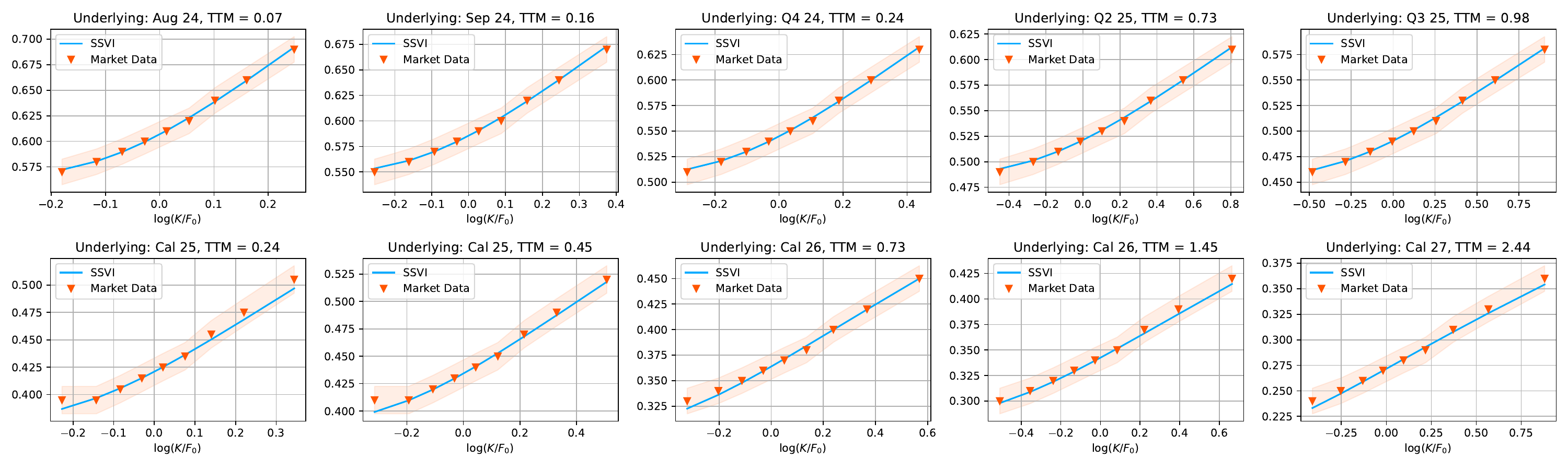}
    \caption{Example of the multi-factor SSVI calibration to the market data as of July $1^{st}$, 2024 on the German power market. We mention in the title of each plot the underlying futures contract and the Time To Maturity (TTM) of the smile in years.}
    \label{fig:SSVI_results}
\end{center}
\end{figure}

\section{Calibration methodology for historical covariances and variance swap volatilities}\label{s:joint_cov_vs_calib}

For such calibration step, we set
\begin{equation}
    V_t \equiv 1, \quad h(t) \equiv 1, \quad g(T) \equiv 1, \quad 0 \leq t \leq T \leq \bar{T},
\end{equation}
in \eqref{eq:HJM_def}.  Our aim is to calibrate the $N$ risk factors' volatility functions $\sigma$ as well as their correlation matrix $R$ in order to match
\begin{itemize}
    \item the realized daily log returns' covariances estimated from \eqref{eq:historical_covariance_matrix} for $P_{\mathrm{hist}}$ different \textit{rolling} futures contracts stripped from market quoting \textit{absolute} futures contracts by non-arbitrage arguments using \eqref{eq:historical_rolling_future_contract},
    \item the VS volatility term structure estimated from market data using the Carr-Madan formula \eqref{eq:vs_carr_formula} for the respective deliveries and maturities associated to $P_{\mathrm{smiles}}$ implied volatility smiles corresponding to $P_{\mathrm{imp}}$ different underlying futures contracts.
\end{itemize}
 First, we derive explicit formulas of the normalized integrated covariance of rolling futures contracts' log returns and of the VS volatility when the infinitesimal forward rate's dynamics is given by the $L$-$S$-$C$ specification of the form \eqref{eq:lsc_deterministic_vol_factors}.

\subsection{Explicit model covariance and VS volatility term structure} \label{ss:model_covariance_vs_volatility}

Let $H \in \mathbb{N}^{*}$ denote the number of days considered for the historical past horizon and $\tau_{d}>0$ a period in days. In this section, we derive explicit formulas for the covariances of the rolling futures contracts' log-returns as well as their VS volatilities under the KV approximation \eqref{eq:kv_approx}.

\paragraph{Stationarity of rolling futures' volatility.} To start with, by leveraging the stationary property of the $L$-$S$-$C$ volatility functions \eqref{eq:l_shape}--\eqref{eq:s_shape}--\eqref{eq:c_shape}, that is
\begin{equation}
    \sigma(t,T) = \sigma(T-t), \quad 0 \leq t \leq T,
\end{equation}
see for example \cite{Andersen2010} for additional insights on this property, we readily obtain, by a change of variable, the following stationary property of the volatility \eqref{eq:futures_contracts_volatility} of the rolling forward contract such that
\begin{equation} \label{eq:stationarity_vol_rolling}
    \Sigma_{t+u}(u+T_{s}, u+T_{e}) = \Sigma_{t}(T_{s}, T_{e}), \quad t \geq 0, \quad u \in \mathbb{R}.
\end{equation}

Following the definition of the rolling forward contract 
\eqref{eq:rolling_forward_model_definition} delivering on $[T_{s}, T_{e}]$, we define its {log returns over a period of $\tau_d$-days} under the KV approximation by
\begin{align}
    r_{t_{h}}(\tau_{d}) := \log \frac{\widetilde{F}_{t_{h}} (t_{h}+T_{s}, t_{h}+T_{e})}{\widetilde{F}_{t_{h}-\tau_{d}} (t_{h}+T_{s}, t_{h}+T_{e})} = & - \frac{1}{2} \int_{t_{h}-\tau_{d}}^{t_{h}} \Sigma_{t}(t_{h}+T_{s}, t_{h}+T_{e})^\top R \Sigma_{t}(t_{h}+T_{s}, t_{h}+T_{e}) \d t \\
    & + \int_{t_{h}-\tau_{d}}^{t_{h}} \Sigma_{t}(t_{h}+T_{s}, t_{h}+T_{e})^\top \d W_{t}, \quad h \in \left\{ 0, \ldots, H-\tau_{d} \right\}.
\end{align}

\begin{lemma} \label{L:stationary_covariance}
   Fix two rolling futures contracts with respective time to deliveries $T_{s}^{i}, \; T_{s}^{j}$ and delivery durations $T_{e}^{i}-T_{s}^{i}, \; T_{e}^{j}-T_{s}^{j}$. Define their \textit{normalized} integrated covariance of their respective {log returns over a period of $\tau_d$-days} $r^{i}(\tau_{d}), \; r^{j}(\tau_{d})$ at observation date $t_{h}, \; h \in \left\{ 0, \ldots, N \right\}$ by
    \begin{equation}
        \Cov{r_{t_{h}}^{i}(\tau_{d})}{r_{t_{h}}^{j}(\tau_{d})} := \frac{1}{\tau_{d}} \int_{t_{h}-\tau_{d}}^{t_{h}} d \langle r_{t_{h}}^{i}(\tau_{d}), r_{t_{h}}^{j}(\tau_{d}) \rangle.
    \end{equation}
    Then,   the covariance is explicitly given by 
    \begin{equation} \label{eq:model_stationary_covariance}
        \Cov{r^{i}(\tau_{d})}{r^{j}(\tau_{d})} = \frac{1}{\tau_{d}} \int_{0}^{\tau_{d}} \Sigma_{t}(\tau_{d}+T_{s}^{i}, \tau_{d}+T_{e}^{i})^\top R \Sigma_{t}(\tau_{d}+T_{s}^{j}, \tau_{d}+T_{e}^{j}) \d t,
    \end{equation}
     and is independent of the observation date $t_{h}$.
\end{lemma}

\begin{proof}
    Straightforward calculus yields
    \begin{align}
        \Cov{r_{t_{h}}^{i}(\tau_{d})}{r_{t_{h}}^{j}(\tau_{d})} & = \frac{1}{\tau_{d}} \int_{t_{h}-\tau_{d}}^{t_{h}} \Sigma_{t}(t_{h}+T_{s}^{i}, t_{h}+T_{e}^{i})^\top R \Sigma_{t}(t_{h}+T_{s}^{j}, t_{h}+T_{e}^{j}) \d t \\
        & = \frac{1}{\tau_{d}} \int_{0}^{\tau_{d}} \Sigma_{t+t_{h}-\tau_{d}}(t_{h}+T_{s}^{i}, t_{h}+T_{e}^{i})^\top R \Sigma_{t+t_{h}-\tau_{d}}(t_{h}+T_{s}^{j}, t_{h}+T_{e}^{j}) \d t \\
        & = \frac{1}{\tau_{d}} \int_{0}^{\tau_{d}} \Sigma_{t}(\tau_{d}+T_{s}^{i}, \tau_{d}+T_{e}^{i})^\top R \Sigma_{t}(\tau_{d}+T_{s}^{j}, \tau_{d}+T_{e}^{j}) \d t,
    \end{align}
    where we used respectively the change of variable $t \mapsto t+t_{h}-\tau_{d}$ and the \textit{stationary property of the volatility function $\Sigma$} \eqref{eq:stationarity_vol_rolling} to get the second and third equalities.
\end{proof}

\paragraph{Separability of rolling futures' volatility.} Furthermore, notice that the respective $L$, $S$ and $C$ volatility functions \eqref{eq:l_shape}, \eqref{eq:s_shape} and \eqref{eq:c_shape} enjoy the following separable property such that for any $n \in \left\{ 1, \ldots, N \right\}$
\begin{equation} \label{eq:def_factors_separability}
    \sigma_{n}^{Y}(t,T) = \sigma_{Y,n} \sum_{k=1}^{K} \sigma_{n}^{X_{k}}(t,T) = \sigma_{Y,n} \sum_{k=1}^{K} a_{X_{k}, n}^{\tau}(t) b_{X_{k}, n}^{\tau}(T), \quad 0 \leq t \leq T, \quad Y \in \left\{ L, S, C \right\},
\end{equation}
with $\sigma_{Y, n}>0$, $\left( a_{X_{k}, n}^{\tau} \right)_{k \in \{1, \ldots, K\}}$ and $\left( b_{X_{k}, n}^{\tau} \right)_{k \in \{1, \ldots, K\}}$ some explicit functions depending only on parameters $\tau$, and where $K \in \mathbb{N}^{*}$ denotes the total number of \textit{separable state variables} $\left( X_{k} \right)_{k \in \{ 1, \ldots, K\}}$. Notice that both $L$ and $S$-factors are readily identified to their unique state variable respectively, while, for $j \in \{ 1, \ldots, N_{c} \}$ we can decompose the $j^{th}$ $C$-factor as the sum of two separable state variables $C_{1}$ and $C_{2}$ such that
\begin{equation}
    \sigma_{j}^{C}(t,T) := \sigma_{C,j} \frac{T-t}{\tau_{C,j}} e^{-\frac{T-t}{\tau_{C,j}}} = \sigma_{C,j} \left( \sigma_{j}^{C_{1}}(t,T) + \sigma_{j}^{C_{2}}(t,T) \right), \quad 0 \leq t \leq T,
\end{equation}
where we define the volatility functions associated respectively to $C_{1}$ and $C_{2}$ by
\begin{align*}
    \sigma_{j}^{C_{1}}(t,T) & := e^{\frac{t}{\tau_{C,j}}} \times \frac{T}{\tau_{C,j}}e^{-\frac{T}{\tau_{C,j}}}, \quad 0 \leq t \leq T,\\
    \sigma_{j}^{C_{2}}(t,T) & := - \frac{t}{\tau_{C,j}}e^{\frac{t}{\tau_{C,j}}} \times e^{-\frac{T}{\tau_{C,j}}}, \quad 0 \leq t \leq T.
\end{align*}
We detail in Table \ref{tab:explicit_a_b_beta_functions} the explicit functions $a_{X}^{\tau}$ and $b_{X}^{\tau}$ for the four distinct state variables $X \in \left\{ L, S, C_{1}, C_{2} \right\}$.

\begin{table}[H]
\centering
\renewcommand{\arraystretch}{1.5}  
\setlength{\tabcolsep}{12pt}       
\caption{Explicit functions $a_{X}^{\tau}$, $b_{X}^{\tau}$ from \eqref{eq:def_factors_separability} and $\beta_{X}^{\tau}$ from \eqref{eq:model_futures_covariance} defined in\eqref{eq:def_beta} for the four state variables $X \in \left\{ L, S, C_{1}, C_{2}\right\}$ from the $L$-$S$-$C$ specification.}
\label{tab:explicit_a_b_beta_functions}
\begin{tabular}{c|c|c|c|c|}
\cline{2-5}  
& $L$ & $S$ & $C_{1}$ & $C_{2}$ \\ \hline
\multicolumn{1}{|c|}{$t \mapsto a_{X}^{\tau}(t)$} & $1$ & $e^{\frac{t}{\tau_{S}}}$ & $e^{\frac{t}{\tau_{C}}}$ & $-\frac{t}{\tau_{C}}e^{\frac{t}{\tau_{C}}}$ \\ \hline
\multicolumn{1}{|c|}{$T \mapsto b_{X}^{\tau}(T)$} & $1$ & $e^{-\frac{T}{\tau_{S}}}$ & $\frac{T}{\tau_{C}}e^{-\frac{T}{\tau_{C}}}$ & $e^{-\frac{T}{\tau_{C}}}$ \\ \hline
\multicolumn{1}{|c|}{$\beta_{X}^{\tau}(T_{e}, T_{s})$} & $1$ & $\frac{\tau_{S}}{T_e - T_s} \left( e^{\frac{-T_s}{\tau_{S}}} - e^{\frac{-T_e}{\tau_{S}}} \right)$ & $\frac{T_s + \tau_{C}}{T_e - T_s}e^{\frac{-T_s}{\tau_{C}}} - \frac{T_e + \tau_{C}}{T_e - T_s}e^{\frac{-T_e}{\tau_{C}}}$ & $\frac{\tau_{C}}{T_e - T_s} \left( e^{\frac{-T_s}{\tau_{C}}} - e^{\frac{-T_e}{\tau_{C}}} \right)$ \\ \hline
\end{tabular}
\end{table}

We are now ready to derive the explicit formulas of the rolling futures contracts log returns and VS volatilities which will play a key role in the formulation of the first step calibration problem.
\begin{proposition}
    Using the same notations as in Lemma \ref{L:stationary_covariance}, the normalized stationary covariance \eqref{eq:model_stationary_covariance} is explicitly given by
    \begin{align} \label{eq:model_futures_covariance}
        \Cov{r^{i}(\tau_{d})}{r^{j}(\tau_{d})} = \sum_{p,k=1}^{N+N_{c}} x_{p,k}^{\sigma, R} \beta_{X_{p}}^{\tau}(T_{e}^{i}, T_{s}^{i}) \beta_{X_{k}}^{\tau}(T_{e}^{j}, T_{s}^{j}) \frac{1}{\tau_{d}} \int_{0}^{\tau_{d}} a_{X_{p}}^{\tau}(u) a_{X_{k}}^{\tau}(u) \d u,
    \end{align}
    where, for $X_{p}, X_{k} \in \left\{ L, S, C_{1}, C_{2} \right\}$, we set the variables
    \begin{equation} \label{eq:def_quadratic_variables}
        x_{p,k}(\sigma, R) := \sigma_{X_{p}}\sigma_{X_{k}} R_{p,k}, \quad p, k \in \left\{ 1, \ldots, N+N_{c} \right\},
    \end{equation}
    and
    \begin{equation} \label{eq:def_beta}
        \beta_{X_{p}}^{\tau}(T_{e}^{i}, T_{s}^{i}) := \frac{1}{T_{e}^{i} - T_{s}^{i}} \int_{T_{s}^{i}}^{T_{e}^{i}} b_{X_{p}}^{\tau}(T) \d T, \quad p \in \left\{ 1, \ldots, N+N_{c} \right\},
    \end{equation}
    with $(\tau, \sigma, R) \in \left( \mathbb{R}_{+}^{*} \right)^{N} \times \left( \mathbb{R}_{+}^{*} \right)^{N} \times \mathbb{S}_{++}^{N}$.

    Similarly, the variance swap variance level \eqref{eq:vs_vol_model} is explicitly given by 
    \begin{align} \label{eq:model_futures_variance_swap}
        \left( \sigma_{\mathrm{VS}, T}^{\mathrm{model}, i} \right)^2 = \sum_{p,k=1}^{N+N_{c}} x_{p,k}^{\sigma, R} \beta_{X_{p}}^{\tau}(T_{e}^{i}, T_{s}^{i}) \beta_{X_{k}}^{\tau}(T_{e}^{i}, T_{s}^{i}) \frac{1}{T} \int_{0}^{T} a_{X_{p}}^{\tau}(u) a_{X_{k}}^{\tau}(u) \d u.
    \end{align}
    
    In particular, all the terms $\left( \beta_{X_{p}}^{\tau}(T_{e}^{i}, T_{s}^{i}) \right)_{p \in \left\{ 1, \ldots, N \right\}}$ and $\left( \int_{0}^{\tau_{d}} a_{X_{p}}^{\tau}(u) a_{X_{k}}^{\tau}(u) \d u \right)_{p,k \in \left\{ 1, \ldots, N \right\}}$ are explicit in the $L$-$S$-$C$ parametrization, as specified respectively in Tables \ref{tab:explicit_a_b_beta_functions} and \ref{tab:cross_term_small_t}.
\end{proposition}

\begin{proof}
    The explicit formulas \eqref{eq:model_futures_covariance} and  \eqref{eq:model_futures_variance_swap} are readily derived respectively from \eqref{eq:model_stationary_covariance} and \eqref{eq:vs_vol_model}, while using the definition \eqref{eq:futures_contracts_volatility} and the separability property of the $L$-$S$-$C$ factors \eqref{eq:def_factors_separability}.
\end{proof}

\begin{table}[H]
\centering
\renewcommand{\arraystretch}{1.5} 
\setlength{\tabcolsep}{8pt}       

\begin{tabular}{c|c|cc}
\cline{1-2}
\multicolumn{1}{|c|}{$L^{p}$} & $t_{2}-t_{1}$ & & \\
\cline{1-3}
\multicolumn{1}{|c|}{$S^{p}$ or $C_{1}^{p}$} &
$\tau_{S/C}^{p} \left( e^{\frac{t_2}{\tau_{S/C}^{p}}} - e^{\frac{t_1}{\tau_{S/C}^{p}}} \right)$ &
\multicolumn{1}{c|}{
\makecell{$\tilde{\tau}^{p,k} \left( e^{\frac{t_2}{\tilde{\tau}^{p,k}}} - e^{\frac{t_1}{\tilde{\tau}^{p,k}}} \right),$\\ $\tilde{\tau}^{p,k} := \frac{\tau_{S/C}^{p}\tau_{S/C}^{k}}{\tau_{S/C}^{p}+\tau_{S/C}^{k}}$}} & \\ \hline
\multicolumn{1}{|c|}{$C_{2}^{p}$} &
\makecell{$- \Big( \left( t_2 - \tau_{C}^{p} \right) e^{\frac{t_2}{\tau_{C}^{p}}} -$ \\ $\left( t_1 - \tau_{C}^{p} \right)e^{\frac{t_1}{\tau_{C}^{p}}} \Big)$} &
\multicolumn{1}{c|}{\makecell{$- \frac{\tilde{\tau}^{p,k}}{\tau_{C}^{p}} \Big( \left( t_2 - \tilde{\tau}^{p,k} \right) e^{\frac{t_2}{\tilde{\tau}^{p,k}}} -$ \\ 
$\left( t_1 - \tilde{\tau}^{p,k} \right)e^{\frac{t_1}{\tilde{\tau}^{p,k}}} \Big),$ \\ $\tilde{\tau}^{p,k} := \frac{\tau_{S/C}^{k}\tau_{C}^{p}}{\tau_{S/C}^{k}+\tau_{C}^{p}}$}} &
\multicolumn{1}{c|}{\makecell{$\frac{\tilde{\tau}^{p,k}}{\tau_{C^a} \tau_{C^b}} \Big( \left( 2\left(\tilde{\tau}^{p,k}\right)^2 - 2 \tilde{\tau}^{p,k} t_2 + t_2^2 \right) e^{\frac{t_2}{\tilde{\tau}^{p,k}}} -$ \\ 
$\left( 2\left(\tilde{\tau}^{p,k}\right)^2 - 2 \tilde{\tau}^{p,k} t_1 + t_1^2 \right) e^{\frac{t_1}{\tilde{\tau}^{p,k}}} \Big)$, \\ $\tilde{\tau}^{p,k} := \frac{\tau_{C}^{k}\tau_{C}^{p}}{\tau_{C}^{k}+\tau_{C}^{p}}$}} \\ \hline
& $L^{k}$ & \multicolumn{1}{c|}{$S^{k}$ or $C_{1}^{k}$} & \multicolumn{1}{c|}{$C_{2}^{k}$} \\
\cline{2-4}
\end{tabular}
\caption{Explicit computations for the cross-terms $\int_{t_{1}}^{t_{2}} a_{X_{p}}^{\tau}(u) a_{X_{k}}^{\tau}(u) \d u$ used in \eqref{eq:model_futures_covariance} for the types of state variables $X_{p}, X_{k} \in \left\{ L, S, C_{1}, C_{2} \right\}, \; p,k \in \left\{ 1, \ldots, N \right\}$. Notice in particular that $S$ and $C_{1}$ state variables share the same function $a_{X}^{\tau}$ as detailed in Table \ref{tab:explicit_a_b_beta_functions}.} 
\label{tab:cross_term_small_t} 
\end{table}

\subsection{Loss function for the historical covariances - VS volatility term structure calibration problem}

{We consider $P_{\mathrm{smiles}} := \sum_{i=1}^{P_{\mathrm{imp}}} N_{i} \in \mathbb{N}^{*}$ Variance Swap (VS) contracts with respective maturities $\left(\left( T_{j}^{i} \right)_{j \in \{1, \ldots, N_i\}} \right)_{i \in \{ 1, \ldots, P_{\mathrm{imp}} \}}$ and volatilities $\left(\left( \sigma_{\mathrm{VS}, T_j^i}^{\mathrm{mkt}, i} \right)_{j \in \{ 1, \ldots, N_i\}}\right)_{i \in \{ 1, \ldots, P_{\mathrm{imp}} \}}$ computed using \eqref{eq:extract_vs_vol_from_vs_swap}. Such VS contracts are indeed associated to the implied volatility smiles we aim to calibrate the initial model \eqref{eq:HJM_def} on, whose underlying are respectively the futures contracts with delivery periods $\left( \left[ T_{s}^{i}, T_{e}^{i} \right] \right)_{i \in \{ 1, \ldots, P_{\mathrm{imp}} \}}$, recall the notations from Section \ref{sect:vs_quote_extr}.} 

Denote by $H$ the number of past observation days until the date of calibration considered for the calibration. Then we introduce a family of $P_{\mathrm{hist}} \in \mathbb{N}^{*}$ rolling contracts $\left(\left( F \left( t_h, t_h + T_{s}^{k}, t_h + T_{e}^{k} \right) \right)_{h \in \{ 1, \ldots, H \}} \right)_{k \in \{ 1, \ldots, P_{\mathrm{hist}} \}}$ constructed by solving \eqref{eq:optimization_problem_stripping} and we use their respective daily log returns time-series
$$
\left( \left( r_{t_{h}}^{\mathrm{mkt}, k} (\tau_{d}) \right)_{h \in \left\{ 0, \ldots, H \right\}} \right)_{k \in \{ 1, \ldots, P_{\mathrm{hist}} \}}
$$ 
computed from \eqref{eq:historical_rolling_future_contract_returns}, for the estimation of the historical daily log returns' covariances via the estimator \eqref{eq:historical_covariance_matrix}. In our case, we typically chose $P_{\mathrm{hist}} \geq P_{\mathrm{imp}}$, with the set of the rolling forwards' delivery periods including those of the absolute underlying futures contracts, i.e.~
$$
\left( \left[ T_{s}^{i}, T_{e}^{i} \right] \right)_{i \in \{ 1, \ldots, P_{\mathrm{imp}} \}} \subset \left( \left[ T_{s}^{k}, T_{e}^{k} \right] \right)_{k \in \{ 1, \ldots, P_{\mathrm{hist}} \}}
$$
to ensure that the historical covariance term structure captures a priori, \textit{somehow}, the \textit{futures correlation structure} of those absolute futures.

We introduce a convex combination of losses between the fit of historical log returns' variance-covariance term structure of rolling forward contracts and the fit of the implied variance swap volatility levels such that
\begin{equation} \label{eq:def_loss}
    J^{\lambda}(\tau, \sigma, R) := \lambda J_{1}(\tau, \sigma, R) + \left( 1 - \lambda \right) J_{2}(\tau, \sigma, R), \quad \lambda \in [0,1],
\end{equation}
with
\begin{align} \label{eq:def_loss_covariance_fit}
    J_{1}(\tau, \sigma, R) & := \| C^{\mathrm{mkt}} - C^{\mathrm{model}}(\tau, \sigma, R) \|_{\Gamma}^{2}, \\ \label{eq:def_loss_vs_variances_fit}
    J_{2}(\tau, \sigma, R) & := \| \left( \sigma_{\mathrm{VS}}^{\mathrm{mkt}} \right)^2 - \left( \sigma_{\mathrm{VS}}^{\mathrm{model}} \right)^2(\tau, \sigma, R) \|_{w}^{2},
\end{align}
where the $(i,j)^{th}$ entry of $C^{\mathrm{model}}$ and $l^{th}$ entry of $\left( \sigma_{\mathrm{VS}}^{\mathrm{model}} \right)^2$ are respectively given by \eqref{eq:model_futures_covariance} and \eqref{eq:model_futures_variance_swap}, while $C^{\mathrm{mkt}}$ and $\sigma_{\mathrm{VS}}^{\mathrm{mkt}}$ are respectively estimated by \eqref{eq:historical_covariance_matrix} and \eqref{eq:vs_carr_formula}. Furthermore, $\Gamma := \left( \Gamma_{i,j} \right)_{i,j \in \{1, \ldots, P_{\mathrm{hist}} \}^{2}} \in \mathbb{R}_{+}^{P_{\mathrm{hist}}^{2}}$ denotes a family of weights associated to the matrix Frobenius norm, while $w := \left( w_{i} \right)_{i \in \{1, \ldots, P_{\mathrm{smiles}} \}} \in \mathbb{R}_{+}^{P_{\mathrm{smiles}}}$ are weights for the vector Frobenius norm, specified  respectively as in \eqref{eq:weight_specification} and uniformly in our case. Finally, notice that the hyper-parameter $\lambda$ conveniently controls the trade-off between the fit of the historical covariances and that of the VS volatility term structure.

Then, for a fixed $\lambda \in [0,1]$, the joint historical - VS variance term structure calibration is formulated as the following minimization problem under constraints
\begin{equation} \label{eq:minimisation_problem}
    \min_{(\tau, \sigma, R) \in U} J^{\lambda}(\tau, \sigma, R),
\end{equation}
where the admissible set of parameters is given by
\begin{equation} \label{eq:admissible_set_minimisation_problem}
    U := \left( \mathbb{R}_{+}^{*} \right)^{N} \times \left( \mathbb{R}_{+}^{*} \right)^{N} \times \mathbb{S}_{++}^{N}.
\end{equation}
The functional $J^{\lambda}$ \eqref{eq:def_loss} is clearly non-convex in $(\tau, \sigma, R)$, depending in particular on exponential terms in the parameters $\tau$, so we cannot guarantee a priori the existence of a global minimizer of such optimization problem \eqref{eq:minimisation_problem}.

\subsection{Solver specification}

The admissible set of parameters \eqref{eq:admissible_set_minimisation_problem} can be very large, with dimension in $\mathcal{O} \left( N^{4} \right)$, where recall $N$ is the total number of factors in the Nelson-Siegel parametrization \eqref{eq:lsc_deterministic_vol_factors}, and include positive constraints as well as a non-trivial positive definite cone constraint for the correlation matrix $R$. In order to simplify and fasten the numerical implementation of the optimization problem \eqref{eq:minimisation_problem}, we start by noticing the three following facts.
\begin{itemize}
    \item[(i)] As soon as the parameters $\tau$ are fixed, the functional $J^{\lambda}$ \eqref{eq:def_loss} becomes quadratic in the variables $\left( x_{p,k}(\sigma, R) \right)_{p, k \in \left\{ 1, \ldots, N+N_{c} \right\}}$ from \eqref{eq:def_quadratic_variables} such that
    \begin{align} \label{eq:loss_functional_quadratic_in_x}
        J^{\lambda}(\tau, \sigma, R) = & \lambda \sum_{i,j=1}^{P_{\mathrm{hist}}} W_{i,j} \left( C_{i,j}^{\mathrm{mkt}} - \sum_{p,k=1}^{N+N_{c}} w_{p,k}^{\tau}(i,j) x_{p,k}^{\sigma, R} \right)^{2}\\ \label{eq:functional_quadratic_in_x}
        & + \left( 1 - \lambda \right) \sum_{l=1}^{P_{\mathrm{smiles}}} w_{l} \left( \left( \sigma_{VS}^{\mathrm{mkt}} \right)^2 - \sum_{p,k=1}^{N+N_{c}}w_{p,k}^{\tau}(l) x_{p,k}^{\sigma, R} \right)^{2},
    \end{align}
    where, for $p, k \in \left\{ 1, \ldots, N+N_{c} \right\}$, the weights in the squares are respectively identified from \eqref{eq:model_futures_covariance} and \eqref{eq:model_futures_variance_swap} such that
    \begin{align} \label{eq:weights_tau_cov}
        w_{p,k}^{\tau}(i,j) & := \beta_{X_{p}}^{\tau}(T_{e}^{i}, T_{s}^{i}) \beta_{X_{k}}^{\tau}(T_{e}^{j}, T_{s}^{j}) \frac{1}{\tau_{d}} \int_{0}^{\tau_{d}} a_{X_{p}}^{\tau}(u) a_{X_{k}}^{\tau}(u) \d u, \quad i,j \in \{ 1, \ldots, P_{\mathrm{hist}}\}, \\ \label{eq:weights_tau_vs_var}
        w_{p,k}^{\tau}(l) & := \beta_{X_{p}}^{\tau}(T_{e}^{l}, T_{s}^{l}) \beta_{X_{k}}^{\tau}(T_{e}^{l}, T_{s}^{l}) \frac{1}{T} \int_{0}^{T} a_{X_{p}}^{\tau}(u) a_{X_{k}}^{\tau}(u) \d u, \quad l \in \{ 1, \ldots, P_{\mathrm{smiles}}\}.
    \end{align}

    \item[(ii)] Note that the loss functional \eqref{eq:functional_quadratic_in_x} is expressed in terms of the state variables, and for each $j \in \{ 0, \ldots, N_{c}\}$, the $j^{th}$ $C$-factor has two state variables sharing the same parameter $\sigma_{C,j}$ and the same Brownian motion so that we need to impose the following $2N_{c}$ equality constraints on the $x$ variables \eqref{eq:def_quadratic_variables} ordered as $\left( \left\{ L, \left(S_{i}\right)_{i \in \left\{ 0, \ldots, N_{s} \right\}}, \left((C_{1,j}, C_{2,j})\right)_{i \in \left\{ 0, \ldots, N_{c} \right\}} \right\} \right)$ such that
    \begin{equation} \label{eq:equality_constraints_c_factors}
    \begin{cases}
        x_{l,l}(\sigma, R) = x_{l+1,l+1}(\sigma, R) \\
        x_{l,l}(\sigma, R) = x_{l,l+1}(\sigma, R)
    \end{cases}, \quad l \in \{ N_{s}+1, \ldots, N_{s}+1+N_{c}\}.
    \end{equation}
    
    \item[(iii)] $S$ (resp. $C$-factors) are inter-changeable, which may cause numerical instability.
\end{itemize}

\paragraph{Successive non-linear -- linear cone program formulation.} Consequently, instead of solving \eqref{eq:minimisation_problem} globally, we will solve the following iterative minimization problems
\begin{equation} \label{eq:iterative_minimisation_problem}
    \min_{a \in \mathbb{R}^{N_{s}+N_{c}}} \min_{x(\sigma, R): \; (\sigma, R) \in \left( \mathbb{R}_{+}^{*} \right)^{N} \times \mathbb{S}_{++}^{N}} J^{\lambda}(\tau(a), x(\sigma, R)).
\end{equation}

On the one hand, the outward minimization is performed by an unconstrained non-linear solver (e.g.~ \href{https://docs.scipy.org/doc/scipy/reference/generated/scipy.optimize.minimize.html#scipy.optimize.minimize}{Powell minimizer from SciPy}), where we optimize on the parameters
\begin{equation}
    a := \left( a_{s}^{1}, \ldots, a_{s}^{N_{s}}, a_{c}^{1}, \ldots, a_{c}^{N_{c}} \right) \in \mathbb{R}^{N_{s}+N_{c}},
\end{equation}
such that the following change of variables
\begin{equation} \label{eq:change_of_variable_tau_params}
    \tau(a) := \left( e^{a_{s}^{1}}, e^{a_{s}^{1}} + e^{a_{s}^{2}}, \ldots, \sum_{i=1}^{N_{s}} e^{a_{s}^{i}}, e^{a_{c}^{1}}, e^{a_{c}^{1}} + e^{a_{c}^{2}}, \ldots, \sum_{i=1}^{N_{c}} e^{a_{c}^{i}} \right)
\end{equation}
ensures that the $\tau(a)$ parameters are indeed positive and strictly increasing for the $N_s$ $S$-factors and the $N_c$ $C$-factors respectively.

On the other hand, for a fixed $a \in \mathbb{R}^{N_{s}+N_{c}}$, $\tau(a)$ is fixed and the inward minimization problem in \eqref{eq:iterative_minimisation_problem} can be formulated and solved in terms of the variables $\left( x_{p,k}(\sigma, R) \right)_{p, k \in \left\{ 1, \ldots, N+N_{c} \right\}}$ from \eqref{eq:def_quadratic_variables} as a linear cone program ensuring $R$ is indeed semi-definite positive, as detailed in Appendix \ref{s:linear_cone_programming}.

\paragraph{Solver initialization and optimal parameters extraction.} All that remains to do is to initialize properly the outward non-linear solver in the iterative formulation \eqref{eq:iterative_minimisation_problem} and then extract the optimal parameters $\left( \sigma, R\right)$ upon convergence. Indeed, a good initialization $a_0 \in \mathbb{R}^{N_{s}+N_{c}}$ is of paramount importance to be able to reach a good local minimum in practice. 

{Fortunately the solver algorithm runs relatively fast, from a few seconds to a few minutes depending on the number of factors, so we can afford to iterate over various randomized initial guesses. As a rule of thumb, we construct the initial values $a_0 \in \mathbb{R}^{N_{s}+N_{c}}$ starting from the mid-points of delivery periods $\left( \left[ T_{s}^{k}, T_{e}^{k} \right] \right)_{k \in \{ 1, \ldots, P_{\mathrm{hist}} \}}$ such that the resulting factors with time-scales $\tau(a_{0})$ cover reasonably well the futures curve across all deliveries. In our case, since we have $N_s \leq P_{\mathrm{hist}}$ and $N_c \leq P_{\mathrm{hist}}$, we first reduce the number of such mid-points by taking convex combinations of them to the number of $S$ and $C$ factors, and then add up some randomness to obtain the initial positive $\tau(a_{0})$.}

Once an optimizer $\left( \hat{a}, \hat{x}(\hat{a}) \right)$ of the iterative optimization problem \eqref{eq:iterative_minimisation_problem} has been reached, then we extract the optimal parameters $\hat{\tau}$ injecting $\hat{a}$ into \eqref{eq:change_of_variable_tau_params}, and by positive-definiteness of the calibrated $L$-$S$-$C$ factors' covariance $\left(\hat x_{i,j} \right)_{i,j \in I}$ (extracted from an optimal solution $\hat{x}(\hat{a})$ of \eqref{eq:linear_cone_program}, with $I$ denoting the set of indices relative to $\left\{ L, S, C_{1} \right\}$ state variables, of cardinal $N$), we get
\begin{equation} \label{eq:extract_sigma_parameters}
    \hat{\sigma_{i}} := \sqrt{\hat{x}_{i,i}} > 0, \quad i \in I,
\end{equation}
and
\begin{equation} \label{eq:extract_correl_parameters}
    \hat{R}_{i,j} := \frac{\hat{x}_{i,j}}{\hat{\sigma}_{i} \hat{\sigma}_{j}}, \quad i,j \in \left\{ 1, \ldots, N \right\}.
\end{equation}

\begin{remark}
    Note that the linear cone program \eqref{eq:linear_cone_program} only guarantees the matrix $\left(\hat x_{i,j}\right)_{i,j \in I}$ to be semi-definite positive while the inversion formulas \eqref{eq:extract_sigma_parameters}--\eqref{eq:extract_correl_parameters} are well-defined if $\left(\hat x_{i,j}\right)_{i,j \in I}$ is positive definite. If zero happens to belong to the spectrum of $\left(\hat x_{i,j}\right)_{i,j \in I}$, then it means the number of factors could be reduced, either by withdrawing any $i^{th}$ factor associated to $\hat x_{i,i} = 0$, or, if $\hat R$ is indeed well-defined by \eqref{eq:extract_correl_parameters}, by summing (resp. subtracting) the volatility shape functions of the perfectly correlated (resp. anti-correlated) factors with respect to all the other factors. See for example \cite[Table 13-16-19]{feron2024estimation} where such latter phenomenon is repeatedly observed when calibrating historical futures' returns on the Italian, Swiss and UK markets respectively. In practice however, it is always possible to regularize or to tune the hyper-parameters of the non-linear solver such that $\left(\hat x_{i,j}\right)_{i,j \in I} \in \mathbb{S}_{++}^{N}$ (e.g.~the tolerance threshold, the number of iterations, etc).
\end{remark}

\section{Calibration methodology for implied smiles} \label{s:implied_calib}

Implied calibration focuses on calibrating the model to observed market option prices. We show how the prices of vanilla options can be computed in our model and propose a way to decouple the implied calibration into the {correction} of the implied volatility term structure given by the VS volatilities {to achieve a perfect fit,} and the calibration of smile shapes. 

\subsection{Fourier option pricing for implied calibration}\label{S:fastpricing}

In this section, we describe the vanilla options pricing techniques used in the implied calibration for a futures contract $\widetilde F_{.} = \widetilde F_{.}(T_s, T_e)$ with arbitrary delivery period $[T_s,\, T_e]$  in our model, where $\widetilde F$ is given by  \eqref{eq:KV_def} following the KV approximation.

\paragraph{The characteristic function of $\log \widetilde F$.}
 The European call option prices $C(T, K)$ can be calculated using the \cite{Lewis2001} formula
    \begin{equation} \label{eq:lewis_formula}
        C(T, K) = \E(\widetilde F_T - K)^+ = F_0 - \dfrac{K}{\pi}\int_0^\infty\Re\left(e^{i\left(u - \frac{i}{2}\right)\log\frac{F_0}{K}}\phi\left(u - \frac{i}{2}\right)\right)\dfrac{\d u}{u^2 + \frac{1}{4}},
    \end{equation}
    where $\phi(u) := \E\left[ e^{ iu\log\frac{\widetilde F_T}{F_0} }\right], \; u \in \mathbb{R}$ denotes the characteristic function of the normalized log-price. Due to the affine structure of the variance diffusion of the lifted Heston model  of \cite{lifted2019}, we show in the following theorem that the characteristic function $\phi$ of $\log  \widetilde F$ in the model  \eqref{eq:KV_def} is known up to the solution of a Riccati equation.
    \begin{theorem}\label{Thm:lifted_heston_riccati} Fix $v  \in\mathbb{C}$ such that   $\Re(v)\in [0,1]$.
        {For $T \leq T_s$}, the characteristic function of $\log \frac{\widetilde F_T}{F_0}$ is given by
        \begin{equation}\label{eq:characteristic_function}
            \E\left[\exp\left(v \log \frac{\tilde{F}_T}{F_0}\right)\right] = \exp\left(
            \int_0^T G\left(T-s, v, \psi(s)\right)\,\d s\right),
        \end{equation}
        where $\psi := \sum\limits_{j=1}^M c_j\psi^{j}$, and $\psi^j$ satisfies the Riccati equation
        \begin{equation} \label{eq:riccati_equation}
            \begin{cases}
                \dot\psi^j(t) &= -x_j\psi^j(t) + G\left(T - t, v, \psi(t)\right) \quad {a.e.} \\
            \psi^j(0) &= 0
            \end{cases},
        \end{equation}
        for $j \in \{1, \ldots, M\},$ and the function $G$ is defined by
\begin{equation}\label{eq:function_F}
            G(t, v, \psi) := \dfrac{\textcolor{black}{h(t)^2}}{2}\Sigma_{t}^\top R\Sigma_{t}(v^2 - v) + {h(t)}(\Sigma_{t}^\top\tilde\rho) v\psi + \dfrac{\psi^2}{2}.
        \end{equation}
        The Riccati equation \eqref{eq:riccati_equation} admits a unique global solution $\psi \in L^2_{\mathrm{loc}}(\R_+, \mathbb{C})$,  which is differentiable a.e.~and satisfies $\Re(\psi) \leq 0$.
    \end{theorem}
    \begin{proof} 
    
    The futures price dynamics \eqref{eq:KV_def} can be rewritten as a one-dimensional diffusion, i.e.~there exists a Brownian motion $\tilde W$ such that
        \begin{equation}
            \dfrac{\d \tilde{F}_t}{\tilde{F}_t} = \textcolor{black}{h(t)}\textcolor{black}{\sqrt{V_t}}\sqrt{\Sigma_t^\top R \Sigma_t}\d \tilde W_t, 
            \quad \d \langle B, \tilde W\rangle_t = \Sigma_t^\top \tilde\rho\,\d t,
            \quad t \in [0, T_s].
        \end{equation}
        Note that our model is a particular specification of the Volterra Heston model \cite[Section 7]{jaber2019affinevolterraprocesses} setting the kernel to
        \begin{equation}
             K(t) := \sum_{j=1}^M c_j e^{-x_j t}, \quad t \geq 0,
        \end{equation}
        and $V_0 = 1, \ \theta = 1, \ \lambda = 0, \ \nu = 1.$
        Repeating the proof of \cite[Theorem 7.1 (ii)]{jaber2019affinevolterraprocesses} adapted to take into account the bounded deterministic volatility component $t \mapsto \theta(t) := h(t)\sqrt{\Sigma_t^\top R \Sigma_t}$, we can prove the existence of $\psi \in L^2_{\mathrm{loc}}(\R_+, \mathbb{C})$ satisfying $\Re(\psi) \leq 0$, which is the solution to the Volterra--Riccati equation
        \begin{equation}
            \psi(t) = \sum_{j=1}^M c_j\int_0^t e^{-x_j(t-s)} G(T - s, v, \psi(s)) \, \d s, \quad t \ge 0.
        \end{equation}
        The solution can be written in the form $\psi = \sum\limits_{j=1}^M c_j\psi^{j}$, where, for $ j \in \{1, \ldots, M\}$, $\psi^j$ satisfies 
\begin{equation}\label{ew:psi_j_integral}
            \psi^j(t) = \int_0^t e^{-x_j(t-s)} G(T - s, v, \psi(s)) \, \d s,
        \end{equation}
        which is equivalent to the equation \eqref{eq:riccati_equation} by the variation of constants formula. Moreover, \eqref{ew:psi_j_integral} implies that $\psi^j$ is absolutely continuous, hence $\psi$ is differentiable a.e., which concludes the proof.
    \end{proof}
    
\paragraph{Numerical scheme for the Riccati equation.}
For $j \in \{1, \ldots, M\}$, the Riccati equation \eqref{eq:riccati_equation} can be rewritten in the integral form 
\begin{equation*}
    \psi^j(t) = \int_0^te^{-x_j(t-s)}G(T - s, iu, \psi(s))\, \d s, \quad t \geq 0,
\end{equation*}
so that $\psi^j(t)$ satisfies 
\begin{align*}
    \psi^j(t + h) &= e^{-x_jh}\psi^j(t) + \int_t^{t+h}e^{-x_j(t+h-s)}G(T - s, iu, \psi(s))\, \d s \\
    &\approx e^{-x_jh}\psi^j(t) + G(T - t, iu, \psi(t))\int_t^{t+h}e^{-x_j(t+h-s)}\, \d s \\
    &= e^{-x_jh}\psi^j(t) + \dfrac{1 - e^{-x_jh}}{x_j}G(T - t, iu, \psi(t)),
\end{align*}
which is used as a numerical scheme to solve the Riccati equation \eqref{eq:riccati_equation}. Such numerical scheme is more stable and precise than the explicit Euler discretization scheme, especially if one of the mean-reversion speeds $x_j, \; j \in \{1, \ldots, M\}$ reaches extreme values.

\subsection{Exact calibration correction of the VS volatility term structure} \label{section:gh_calib}

{The second step of the joint calibration consists of the correction and perfect match of the VS volatility term structure using the functions $g$ and $h$. Precise calibration of the VS volatilities is indeed crucial for achieving a good overall implied fit in step 3: since $\mathbb{E}[V_t] \equiv 1, \; t \geq 0 $, the stochastic volatility parameters responsible for the smile shape do not impact the VS volatilities, and hence cannot correct the VS volatility term structure.

Although the VS volatilities are calibrated by the $L$-$S$-$C$ risk factors in step 1, as described in Section~\ref{s:joint_cov_vs_calib}, their fit is often not exact since the primary goal of this step is to capture the historical futures' correlation structure with as few factors as possible but also match a priori the VS volatilities, where fitting the latter can be considered an \textit{implied regularization} in the loss $J^{\lambda}$ \eqref{eq:def_loss}. Yet, at the same time, this regularization is necessary to capture the VS volatility term structure at least approximately in order to avoid over-fitting the functions \( g \) and \( h \). 
Indeed, if the VS term structure were calibrated poorly at step 1, then the functions $g$ and $h$ would vary significantly, which would lead to a strong non-stationarity of the volatility (see Lemma \ref{L:stationary_covariance}) and eventually degrade the model's extrapolation ability, mostly induced by step 1. Hence, it is important to keep these functions $g$ and $h$ close to one in practice.}

In what follows, we assume that the functions $\color{black} T \mapsto g(T)$ and  $\color{black} t \mapsto h(t)$ are piece-wise constant. The discontinuity points of the function $g$ are the delivery start and end dates of the underlying futures contracts, while the discontinuities of $h$ are the maturities of the associated vanilla options.

\paragraph{Motivating example.}
Here, we provide some intuition on the calibration of the functions $h$ and $g$ and clarify why only one of these functions is not enough {to fit the VS volatility term structure} by considering a simple calibration set containing only five futures contracts: Oct 24, Nov 24, Dec 24, Q4 24, and Cal 25, with observation date the $1^{st}$ September 2024, see Figure \ref{fig:gh_illustr}. The maturities of options on monthly contracts are five days before the start of delivery (at $T_1, T_2$, and $T_3$ for Oct 24, Nov 24, and Dec24 respectively). The quarterly contract Q4 24 expires in September at $T_1$, and the calendar futures Cal 25 has two maturities $T_1$ and $T_4$ corresponding to September and December expiries.

The VS term structure calibration consists in matching model and market log-contract prices {respectively given by} \eqref{eq:vs_price_model} and \eqref{eq:vs_market}, i.e.~for any IV smile on a futures contract with delivery period $[T_s,\, T_e]$ and maturity $T$, one should ensure that 
\begin{equation}\label{eq:TS_calibration_condition}
    \int_0^T{h(t)^2}\Sigma_t^\top(T_s, T_e) R\Sigma_t(T_s, T_e)\,\d t = \mathrm{VS}_{T}^{\mathrm{mkt}},
\end{equation}
where $\Sigma_{t}(T_s, T_e) = \dfrac{1}{T_e - T_s}\int_{T_s}^{T_e}\sigma(., \tau)g(\tau)\,\d \tau$.
This can be achieved in two ways:
\begin{itemize}
    \item either adjusting the function $T \mapsto g(T)$ on the interval $[T_s,\, T_e]$;
    \item or modifying the function $t \mapsto h(t)$ on $[0, T]$.
\end{itemize}
In our example, each monthly futures can be calibrated with the function $T \mapsto g(T)$ which is constant on each of three months with the values chosen to verify \eqref{eq:TS_calibration_condition}. However, if all the three monthly futures are calibrated with the function $g$, there is no more degree of freedom to adjust the variance swap price for the quarterly contract Q4 24. { Alternatively, one could calibrate three smiles on Oct 24, Nov 24 and Q4 24 using the value of $g$ on Dec 24 to fit the calibration condition \eqref{eq:TS_calibration_condition} for the smile on Q4 24.} 

Moreover, the function $g$ cannot calibrate the early expiries of the calendar contract Cal 25 (for instance, calibrate simultaneously the log-contract prices for the maturities $T_1$ and $T_4$) since $g$ impacts on the overall volatility level of the futures contract: {therefore $g$ can fit the smile level for only one maturity per delivery.}

{For these reasons, the function $t \mapsto h(t)$ needs to be introduced:} we could fix the value of $h$ on $[0,\, T_1]$ to satisfy \eqref{eq:TS_calibration_condition} for the smile on Cal 25 with maturity $T_1$, then adjust $h$ on $[T_1, T_4]$ to ensure \eqref{eq:TS_calibration_condition} holds its smile with maturity $T_4$. {Alternatively}, one could calibrate the value of $h$ on $[0, T_1]$ to fit \eqref{eq:TS_calibration_condition} for the smile on Q4 24 while the smiles on monthly futures are calibrated with the function $g$. However, two problems arise. 

First, is is impossible to calibrate simultaneously both smiles written on Q4 24 and on Cal 25 since they share the same maturity $T_1$. Hence, one of them should be necessarily calibrated with the function $g$. In this example, the following solution is possible:
\begin{enumerate}
    \item the function $g$ is calibrated to the smiles on Oct 24, Nov 24, Dec 24, and Cal 25 with maturity $T_1$,
    \item the function $h$ is calibrated to the smiles on Q4 24 and on Cal 25 with maturity $T_4$.
\end{enumerate}

Second, once the monthly contracts are calibrated with the function $g$, any modification of the function $h$ (up to the smile maturity) made during the calibration of smile on Q4 24, will impact the log-contract price corresponding to the monthly futures, and vice-versa. This is simply due to the fact that the log-contract price in \eqref{eq:TS_calibration_condition} depends both on $g$ on $[T_s,\, T_e]$ and on $h$ on $[0, T]$. 

That is why an iterative calibration algorithm is needed. One step of such algorithm should include the calibration of the function $g$ to fit \eqref{eq:TS_calibration_condition} for one part of smiles and then the calibration of $h$ taking into account the modifications of $g$ to verify \eqref{eq:TS_calibration_condition} for the remaining smiles (possibly breaking this condition for the smiles from the first group). If such iterative algorithm converges to a fixed-point, at the limit, the smiles from both groups will be calibrated. This fixed-point algorithm is the key idea of the proposed calibration method and will be discussed in more detail below.

\begin{figure}[H]
\begin{center}
    \includegraphics[width=0.8\linewidth]{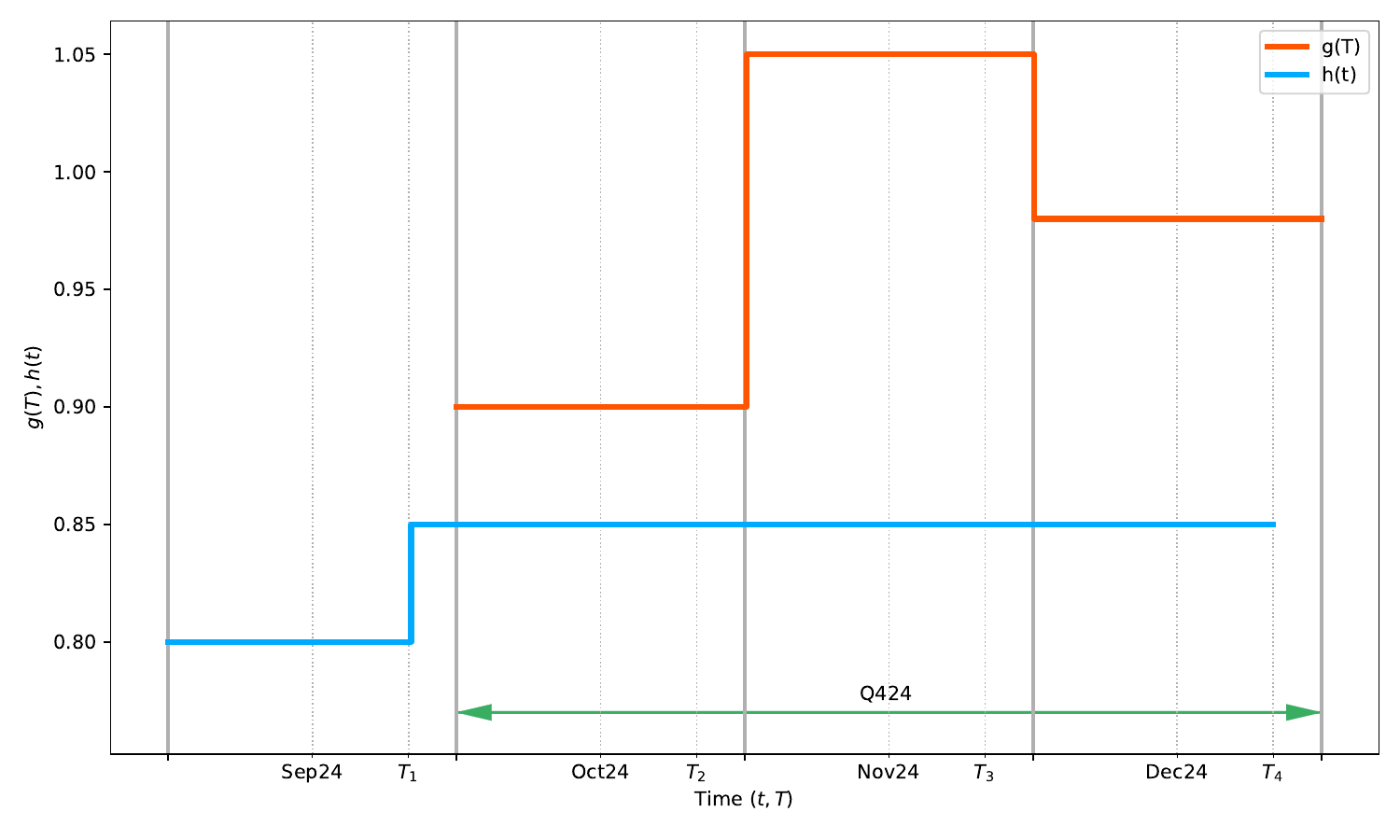}
    \caption{An example of the functions $g$ (in dark orange) and $h$ (in blue) covering the period from September 1, 2024 to January 1, 2025. The function $g$ is constant in October, in November and in December (months are separated by solid gray lines). The function $h$ is constant between the option maturities $T_1, T_2, T_3, T_4$ corresponding to five days before the start of the next months.}
    \label{fig:gh_illustr}
\end{center}
\end{figure}

\paragraph{Two groups of smiles.} From the discussion above it is clear that the implied volatility smiles should be divided into two groups: for the smiles in the first group the calibration condition is guaranteed by the function $g$, while the smiles from the second one are calibrated with the function $h$. It also follows from the example, that these groups should satisfy the following conditions.

Conditions on the smiles calibrated with the function $g$:
\begin{enumerate}
    \item \textbf{one maturity per futures}: Only one smile for a given futures contract can be calibrated with $g$,
    \item \textbf{linear independence}: The underlying futures corresponding to these smiles should be linearly independent, i.e.~no delivery period should be representable as a union of delivery periods of other futures contracts.
\end{enumerate}

Condition on the smiles calibrated with the function $h$:
\begin{enumerate}
    \item[3] \textbf{no coinciding maturities}: for a given maturity, only one smile may be chosen.
\end{enumerate}

Note that such division may be not unique. To fix the division, we decide to calibrate the smiles on monthly and quarterly contracts with the function $g$ and all the smiles on calendar contracts with the function $h$. Although the latter arbitrary division is not always available, it is almost always possible to find another division satisfying the conditions 1--3.

\paragraph{Calibration of the function $t \mapsto h(t)$.}
    We strip the function $h$ to match the variance swap prices implied from the volatility smiles in \eqref{eq:vs_carr_formula} and the variance swap prices given by the model \eqref{eq:vs_price_model} for futures contracts $F^i$, for each $i \in \{ 1, \ldots, P_{\mathrm{imp}}\}$. Fixing $j \in \{1, \ldots, N_i\}$ referring to the $j$-th smile of the $i$-th contract, we set $h$ on $[T_{j, \mathrm{prev}}^i,\, T_{j}^i)$ equal to
    \begin{equation}\label{eq:h_calibration}
         {h(t)} :=  \sqrt{\frac{\mathrm{VS}_{T_{j}^i}^{i} - \mathrm{VS}_{T_{j, \mathrm{prev}}^i}^{i}}{\displaystyle \int_{T_{j, \mathrm{prev}}^i}^{T_{j}^i} \Sigma^\top_{s}(T_s^i, T_e^i)  R \Sigma_{s}(T_s^i, T_e^i) \d s}}, \quad t \in [T_{j, \mathrm{prev}}^i,\, T_{j}^i),
    \end{equation}
    where $T_{j, \mathrm{prev}}^i$ is the previous smile maturity in the sorted list of all maturities and
    \begin{equation*}
        \mathrm{VS}_{T_{j, \mathrm{prev}}^i}^{i} :=  \int_0^{T_{j, \mathrm{prev}}^i}h^2(s)\Sigma^\top_{s}(T_s^i, T_e^i)  R \Sigma_{s}(T_s^i, T_e^i)\, \d s
    \end{equation*}
    Note that $t \mapsto h(t)$ in \eqref{eq:h_calibration} depends on the function $T \mapsto g(T)$ through the deterministic variance $\Sigma_t(T_s^i, T_e^i)$.
    
\paragraph{Calibration of the function $T \mapsto g(T)$.}
The discontinuity points of $g$ coincide with contracts' delivery start and end dates $(T_s^i)_{i \in 1, ..., P_{\mathrm{imp}}}, \; (T_e^i)_{i \in 1, ..., P_{\mathrm{imp}}}$. We adopt the following notation for the values of $g$: if the function $g$ is constant on a delivery period $[T_s^i, T_e^i]$, we will denote its value by $\bar g_{i}$. Otherwise, if other delivery periods of futures contracts are included in $[T_s^i, T_e^i]$, we will denote by $\bar g_{i-}$ the value of $g$ on $[T_s^i, T_e^i]\setminus \cup_{k \in \mathcal{J}_i} [T_s^k, T_e^k]$, where $\mathcal{J}_i$ stands for the indices of futures contracts which smiles are calibrated with the function $g$ and which delivery periods are included in $[T_s^i, T_e^i]$. We also denote by $\hat\Sigma_.$ the volatility of a futures contract delivering over a set $I$ with $g \equiv 1$ such that
\begin{equation}\label{eq:Sigma_hat}
    \hat\Sigma_{.}(I) := \frac{1}{\mathrm{Leb}(I)} \int_I\sigma (\cdot, T) \d T,
\end{equation}
where $\mathrm{Leb}(I)$ stands for the Lebesgue measure of $I$. With a slight abuse of notation, we use interchangeably $\Sigma_{.}([T_s^i, T_e^i])$ and $\Sigma_{.}(T_s^i, T_e^i)$, {recall \eqref{eq:futures_contracts_volatility}}.

Hence, the deterministic volatility component of the contract $F^i$ is given by
\begin{equation}\label{eq:Sigma_decomposition}
    \Sigma_{.}(T_s^i,\, T_e^i) = \sum_{k \in \mathcal{J}_i} \omega_{i}^k \textcolor{black}{\bar g_k} \hat\Sigma_{.}(T_s^k,\, T_e^k) + \Bigl(1-\sum_{k \in \mathcal{J}_i} \omega_{i}^k\Bigl)\textcolor{black}{\bar g_{i-}}\hat\Sigma_{.}([T_s^i, T_e^i]\setminus \cup_{k \in \mathcal{J}_i} [T_s^k, T_e^k]),
\end{equation}
where 
\begin{equation}\label{eq:omega_def}
    \omega_{i}^k = \dfrac{T_e^k - T_s^k}{T_e^i - T_s^i}
\end{equation}
denotes the relative weight of the contract $F^k$ volatility in the volatility of $F^i$ which is proportional to the length of its delivery period.

For the volatility smiles with underlying futures contract $F^i$ and maturity $T^i_j$ (for some $j \in \{1, \ldots, N_i\}$, uniquely determined thanks to condition 1 of smile division), such that $\mathcal{J}_i = \varnothing$, matching the log-contract prices \eqref{eq:TS_calibration_condition} leads to

\begin{equation}\label{eq:g_calibration_1}
    \textcolor{black}{g(T)} = \bar{g}_i = \sqrt{\frac{\mathrm{VS}_{T_j^i}^{i}}{ \displaystyle \int_0^{T_j^i} \textcolor{black}{h(s)^2} \hat\Sigma_{s}^\top(T_s^i, T_e^i)  R \hat\Sigma_{s}(T_s^i, T_e^i) \d s}}, \quad T \in [T_s^i, T_e^i],
\end{equation}
since $\Sigma_{.}(T_s^i, T_e^i) = \bar{g}_i\hat\Sigma_{.}(T_s^i, T_e^i)$ by \eqref{eq:Sigma_decomposition}. Note that, since $j$ is uniquely determined for each contract, it can be omitted in the notation $\bar{g}_i$.

For the contracts $F^i$ not calibrated yet, i.e.~such that $\mathcal{J}_i \not= \varnothing$, the function $g$ should be chosen
\begin{equation}\label{eq:g_calibration_2}
    \textcolor{black}{g(T)} = \textcolor{black}{\bar g_{i-}}, \quad T \in [T_s^i, T_e^i] \setminus \bigcup_{k \in \mathcal{J}_i} [T_s^k, T_e^k],
\end{equation}
where $\bar g_{i-}$ is a solution of the quadratic equation
\begin{equation}
 \mathrm{VS}_{T_{j}^i}^{i} = \int_0^{T_j^i} \textcolor{black}{h(s)^2} \Sigma_{s}^\top(T_s^i, T_e^i) R \Sigma_{s}(T_s^i, T_e^i) \d s,
\end{equation} 
with the deterministic volatility component given by \eqref{eq:Sigma_decomposition} with $\{\bar g_{k}\}_{k \in \mathcal{J}_i}$ already calibrated by \eqref{eq:g_calibration_1}. Note that the set $[T_s^i, T_e^i] \setminus \bigcup_{k \in \mathcal{J}_i} [T_s^k, T_e^k]$ is non-empty thanks to the linear independence condition imposed on the smiles division.

\paragraph{Fixed-point calibration algorithm.}
Note that the model is calibrated when the equation \eqref{eq:TS_calibration_condition} is verified for all smiles simultaneously.
Since the functions $g$ and $h$ are interdependent, we propose the following algorithm which is supposed to converge to the desired solution.

\begin{algorithm}[H]
\caption{Fixed-point calibration algorithm}\label{alg:fixed-point}
\begin{algorithmic}
\State \textbf{set} $\textcolor{black}{h^0(t)} \equiv 1$ and $\textcolor{black}{g^0(T)} \equiv 1 $
\State $n \gets 0$
\While{$||\textcolor{black}{g^{n}} - \textcolor{black}{g^{n-1}}||_{\infty} > \epsilon$}
\State \textbf{recalculate} the deterministic volatilities $(\Sigma_.(T_s^i, T_e^i))_{i \in \{1, \ldots, P_{\mathrm{imp}}\}}$
\State \textbf{calibrate} the function $\textcolor{black}{h^{n+1}}$ using \eqref{eq:h_calibration} with $g = g^n$.
\State \textbf{calibrate} the function $\textcolor{black}{g^{n+1}}$ using \eqref{eq:g_calibration_1} and \eqref{eq:g_calibration_2} with $h = h^{n + 1}$.
\State $n \gets n + 1$
\EndWhile
\end{algorithmic}
\end{algorithm}

Denoting one iteration of this algorithm by $\psi: \textcolor{black}{g^{n}} \mapsto \textcolor{black}{g^{n+1}}$, we note that model is calibrated if and only if $\textcolor{black}{g}$ is a fixed-point of this mapping, i.e.~$\psi(\textcolor{black}{g}) = \textcolor{black}{g}$.

Under additional conditions on the implied volatility data consistency, the following result holds:

\begin{theorem}\label{Thm:fixed_point}
    Suppose that all the smiles calibrated with $h$ have the same underlying $F^1$ and that $\mathcal{J}_i = \varnothing$ for the remaining contracts $F^i, \; i \in \{2, \ldots, P_{\mathrm{imp}}\}$. Suppose also that all the instantaneous correlations between the futures contracts are positive. Then, 
    \begin{enumerate}
        \item[(i)] $\psi$ admits a fixed-point if
    \begin{equation}\label{eq:fixed-point_condition} 
        \sum_{i=2}^{P_{\mathrm{imp}}}\dfrac{\mathrm{VS}_{T^i}^i}{\sum\limits_{k=1}^{N_1}\Big(\mathrm{VS}^1_{T^1_k} - \mathrm{VS}^1_{T^1_{k-1}}\Big)\dfrac{1}{\|Q_k\|}\int\limits_{[0,\,T^i]\cap[T_{k-1},\, T_k]}\hat\Sigma_{s}^\top(T_s^i, T_e^i)R\hat\Sigma_{s}(T_s^i, T_e^i)\, \d s} < 1,
    \end{equation} 
    where $(Q_k)_{ij} := \omega_{1}^i\omega_{1}^j \int_{T^1_{k-1}}^{T^1_k} \hat\Sigma_{s}^\top(T_s^i, T_e^i)R\hat\Sigma_{s}(T_s^j, T_e^j)\,\d s$ and $\omega_1^i$ are defined by \eqref{eq:omega_def}, $i, j = \{2, \ldots, {P_{\mathrm{imp}}}\}$.
        \item[(ii)] All the fixed points are stable. 
        \item[(iii)] If the contract $F^1$ has only one maturity, the fixed point is unique. 
    \end{enumerate}
\end{theorem}

\begin{proof}
    The proof is given in Appendix \ref{S:thm_fixed-point}.
\end{proof}

This result is clearly partial as it does not cover the case of multiple underlyings corresponding to smiles calibrated with $h$ and does not admit the nested contracts for smiles calibrated with $g$. However, the numerical experiments demonstrated the existence of a unique stable fixed point for all the test cases, even the ones not covered by Theorem \ref{Thm:fixed_point}. Thus, we believe that the provided result can be proved in a much more general case, though the direct approach used in the proof of Theorem \ref{Thm:fixed_point} is not applicable there.

{
\begin{remark}
Note that Theorem \ref{Thm:fixed_point} provides condition \eqref{eq:fixed-point_condition} which seems to be a universal condition ensuring that the model can be calibrated. Indeed, it can be interpreted as a no-arbitrage condition for the variance swaps. For example, consider the smile on Cal 25 with maturity $T_{\mathrm{Dec}}$ calibrated with $h$ and Q2 25 with maturity $T_{\mathrm{Mar}}$ calibrated by $g$. In this case, \eqref{eq:fixed-point_condition} reads
\begin{equation} \label{eq:condition_no_arbitrage_VS_example}
     \dfrac{\mathrm{VS}_{T_{\mathrm{Mar}}}^{\mathrm{Q2  }\ {25}}}{\mathrm{VS}^{\mathrm{Cal\;25}}_{T_{\mathrm{Dec}}}}\dfrac{(\frac14)^2 \int_{0}^{T_{\mathrm{Dec}}} \hat\Sigma_{s}^\top(\mathrm{Q2\;25})R\hat\Sigma_{s}(\mathrm{Q2\;25})\,\d s}{\int\limits_{0}^{T_{\mathrm{Mar}}}\hat\Sigma_{s}^\top(\mathrm{Q2\;25})R\hat\Sigma_{s}(\mathrm{Q2\;25})\,\d s} < 1.
\end{equation}
Since in the calibrated model
\begin{equation}
    \mathrm{VS}_{T_{\mathrm{Mar}}}^{\mathrm{Q2\;25}} = \int\limits_{0}^{T_{\mathrm{Mar}}}h(s)^2\Sigma_{s}^\top(\mathrm{Q2\;25})R\Sigma_{s}(\mathrm{Q2\;25})\,\d s = h_1^2g_{\mathrm{Q2\;25}}^2\int\limits_{0}^{T_{\mathrm{Mar}}}\hat\Sigma_{s}^\top(\mathrm{Q2\;25})R\hat\Sigma_{s}(\mathrm{Q2\;25})\,\d s,
\end{equation}
the condition \eqref{eq:condition_no_arbitrage_VS_example} is equivalent to
\begin{equation}
    \left(\dfrac{1}{4}\right)^2\int_{0}^{T_{\mathrm{Dec}}}h(s)^2\Sigma_{s}^\top(\mathrm{Q2\;25})R\Sigma_{s}(\mathrm{Q2\;25})\,\d s < \mathrm{VS}^{\mathrm{Cal\;25}}_{T_{\mathrm{Dec}}}.
\end{equation}
The integrand in the left-hand side expression is the part of variance of Cal 25 contract corresponding to Q2 25. Since the covariances are supposed to be positive, this part of variance is smaller than the whole variance of Cal 25, and thus, the price of the variance swap on Cal 25 should be greater than the integral on the left. Thereby, \eqref{eq:fixed-point_condition} is a condition describing the consistency of the variance swap data with historically calibrated deterministic volatility functions $\Sigma_{.}(T_s, T_e)$.
\end{remark}
}

\subsection{Smile shape calibration}\label{sect:smile_shape_calib}

\paragraph{Parametrization of the correlations.}\label{section:spot_vol_correls}
Fix $i \in \{ 1, \ldots, P_{\mathrm{imp}}\}$. For the approximated futures contract $\widetilde{F}^i$, the smile skew is determined by the ``spot-vol'' correlation 
\begin{equation}\label{eq:spot_vol_correl}
    \rho_{i}(t) = \dfrac{\left\langle dV,\, d\log \widetilde F^i\right\rangle_t}{\sqrt{\langle dV\rangle_t}\sqrt{\langle d\log \widetilde F^i\rangle_t}} 
    = \dfrac{\Sigma^\top_{t}(T_s^i, T_e^i)\tilde\rho}{\sqrt{\Sigma_{t}^\top(T_s^i, T_e^i)R\Sigma_{t}(T_s^i, T_e^i)}} = \hat\rho . \dfrac{L^\top\Sigma_{t}(T_s^i, T_e^i)}{\|L^\top\Sigma_{t}(T_s^i, T_e^i)\|}, \quad t \geq 0,
\end{equation}
where $\mathbf{1}$ denotes a vector filled with ones, and $L$ is a lower triangular matrix from the Cholesky decomposition \eqref{eq:Cholesky}. As $\hat\rho$ lies in $\{\rho\colon\, \|\rho\| \leq 1\}$, $\tilde\rho$ should be in $\{\rho\colon\, \|L^{-1}\rho\| \leq 1\}$ for the extended covariance matrix of $(W, B) \in \R^{N + 1}$ to be well-defined.

Since the skew of the the implied volatility smile is determined by the correlations \eqref{eq:spot_vol_correl}, it is more natural to use them and not $\tilde\rho$ in the calibration routine for several reasons. First, if the number of contracts is smaller than the number of historical risk factors $N$, this will reduce the number of parameters. Second, one spot-vol correlation impacts only one smile related to the corresponding futures contract, whereas the coefficients $\tilde\rho$ impact all the smiles in a way difficult to interpret. Thus, we expect to see a better solver behavior when the variables being optimized are the ``spot-vol'' correlations. {Finally, even if one decides to calibrate $\tilde\rho$ directly, the ``spot-vol'' correlations may provide a reasonable initial guess given by the calibrated correlations in the multi-contract SSVI parametrization described in Section \ref{sect:vs_quote_extr}.}

In order to reconstruct $\tilde\rho$ from a set of ``spot-vol'' correlations $(\rho_{1}^*, \ldots, \rho_{P}^*)$, we find numerically a solution $\hat\rho^*$ to the following optimization problem
\begin{equation}\label{eq:spot_vol_to_rho_tilde}
    \min_{\hat\rho\colon\ \|\hat\rho\| \leq 1} \sum_{i=1}^P \left\| {\int_0^{T_{C_M}} 
    \hat\rho^\top\,{L^\top\Sigma_{t}(T_s^i, T_e^i)}
    \, \d t} - \rho_{i}^* \int_0^{T_{C_M}}\|L^\top\Sigma_{t}(T_s^i, T_e^i)\| \, \d t\right\|^2
\end{equation}
and set $\tilde\rho = L\hat\rho^*$.

\paragraph{Optimization problem.}
We denote the parameters being calibrated by $\mathcal{P} = (\rho,\, c,\, x)$, where the spot-vol correlations $\rho$ are transformed to $\tilde\rho$ by \eqref{eq:spot_vol_to_rho_tilde} and \eqref{eq:get_rho_tilde_from_rho_hat}. The optimization problem then reads
\begin{equation}\label{eq:optimization_problem}
    \min_{\mathcal{P}} \sum_{i=1}^{P_{\mathrm{imp}}} \sum_{j=1}^{N_i}\sum_{K \in \mathcal{K}_j^i}\displaystyle \left(\dfrac{\mathrm{Call}^{\mathrm{mkt}, i}(T_{j}^i, K) -\mathrm{Call}^{\mathrm{Model}, i}(T_{j}^i, K)}{\mathcal{V}(K,T, \sigma_{\mathrm{IV}}^{\mathrm{mkt}, i}(T_{j}^i, K))}\right)^2, 
\end{equation}

where $\mathcal{V}(K,T, \sigma_{\mathrm{IV}}^{\mathrm{mkt}, i}(T_{j}^i, K))$ denotes the Black-Scholes vega corresponding to the market implied volatility $\sigma_{\mathrm{IV}}^{\mathrm{mkt}, i}(T_{j}^i, K)$.

\section{Numerical results} \label{s:numerical_results}

\subsection{Calibration results} \label{S:calib_res}

In this section, we illustrate and detail all three calibration steps of the HJM model \eqref{eq:HJM_def}, cf.~Figure \ref{F:calibration_methodology}, onto the German power market, and we postpone to Appendix \ref{ss:ttf_calibration_results} the calibration results obtained on the TTF gas market.

To assess the quality of the calibration we consider the differences between model quantities and those used to calibrate the model: historical correlations and volatilities of rolling futures contracts, variance swap variance term structure and the implied volatility smiles. We also validate the quality of the Kemna-Vorst approximation.
The calibrated model can be used to represent the market, and especially to deduce market quantities that are not quoted: the smile for all monthly contracts especially those which are not quoted, the at-the-money volatility of daily, monthly and quarterly contracts and finally the instantaneous correlations. Finally, we check the coherence between calibrated and interpolated quantities.

\subsubsection{Step 1: Joint historical covariances -- implied VS variances calibration results} \label{ss:calibration_cov_VS_results}

The first calibration step aims at calibrating the $N := 1 + N_{s} + N_{c}$ $L$-$S$-$C$ factors' parameters $\sigma_L$, $(\sigma_{S, i}, \tau_{S, i})_{i \in \{1, \ldots, N_s\}}$, $(\sigma_{C, i}, \tau_{C, i})_{i \in \{1, \ldots, N_c\}}$ as well as their correlation matrix $R$, following the methodology described in Section \ref{s:joint_cov_vs_calib}.

We consider the German power i.e.~DE PW futures market with historical covariances estimated from daily log returns' time series running from January $1^{\text{st}}$ 2023 to July $1^{\text{st}}$ 2024, and the estimated realized futures correlation structure is displayed on the bottom left corner of Figure \ref{F:estimated_correlation_term_structure}.

In order to avoid to have more calibratable parameters than market quantities to calibrate on, one needs to ensure that
\begin{equation} \label{eq:underparametrization_inequality}
    \frac{N(N+3)}{2} - 1 \leq P_{\mathrm{smiles}} + \frac{P_{\mathrm{hist}}\left( P_{\mathrm{hist}}+1 \right)}{2}.
\end{equation}
Furthermore, although season-averaging effects tend to distort Principal Components Analysis (PCA) results, the number of Principal Components (PCs) required to reach a given threshold of explained variance gives some insights on the number of factors to model the futures curve, see for example \cite{koekebakker2005forward}, \cite{Andersen2010}, \cite{gardini2023heath}. Figure \ref{F:pca_de_market} displays the PCA's results for the rolling futures contracts' daily log returns. Note that $5$ PCs are sufficient to explain at least $95\%$ of variance, while $7$ PCs are required to reach $99\%$, which is consistent with 
    \begin{enumerate}
        \item our quality of fit results for the historical calibration described in Figures \ref{F:calibration_fits};
        \item the observations made by \cite{feron2024estimation}[Figure 3] when applying BIC and AIC statistical measurements to the German market, where they found that $5$ factors appears to be a good trade-off between quality of fit and model complexity.
    \end{enumerate}

    \begin{figure}[H]
    \begin{center}
        \includegraphics[width=0.6\linewidth]{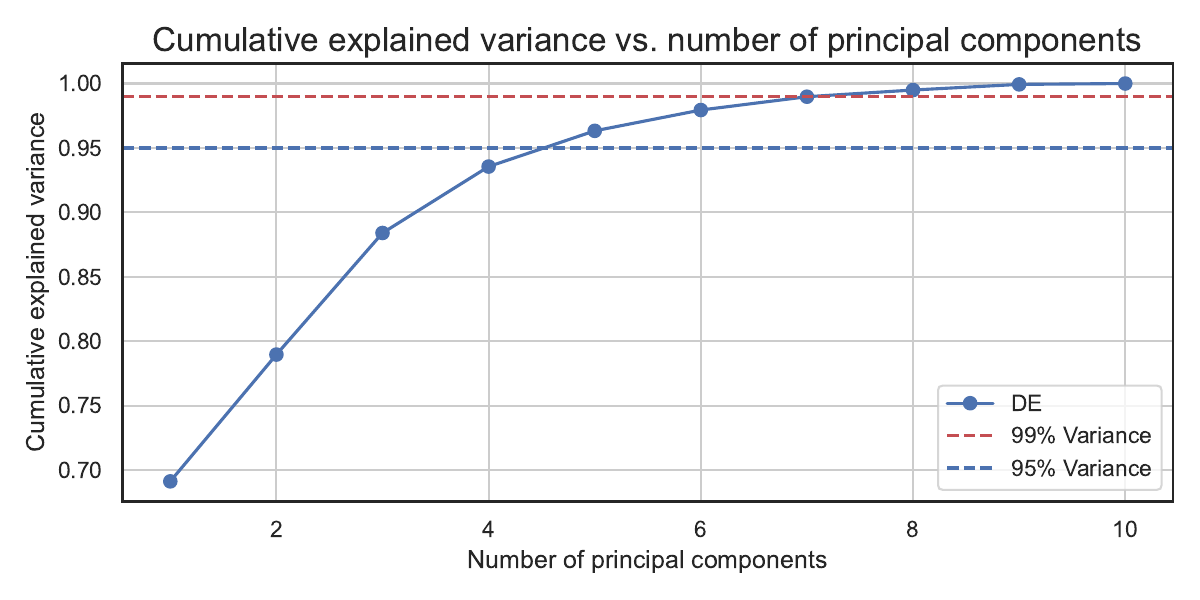}
        \caption{Principal component analysis performed on the raw data of rolling futures contracts reconstituted from the German market from January $1^{\text{st}}$ 2023 to July $1^{\text{st}}$ 2024.}
        \label{F:pca_de_market}
    \end{center}
\end{figure}

The results of the joint historical covariance and VS variances calibration is then summarized in the left matrix of Figure \ref{F:calibration_fits} where, given there is one $L$-factor, we display in a single matrix the losses $J_{1}$ and $J_{2}$ from \eqref{eq:def_loss_covariance_fit}--\eqref{eq:def_loss_vs_variances_fit} respectively for all possible pairs $(N_s,N_c)$ of $S$- and $C$-factors such that the under-parametrization inequality \eqref{eq:underparametrization_inequality} is satisfied for the strict lower part of the matrix, and is strictly violated by one additional factor on the main diagonal, corresponding to over-parametrization. For each cell, we used $100$ random initializations of $a_0$ and kept the best calibration fit. Notice that the overall joint loss $J^{\lambda}$ from \eqref{eq:def_loss}, for $\lambda=0.5$, indeed decreases as the number of factor increases when moving from the lower left angle to the main diagonal of the fit matrix. For a fixed number of factors $N$, we also highlighted by a bold square the model with the best fit along each diagonal.

Furthermore, we show in Figure \ref{F:joint_calibration_fit_plots} the respective quality of fits obtained on the selected rolling contracts' historical volatility and correlation term structures, as well as the fit of the VS volatility term structure. Notice that the fit is not perfect, yet it's possible to additionally fine-tune the parameter $\lambda$ in the joint calibration loss function $J^{\lambda}$ in order to better fit either the historical target covariance values (with $\lambda>0.5$) or the implied VS volatility term structure (with $\lambda<0.5$). In the next step, we calibrate the functions $g$ and $h$ in order to actually fit perfectly the latter in the case the joint calibration has not succeeded in fitting it well enough to fit the smile shapes in the latter calibration step.

We did the same numerical experiments for the historical covariance fit only by taking $\lambda=1$ in the loss $J^{\lambda}$, and display the values of the losses in the right hand-side matrix of Figure \ref{F:calibration_fits}, and the quality of fits in Figure \ref{F:historical_calibration_fit_plots}. Notice that the quality of fit of both the historical volatility and correlation term structures is much better, but the associated VS volatility term structure is systematically under-estimated by such historical calibration.

\begin{figure}[H]
    \centering
    \adjustbox{trim={0.1cm 0.1cm 0.1cm 0.1cm},clip}{%
        \includegraphics[width=0.5\textwidth]{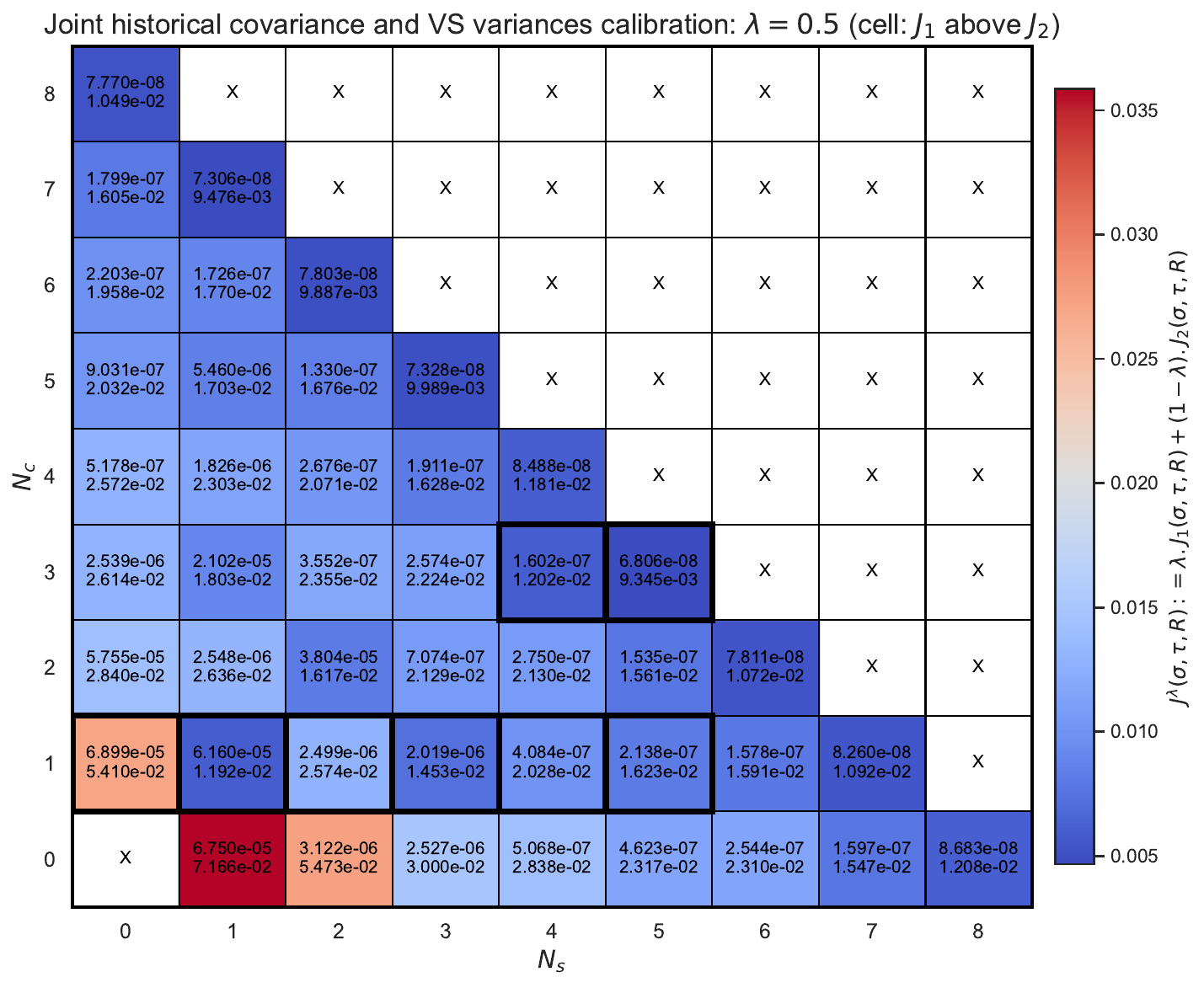}
    }
    \adjustbox{trim={0.1cm 0.1cm 0.1cm 0.1cm},clip}{%
        \includegraphics[width=0.48\textwidth]{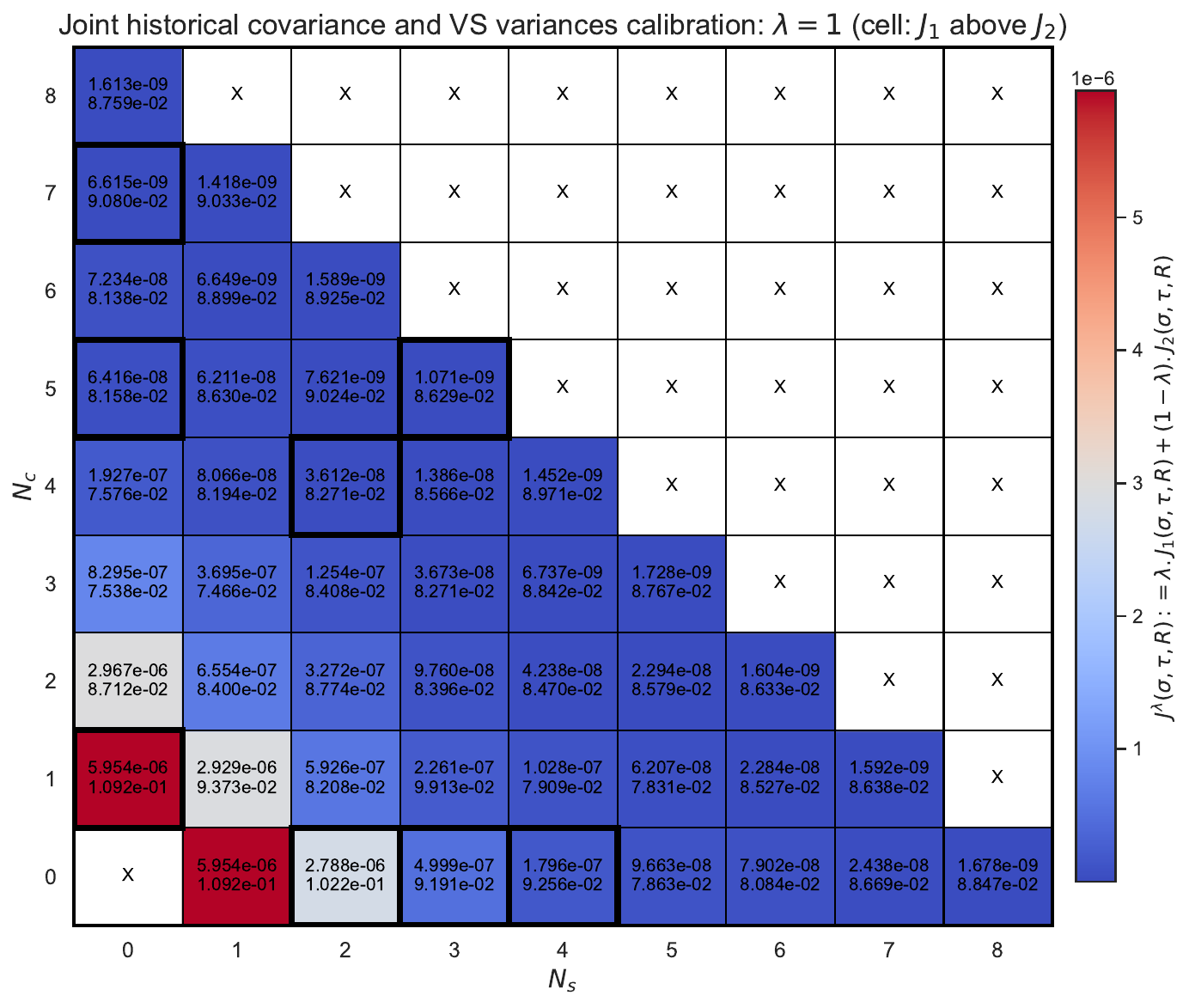}
    }
    \caption{Calibration results of step 1 for the German power market for all possible $L$-$S$-$C$ model specifications up to height factors. Left (resp. Right): Best fits for the joint historical covariance and VS variances calibration (resp. the historical covariance calibration) when taking $\lambda=0.5$ (resp. $\lambda=1$) in the loss function $J^{\lambda}$ \eqref{eq:def_loss}, for $100$ random initializations of $\tau(a_0)$, varying the number of $S$ and $C$-factors until violation of the under-parametrization inequality \eqref{eq:underparametrization_inequality}. In each cell, the above number is the value of $J_{1}$ from \eqref{eq:def_loss_covariance_fit} and below is the value of $J_{2}$ from \eqref{eq:def_loss_vs_variances_fit} (which we do not calibrate on in the right-hand side results). For each matrix, we put bold squares on the best fits along each diagonal where the number of risk factors is constant.}
    \label{F:calibration_fits}
\end{figure}

\begin{figure}[H]
    \centering
    \adjustbox{trim={0cm 0.2cm 0cm 0.2cm},clip}{%
        \includegraphics[width=0.9\textwidth,angle=0]{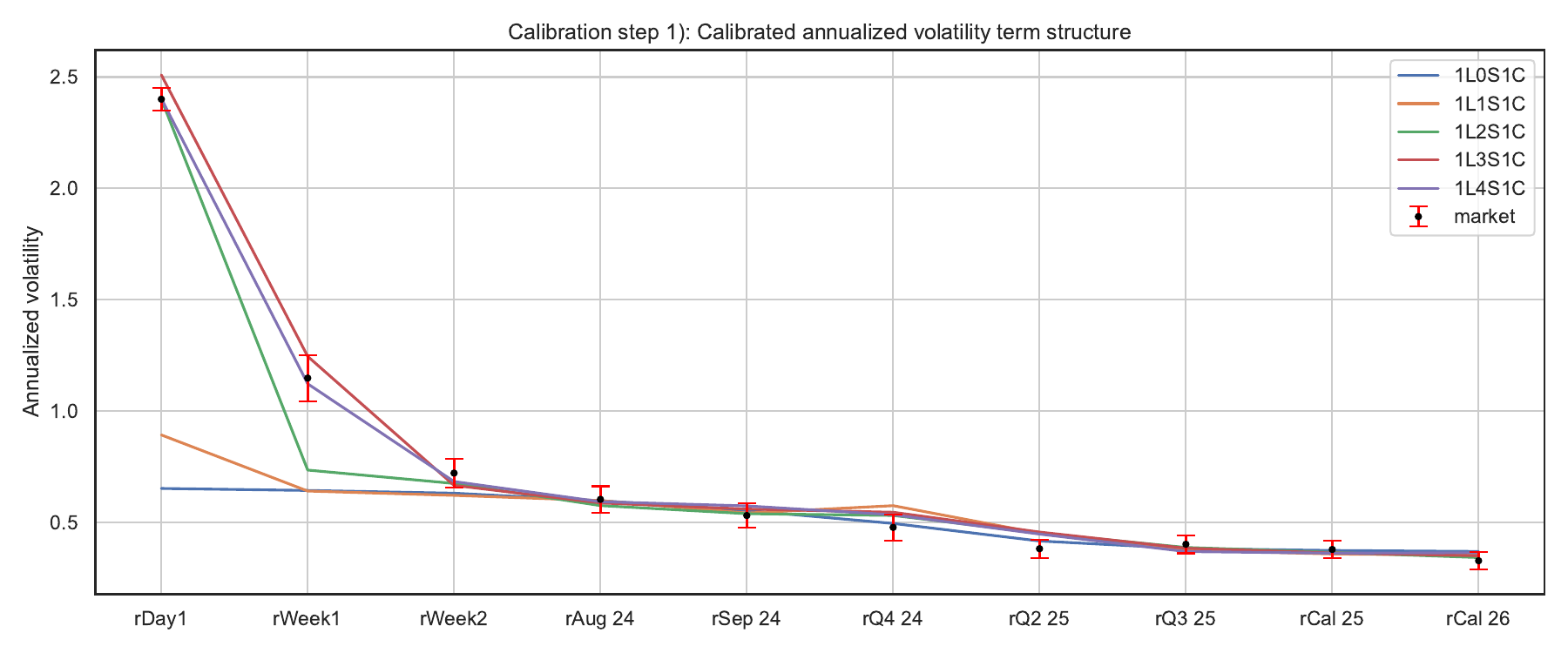}
    }

    \vspace{0.5cm} 

    \adjustbox{trim={0cm 0.2cm 0cm 0.2cm},clip}{%
        \includegraphics[width=0.9\textwidth,angle=0]{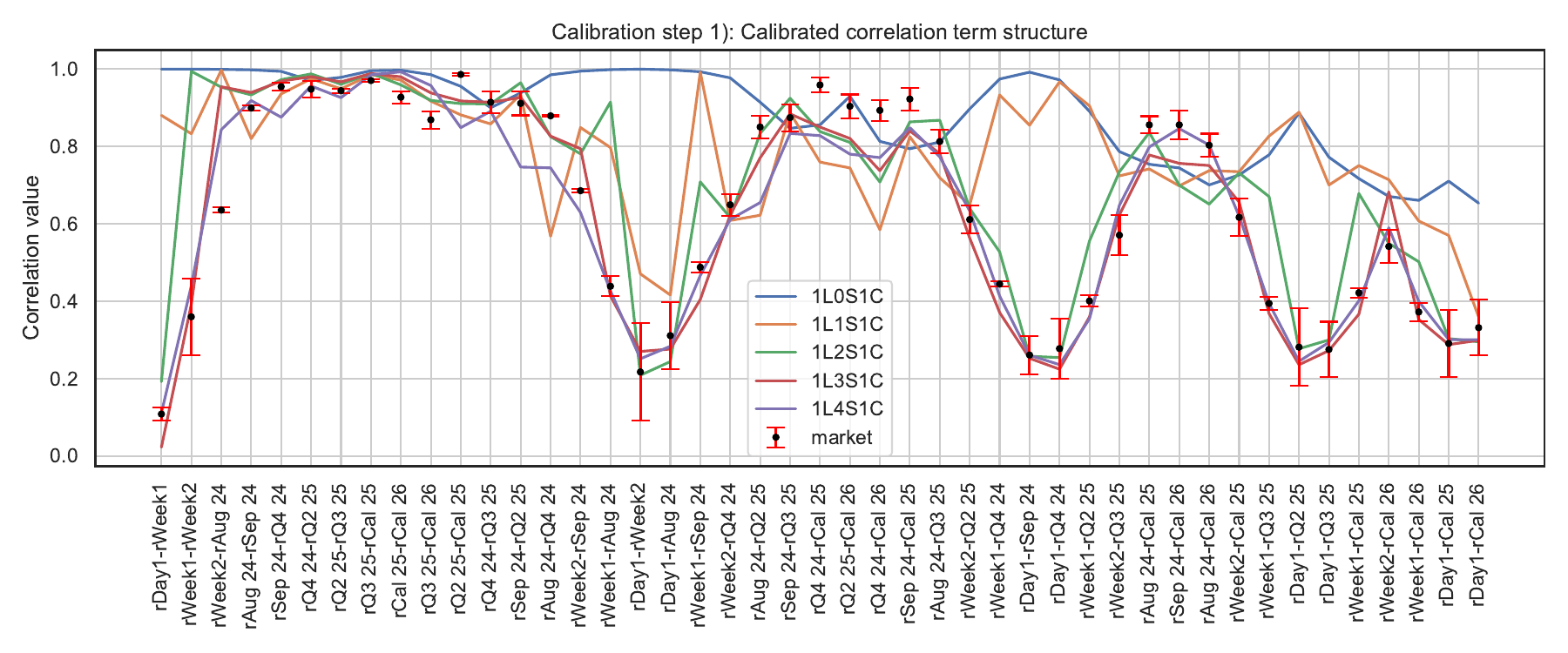}
    }

    \vspace{0.5cm} 

    \adjustbox{trim={0cm 0.2cm 0cm 0.2cm},clip}{%
        \includegraphics[width=0.9\textwidth,angle=0]{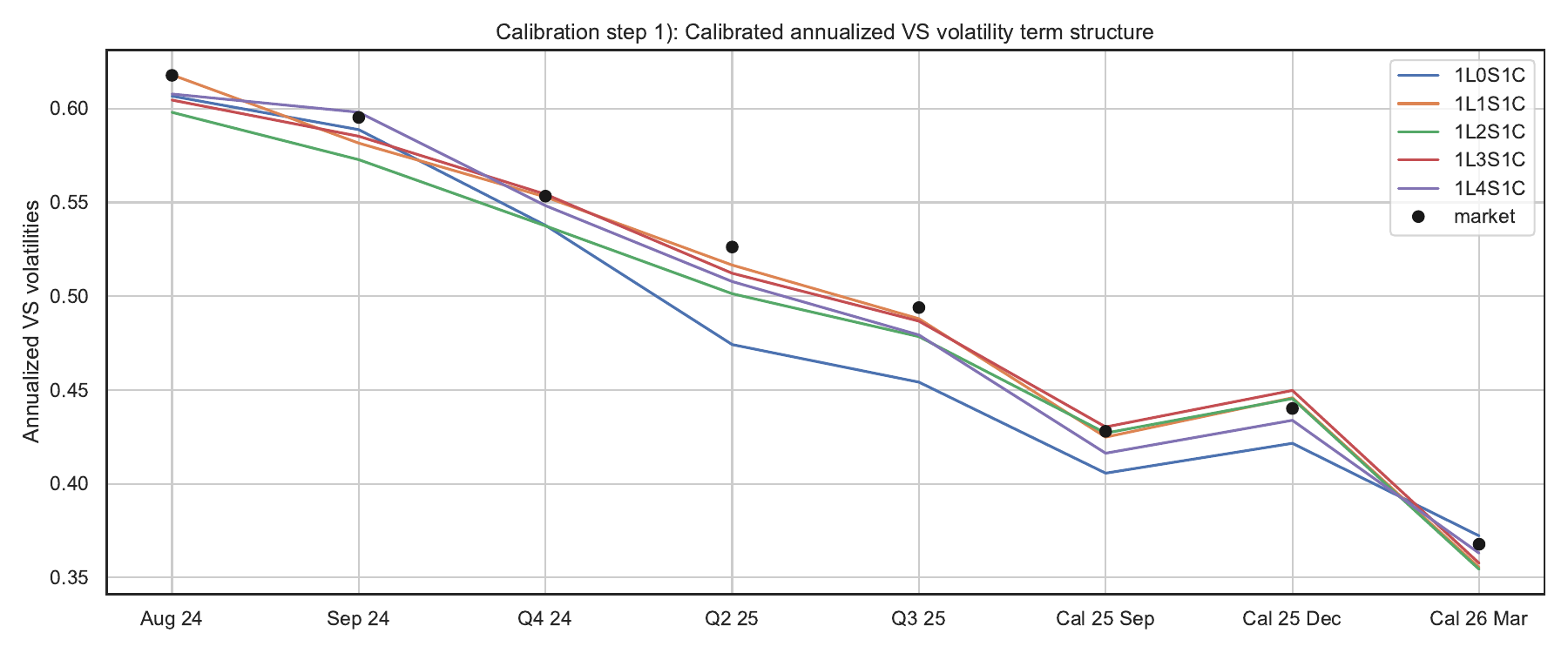}
    }

    \caption{Calibration results of step 1 for the German power market for the best $L$-$S$-$C$ model specifications up to six factors. The joint calibration fits i.e.~$\lambda=0.5$ in the loss $J^{\lambda}$ \eqref{eq:def_loss}, are displayed with respect to (top) the historical realized volatility term structure, (middle) the historical futures correlation structure and (bottom) the VS variance term structure.}
    \label{F:joint_calibration_fit_plots}
\end{figure}

\begin{figure}[H]
    \centering
    \adjustbox{trim={0cm 0.2cm 0cm 0.2cm},clip}{%
        \includegraphics[width=0.9\textwidth,angle=0]{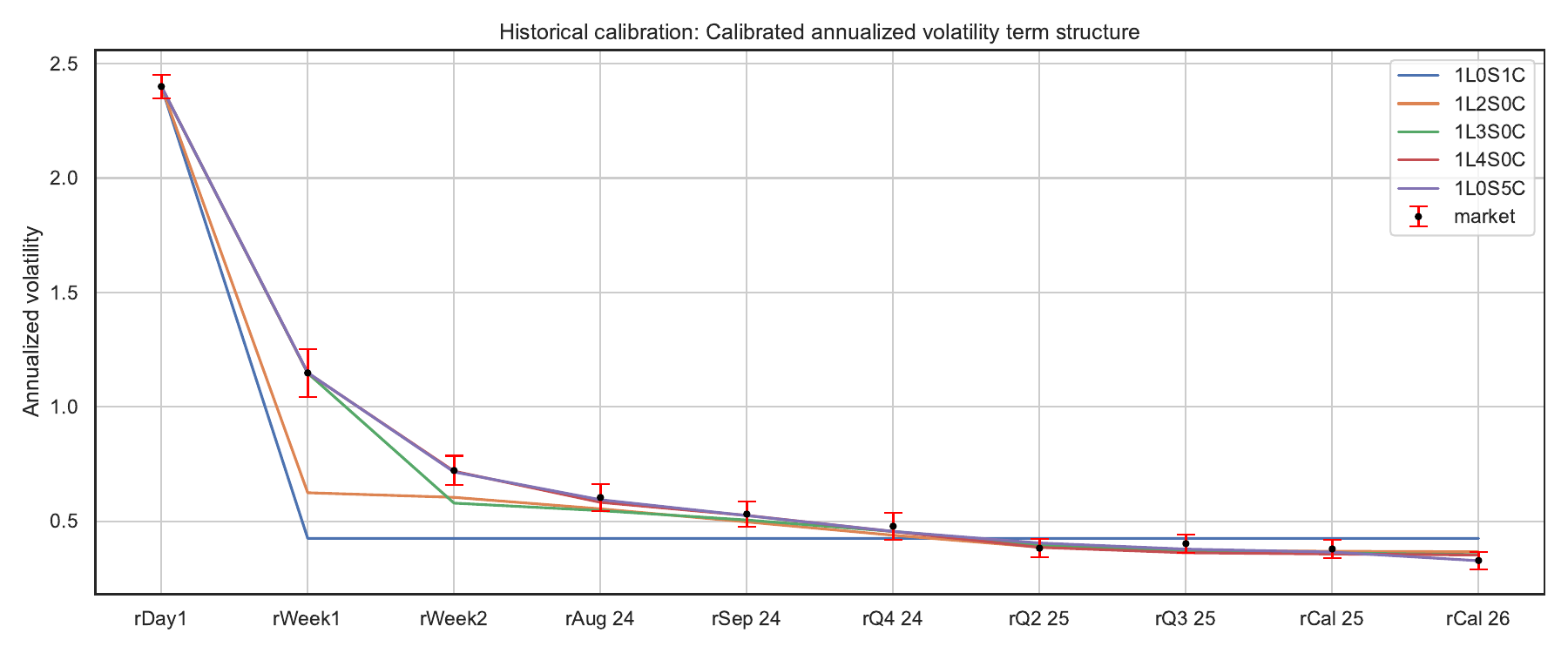}
    }

    \vspace{0.5cm} 

    \adjustbox{trim={0cm 0.2cm 0cm 0.2cm},clip}{%
        \includegraphics[width=0.9\textwidth,angle=0]{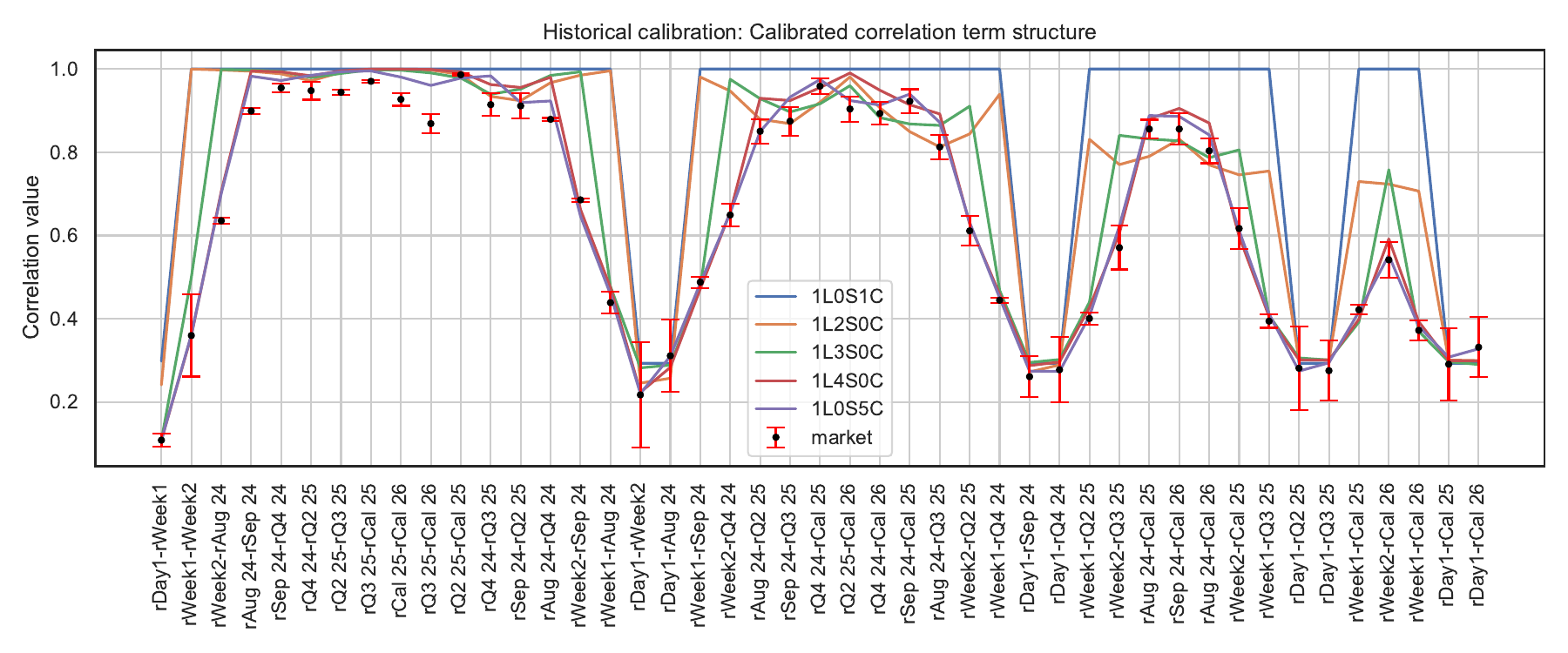}
    }

    \vspace{0.5cm} 

    \adjustbox{trim={0cm 0.2cm 0cm 0.2cm},clip}{%
        \includegraphics[width=0.9\textwidth,angle=0]{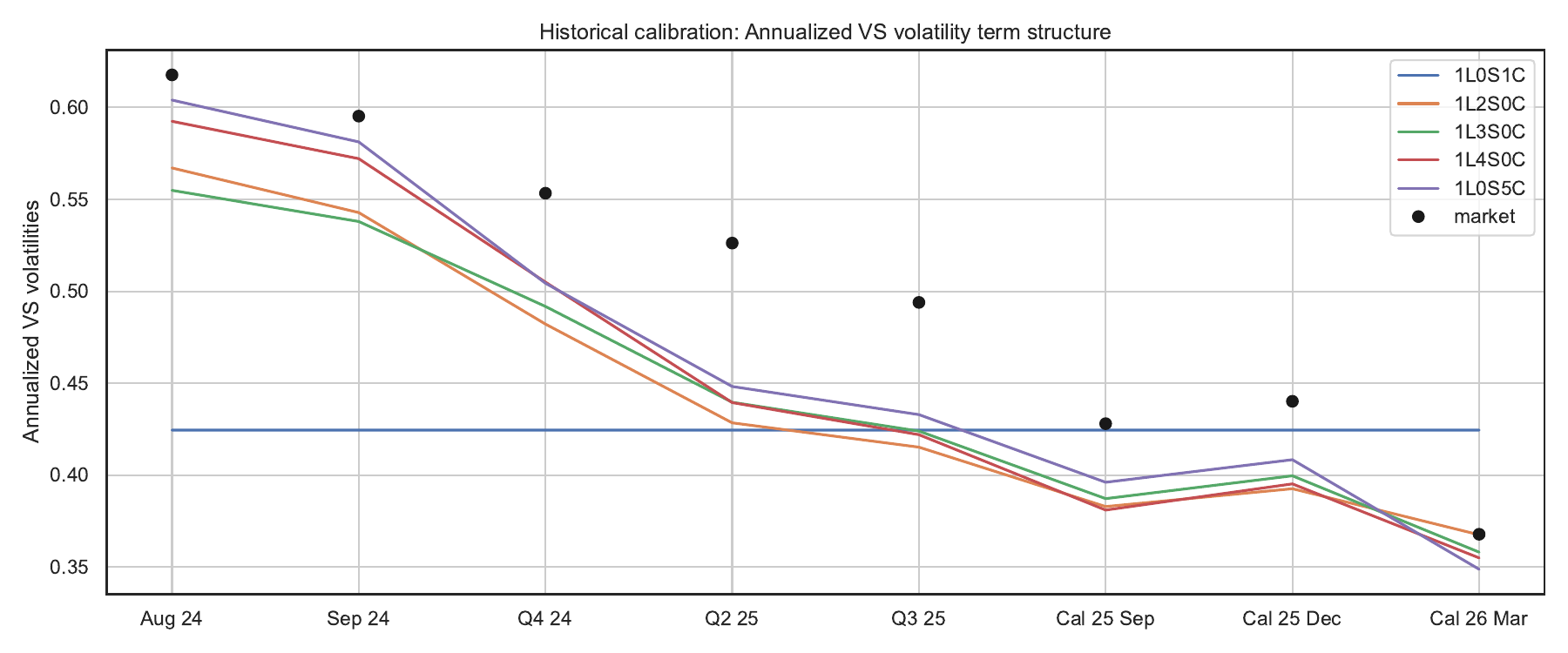}
    }

    \caption{Historical covariance calibration fits i.e.~$J_{2}$ \eqref{eq:def_loss_vs_variances_fit} is ignored by taking $\lambda=1$ in the loss $J^{\lambda}$ \eqref{eq:def_loss}, of the best up-to-six-factors models whose loss functions are highlighted in the right-hand side matrix of Figure \ref{F:calibration_fits} to (top) the historical realized volatility term structure, (middle) the historical futures correlation structure and (bottom) the VS variance term structure.}
    \label{F:historical_calibration_fit_plots}
\end{figure}

For the implied calibration, we choose the model $1L3S1C$ as it gives a good trade-off between parsimony and the joint calibration quality, and display its calibrated parameters in Figure \ref{F:sigma_tau_and_correlation_1L_3S_1C}. Notice the $L$-factor indeed captures the long-term volatility level, with $\sigma_{L}$ lying in-between the annualized volatilities of long-term contracts rCal 25 and rCal 26, respectively equal to $0.3784$ and $0.3277$ while the $C$ factor is placed at the beginning of the curve to capture the short-term volatility behavior in the case of $1L3S1C$. For conciseness, the calibrated parameters of the other models mentioned in Figure \ref{F:calibration_fits} are not displayed.

\begin{figure}[H]
    \centering
    \adjustbox{trim={0.1cm 0.1cm 0.1cm 0.1cm},clip}{%
        \includegraphics[width=0.48\textwidth,angle=0]{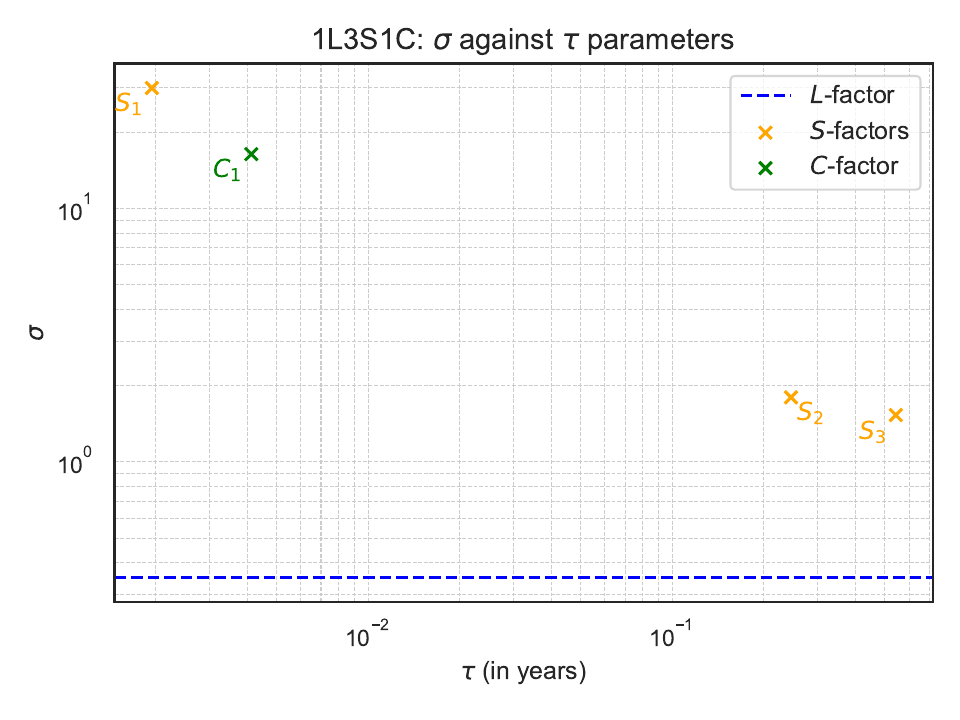}
    }
    \hspace{0.02\textwidth} 
    \adjustbox{trim={0.1cm 0.1cm 0.1cm 0.1cm},clip}{%
        \includegraphics[width=0.45\textwidth,angle=0]{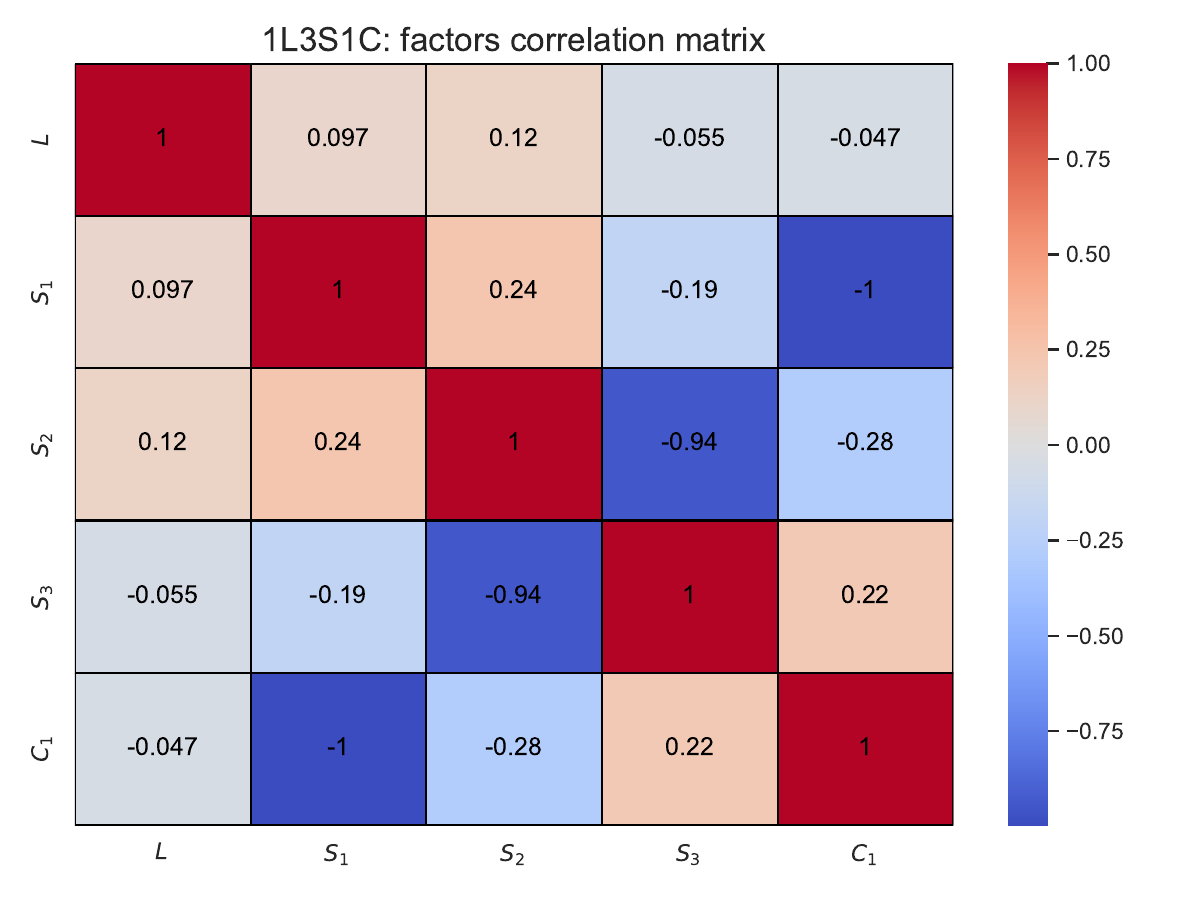}
    }

    \caption{Plots of $\left(\sigma_{i}\right)_{i}$ against $\left(\tau_{i}\right)_{i}$ parameters (left) and of the factors' correlation matrix (right) for the $1L3S1C$ model, when taking $\lambda=0.5$ in the loss function $J^{\lambda}$ \eqref{eq:def_loss}. Model parameters are, in $L$-$S$-$C$ order: $\left(\sigma_{i}\right)_{i} := [0.3499, 29.9849, 1.7992, 1.5325, 16.4317]$, $\left(\tau_{i}\right)_{i} := [0.0019, 0.2454, 0.5428, 0.0041]$.}
    \label{F:sigma_tau_and_correlation_1L_3S_1C}
\end{figure}

It is worth mentioning that the choice of the historical volatility factors does not impact significantly the results of the following calibration steps as long as the fit of the VS term structure is satisfactory. Otherwise, the calibrated functions $g$ and $h$ may deviate significantly from $1$ degrading the model's interpolation capabilities.

\subsubsection{Step 2: Term structure exact calibration correction}

{This part addresses the correction of the VS volatility term structure with the functions $g$ and $h$ following the methodology proposed in Section \ref{section:gh_calib}.}

In order to avoid jumps of the day-ahead contracts volatility, both functions $g$ and $h$ are smoothed during calibration. At each step of the fixed point algorithm, we update $g$ and $h$ as piece-wise constant functions and then interpolate it using the monotonicity preserving algorithm PCHIP of \cite{fritsch1984method}. Despite the smoother results, the algorithm is more time-consuming due to the numerical optimization involved in the calibration. On the left-hand side of Figure \ref{fig:calibrated_g} below, we plot the results of the term structure calibration of piece-wise constant functions $g$ and $h$, and on the right-hand side, same results with smoothing interpolation. In both cases, the fixed-point algorithm converges in approximately $5$–$10$ iterations, providing an almost perfect fit for the VS volatilities {shown in Figure~\ref{fig:vs_ts_fit}}. 

We also note that thanks to the calibration of the VS volatilities at the first step of the calibration, the functions $g$ and $h$ are needed only to slightly correct the volatility level, so that they remain close to $1$ for the months following the observation date. For longer maturities, the function $h$ deviates significantly from $1$ since the VS volatilities of smiles Cal 26 Dec and Cal 27 Dec were not calibrated at the first calibration step due to the poor quality of covariance estimations for contracts with long time to delivery.


\begin{figure}[H]
    \centering
    \subfigure[\centering Piece-wise constant functions $h$ and $g$.]{{\includegraphics[width=0.45\linewidth]{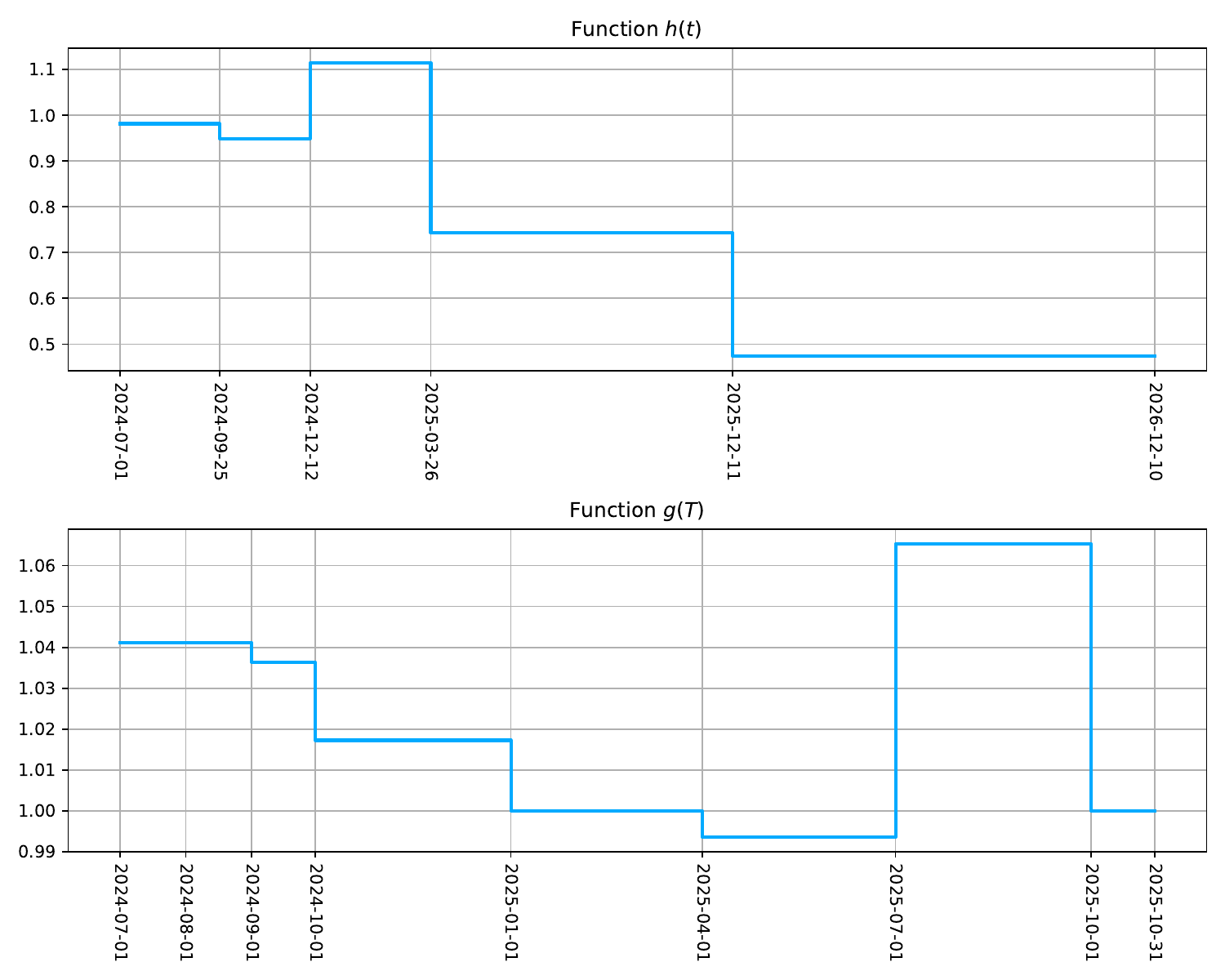} }}%
    \qquad
    \subfigure[\centering Smoothed functions $h$ and $g$.]{{\includegraphics[width=0.45\linewidth]{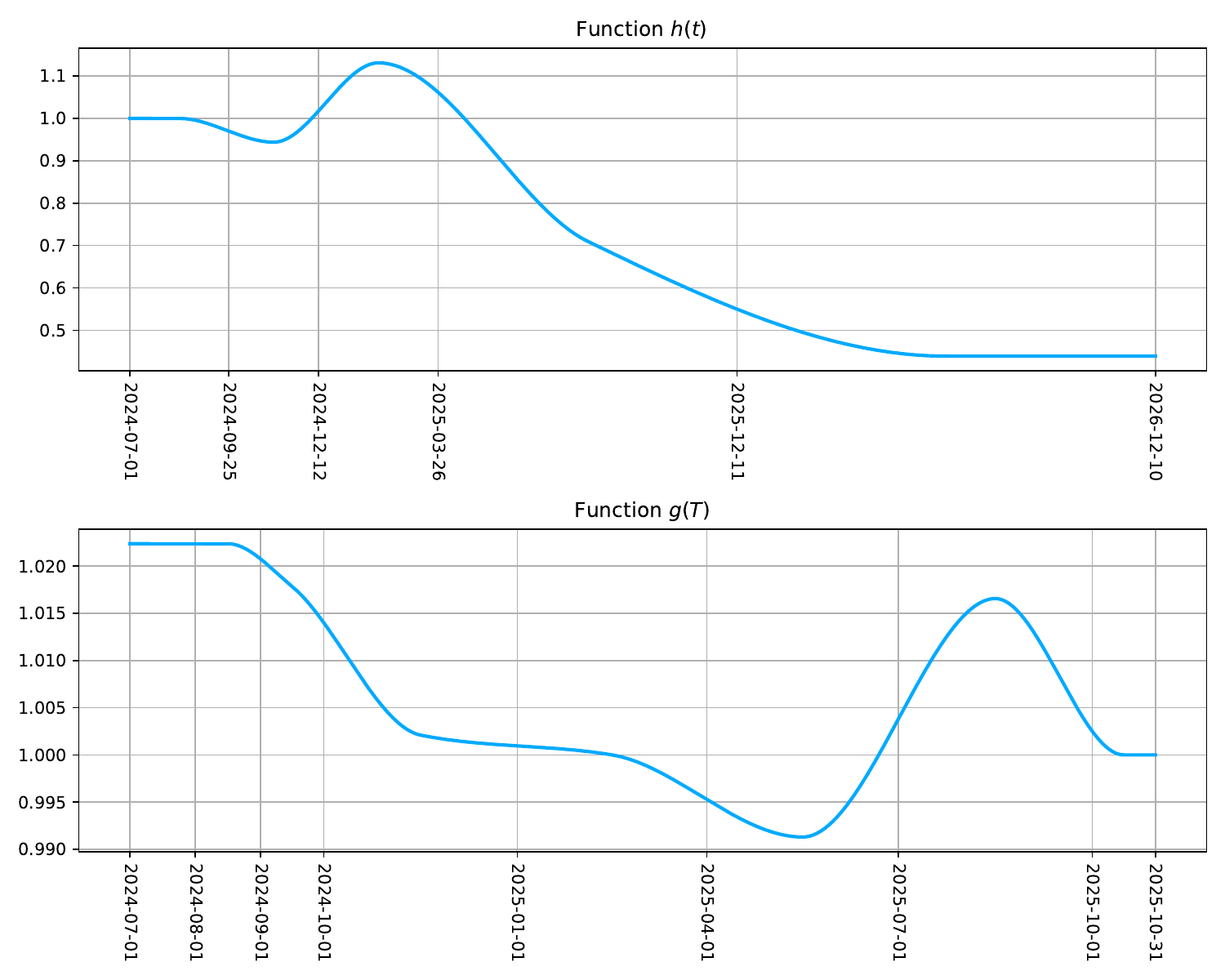}}}%
    \caption{Calibrated functions $\textcolor{black}{h}$ and $\textcolor{black}{g}$ with and without smoothing ((a) and (b) correspondingly).}%
    \label{fig:calibrated_g}%
\end{figure}

\begin{figure}[H]
\begin{center}    \includegraphics[width=0.6\linewidth]{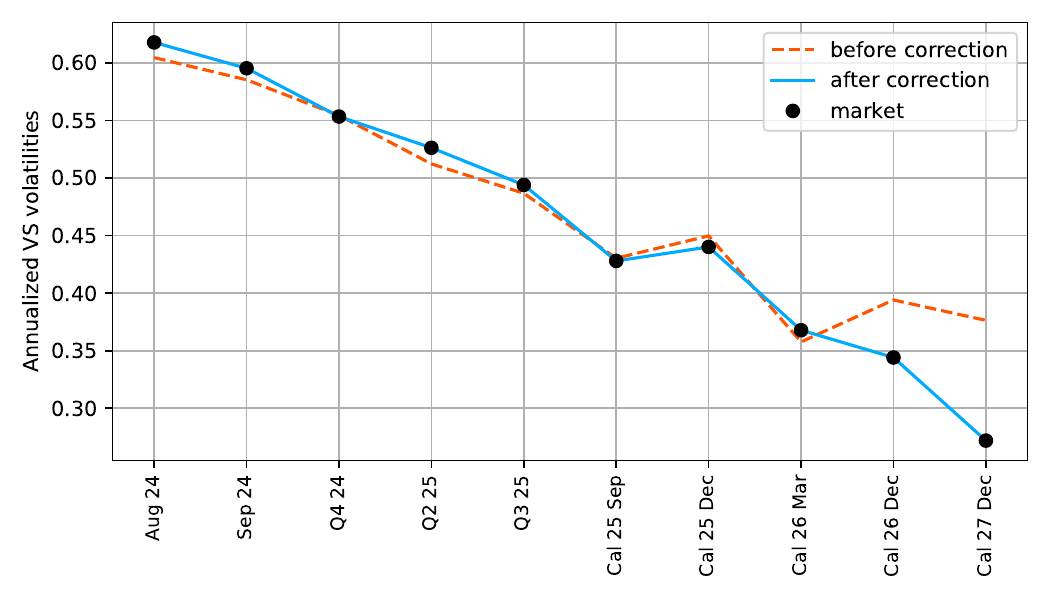}
    \caption{VS term structure fit before the correction (orange) and after the correction (blue).}
    \label{fig:vs_ts_fit}
\end{center}
\end{figure}

\subsubsection{Step 3: Implied smile shapes calibration}

{Finally, we calibrate the parameters $(c_i, x_i)_{i \in \{1, \ldots, M\}}$ of the stochastic variance process $V$, as well as the correlations $(\tilde\rho_i)_{i \in \{1, \ldots, N\}}$, using the procedure described in Section~\ref{sect:smile_shape_calib}.}

For the stochastic variance component calibration, we choose a model with $M = 3$ pseudo-factors in \eqref{eq:V_def} which allow us to cover different volatility timescales and consistently achieve an acceptable fit on various calibration sets. The advantage of using multiple pseudo-factors is demonstrated in Appendix \ref{section:attainable}, where we compare the quality of fit between the Heston model and the Lifted Heston model.

The calibrated parameters of the stochastic volatility are provided in Table \ref{calibrated_params}. For the calibrated values, the loss function defined by \eqref{eq:optimization_problem} equals $0.001754$. It is interesting to notice that the first mean-reversion coefficient is very close to zero and corresponds to very long mean-reversion periods, while the two others correspond approximately to the timescales of one month and 2.5 weeks respectively.

\begin{table}[H]
\begin{center}
\begin{tabular}{ c c} 
 \hline
 Parameter & Calibrated value\\
 \hline \hline
 \rule{0pt}{2ex} $c$ & $(0.492,\, 0.68,\, 2.79)$ \\
 
 \rule{0pt}{2ex} $x$ &  $(4.6\cdot 10^{-6},\,  9.712,\, 20.249)$\\ 
 
 \rule{0pt}{2ex} $\tilde\rho$ & $(0.648,\, -0.516, \, 0.16,\, -0.148,\,  0.541)$\\ 
 
 \hline
\end{tabular}
\end{center}
    \caption{Calibrated lifted Heston model parameters}
    \label{calibrated_params}
\end{table}

The implied volatility smiles in the calibrated model with smoothed $g$ and $h$ functions, are shown in Figure \ref{fig:iv_smiles_results}.  We also plot the bid-ask spread equal to 5\% as a realistic and even conservative proxy to the market spread observed in the power market. Thus, the plot demonstrates sufficiently high quality of fit for the model to be used for practical needs.

\begin{figure}[H]
\begin{center}    \includegraphics[width=1\linewidth]{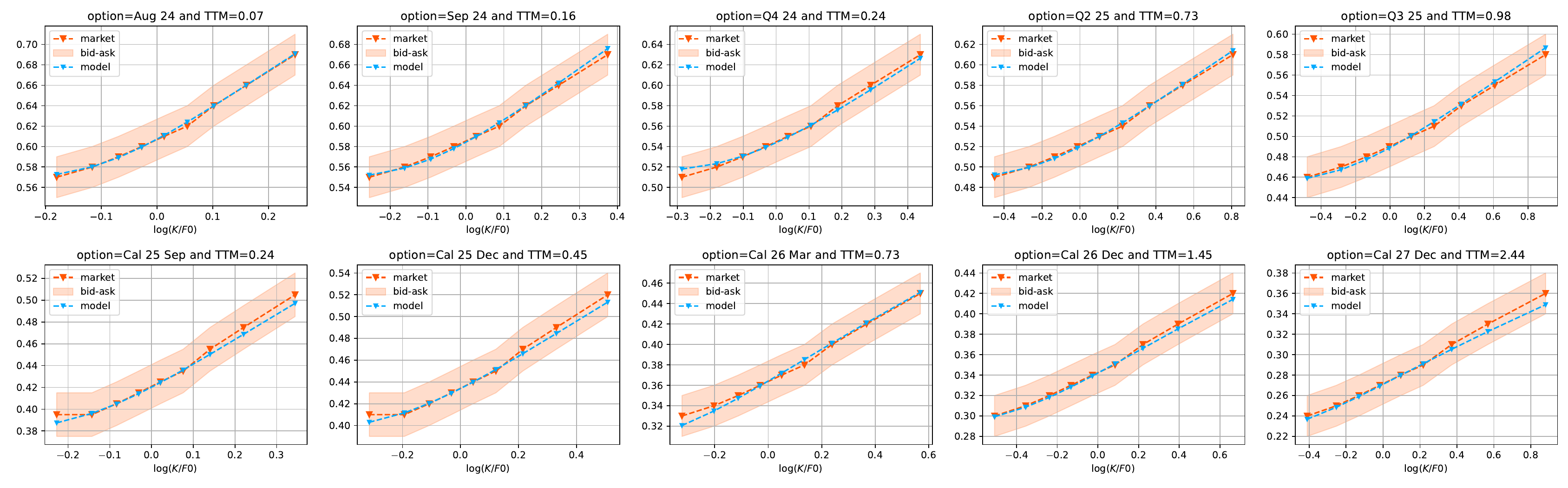}
    \caption{IV smiles in the calibrated model (blue) and market IV smiles (orange).}
    \label{fig:iv_smiles_results}
\end{center}
\end{figure}

\subsection{A posteriori validation of model approximations}\label{ss:validation_by_MC_of_kv_futures}

As discussed in Section \ref{section:KV_approx_presentation}, the use of the Kemna-Vorst (KV) approximated model futures $\tilde{F}$ from \eqref{eq:KV_def} leads to potential arbitrage opportunities for futures contracts with overlapping deliveries, e.g.~between monthly futures and the resulting quarter composed by such delivery months. However, we show in the next subsections that the futures prices' trajectories, as well as the futures correlation structures and implied volatility smiles are extremely close in the approximated model $\tilde{F}$ and in the exact model future $F$ from \eqref{eq:futures_contract_def}.

From the forward rate definition \eqref{eq:HJM_def}, the exact model is given explicitly by
\begin{align}
    F_t(T_s, T_e) = \frac{1}{T_e - T_s} \int_{T_s}^{T_e} f(0,T) \exp\Bigg( &-\frac{g^2(T)}{2} \int_0^t h^2(s)V_s \sigma^\top (s,T)  R \sigma (s,T)\, \d s \\ &+ g(T) \int_0^t h(s)\sqrt{V_s}  \sigma^\top (s,T)\, dW_s \Bigg) \d T.
\end{align}
Numerically, we consider a one-day discretization step, which corresponds in practice to the futures contracts with the shortest delivery period quoted in the futures market
\begin{align}
    F_t(T_s, T_e) & \approx \frac{\delta}{T_e - T_s} \sum_{i=1}^{N-1} f(t, T_i)\\
    T_s = T_1 < ... < T_N &= T_e \text{ and }\ T_{i+1} - T_i = \delta = \dfrac{1}{365}, \ i = 1, \ldots, N-1. 
\end{align}

\begin{remark}
    In real energy markets, the smallest possible delivery intervals $\Delta$ are typically 15 minutes, 30 minutes, or one hour. Consequently, the price of any futures contract $F_t(T_s, T_e)$ can be expressed as the average
    $$
    F_t(T_s, T_e) = \frac{\Delta}{T_e - T_s} \sum_{i=1}^{(T_e - T_s) / \Delta} f(t, T_s + i\Delta),
    $$
    of the prices $f(t, T_s + i\Delta)$ of all contracts with the smallest granularity $\Delta$ included in $[T_s, T_e]$. This implies that the futures curve is inherently discrete, whereas the HJM model provides a continuous approximation. By approximating the contract $F(T_s, T_e)$ using formula \eqref{eq:futures_discrete_approx}, we effectively perform an inverse approximation, representing a continuous curve by a discrete sum. For numerical feasibility in the subsequent Monte Carlo simulations, the value $\delta = \frac{1}{365}$, corresponding to one day, was used instead of the finer granularity of one hour or one quarter.
\end{remark}

Since the exact daily contracts are computed using a Monte Carlo scheme, the exact model becomes significantly slower in terms of pricing and simulation time than the approximated one. 

Besides the KV approximation, it is important to notice that the introduction of the function $g \not= 1$ prevents the futures volatility from being stationary. However, if $g \approx 1$, which is the case in our model as the variance swap volatilities are mostly calibrated at the first calibration step, then the futures contracts dynamics remains almost stationary. Furthermore, the function $g$ impacts the instantaneous correlations between absolute contracts. These correlations are affected by the KV approximation as well. In the following numerical experiments, we demonstrate that neither KV approximation nor the function $g$ have a significant impact on the futures correlation structure.

\subsubsection{Sample path trajectories errors}

As an illustration of the quality of this approximation, we provide in Figure~\ref{fig:comparison_between_kv_and_no_kv} the trajectories of the exact futures contracts $F$  and their approximations $\tilde{F}$, simulated with the same random numbers, for three different futures contracts: Sep 24, Q1 25 and Cal 25. Recall that $F$ is defined in \eqref{eq:futures_contract_def} with $f(t, T)$ explicitly given by \eqref{eq:f_explicit}. 
The mean $L^2$ distance on $[0, \theta]$ is defined as a sample Root Mean Squared Error (RMSE) over $m \in \mathbb{N}^{*}$ sample trajectories by
$$
\sqrt{\frac{1}{m} \sum_{i=1}^{m} \|F(\omega_{i})-\widetilde{F}(\omega_{i})\|_{L^2(\theta)}^2} := \sqrt{\frac{1}{m} \sum_{i=1}^{m} \int_{0}^{\theta} \left( F_{t}(\omega_{i}) - \widetilde{F}_{t}(\omega_{i}) \right)^{2} \d t} \approx \sqrt{\frac{1}{M} \sum_{i=1}^{m} \frac{\theta}{N_{\theta}} \sum_{j=1}^{N_{\theta}} \left( F_{t_{j}}(\omega_{i}) - \widetilde{F}_{t_{j}}(\omega_{i}) \right)^{2}},
$$
where $N_{\theta}$ is the number of steps in the daily subdivision of $[0, \theta]$. The estimated distance with $\theta = T_s$
between the approximated and exact trajectories estimated over $10^4$ simulations equals $0.0602$, $0.0897$ and $0.1898$ for three futures contracts correspondingly. We observed that this error depends both on the interval length $\theta$ as well as on the length of the delivery period. Both dependencies are non-linear, the first one is strictly increasing in $\theta$. As for second one, the error is greater for small (up to one-two months)  and for large intervals (two years or more) delivery periods. The worst-case RMSE numerically reaches $0.5\%$ of the initial futures value $F_{0}$.

\begin{figure}[H]
    \centering
    \includegraphics[width=0.9\linewidth]{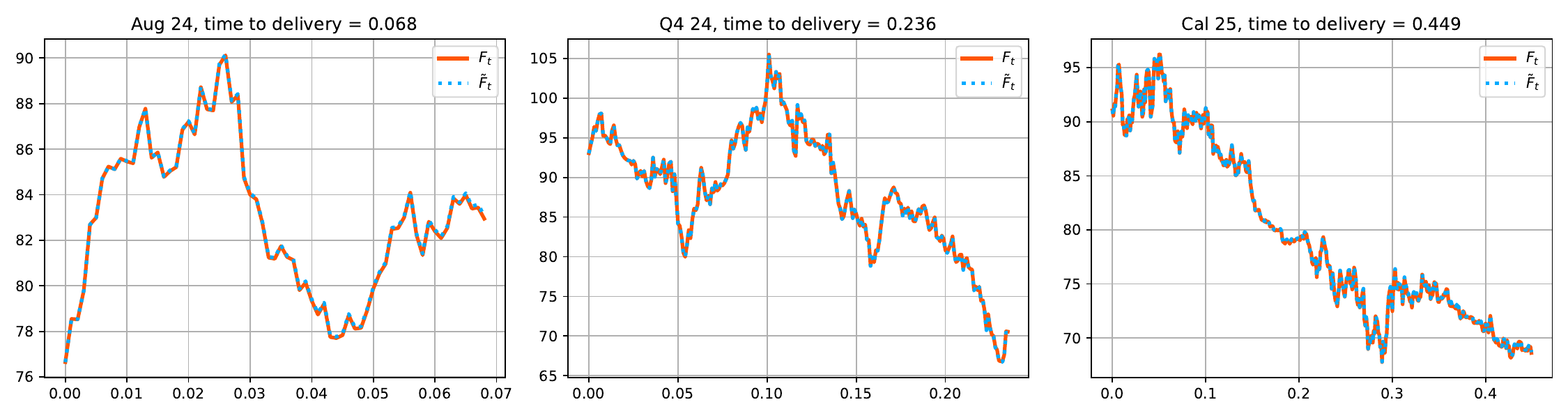}
    \caption{Sample paths of $\widetilde F(T_s, T_e)$ and $F(T_s, T_e)$ for the same random numbers for three futures contracts: monthly contract Sep 24 delivering during September 2024, quarterly contract Q1 25 delivering during the first quarter of 2025, and yearly (or calendar) contract Cal 25 delivering during the whole year 2025. In the figures, time to delivery is given in years.}
\label{fig:comparison_between_kv_and_no_kv}
\end{figure}

\subsubsection{Instantaneous correlations}\label{section:num_correls}

In this section, we assess numerically the impact of the function $g$ on the futures contracts' instantaneous correlations. Indeed,
unlike the function $h$ and the stochastic variance component $V$ which only modify the futures' volatility term structure but not their correlations, the function ${g}$ also impacts the correlations as it changes the approximated futures contract volatilities $\Sigma_{.}$ { defined by \eqref{eq:futures_contracts_volatility}}. Indeed, for $i, j \in \{ 1, \ldots, P \}$, the instantaneous correlation between the approximated futures contracts $\widetilde F_{.}^{i} := \widetilde F_{.}^{i} (T_s^i,\, T_e^i)$ and $\widetilde F_{.}^{j} := \widetilde F_{.}^{j}(T_s^j,\, T_e^j)$ obtained just after the first calibration step, {and implied by the historical futures correlation structure of rolling futures}, is defined by 
\begin{equation}\label{eq:inst_correl_hist}
     \widetilde\rho_{ij}^{\mathrm{hist}}(t) = \dfrac{\left\langle \d \log \widetilde{F}^i,\, \d \log \widetilde{F}^j\right\rangle_t}{\sqrt{\langle \d \log \widetilde{F}^i\rangle_t}\sqrt{\langle \d \log \widetilde{F}^j\rangle_t}} = \dfrac{\hat\Sigma^\top_{t}(T_s^i,\, T_e^i)R\hat\Sigma_{t}(T_s^j,\, T_e^j)}{\sqrt{\hat\Sigma_{t}^\top(T_s^i,\, T_e^i)R\hat
    \Sigma_{t}(T_s^i,\, T_e^i)}\sqrt{\hat\Sigma_{t}^\top(T_s^j,\, T_e^j)R\hat\Sigma_{t}(T_s^j,\, T_e^j)}}, \quad 0 \leq t \leq \min(T_s^i, T_s^j),
\end{equation}
where $\hat\Sigma_{.}(T_s^j,\, T_e^j)$ was defined by (\ref{eq:Sigma_hat}). However, after the second calibration step, the futures contract volatility is given by $\Sigma_{.}$, and the instantaneous correlation equals
\begin{equation}\label{eq:inst_correl_model}
     \widetilde\rho_{ij}(t) = \dfrac{\Sigma^\top_{t}(T_s^i,\, T_e^i)R\Sigma_{t}(T_s^j,\, T_e^j)}{\sqrt{\Sigma_{t}^\top(T_s^i,\, T_e^i)R\Sigma_{t}(T_s^i,\, T_e^i)}\sqrt{\Sigma_{t}^\top(T_s^j,\, T_e^j)R\Sigma_{t}(T_s^j,\, T_e^j)}}, \quad 0 \leq t \leq \min(T_s^i, T_s^j).
\end{equation}
Note that these two values coincide if and only if $\textcolor{black}{g(T)}$ is constant on $[T_s^i,\, T_e^i]$ and on $[T_s^j,\, T_e^j]$, as in this case, $\Sigma_{.}(T_s^i,\, T_e^i) = \textcolor{black}{\bar g_{i}}\hat\Sigma_{.}(T_s^i,\, T_e^i)$ and $\Sigma_{.}(T_s^j,\, T_e^j) = \textcolor{black}{\bar g_{j}}\hat\Sigma_{.}(T_s^j,\, T_e^j)$, where $\bar g_{i}$ and $\bar g_{j}$ denote the values of the function $g$ on the intervals $[T_s^i,\, T_e^i]$ and $[T_s^j,\, T_e^j]$ correspondingly.

However, we can show numerically that the difference between $\widetilde\rho_{ij}$ and $\widetilde\rho_{ij}^{\mathrm{hist}}$ is not significant from a practical point of view even if $g$ is not constant. Namely, we examine this difference for the model calibrated in Section \ref{S:calib_res} with the function $g$ given by Figure \ref{fig:calibrated_g}(b).
The figure \ref{fig:contract_corr_mat} shows the distances $\|\widetilde\rho_{ij} - \widetilde\rho_{ij}^{\mathrm{hist}}\|_\infty$ for each pair of contracts for the calibration set. In the worst case, the difference is less than 0.1\%. Thus, despite the introduction of the function $g$, the correlations in the calibrated model are still very close to the ones given by the historical calibration.

\begin{figure}[H]
    \begin{center}
        \includegraphics[width=0.5\linewidth]{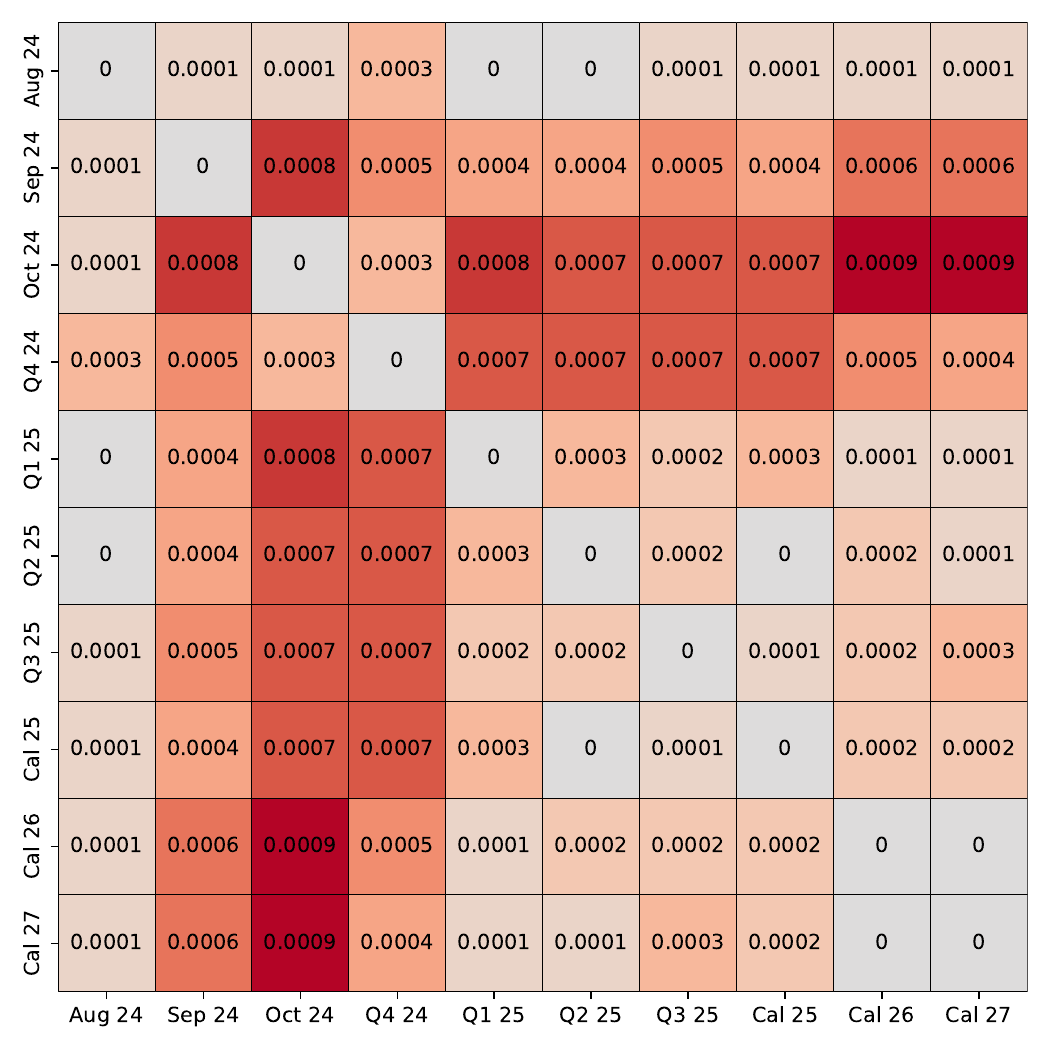}
    \end{center}
    \caption{Differences $\|\widetilde\rho_{ij} - \widetilde\rho_{ij}^{\mathrm{hist}}\|_\infty$ between historically calibrated and model instantaneous correlations.}
    \label{fig:contract_corr_mat}
\end{figure}

\subsubsection{Quality of the Kemna--Vorst approximation on the instantaneous correlations} \label{section:fwd_correls}
As we have just seen, the Kemna--Vorst approximation allows for an explicit computation of instantaneous correlations between futures contracts in \eqref{eq:inst_correl_model}. However, these quantities do not coincide with the correlations between exact futures contracts
$$
\rho_{ij}(t) := \dfrac{\left\langle \d \log {F}^i,\, \d \log {F}^j\right\rangle_t}{\sqrt{\langle \d \log {F}^i\rangle_t}\sqrt{\langle \d \log {F}^j\rangle_t}}, \quad 0 \leq t \leq \min(T_s^i, T_s^j),
$$
where $F^i$ and $F^j$ follow \eqref{eq:futures_contract_def} for $i, j \in \{ 1, \ldots, P \}$. More precisely
\begin{equation}
    \rho_{ij}(t) = \dfrac{{\Sigma^{\mathrm{exact}}_{t}}^\top(T_s^i, T_e^i)R{\Sigma^{\mathrm{exact}}_{t}}(T_s^j, T_e^j)}{\sqrt{{\Sigma^{\mathrm{exact}}_{t}}^\top(T_s^i, T_e^i)R
    {\Sigma^{\mathrm{exact}}_{t}}(T_s^i, T_e^i)}\sqrt{{\Sigma^{\mathrm{exact}}_{t}}^\top(T_s^j, T_e^j)R{\Sigma^{\mathrm{exact}}_{t}}(T_s^j, T_e^j)}}, \quad 0 \leq t \leq \min(T_s^i, T_s^j),
\end{equation}
with $\Sigma^{\mathrm{exact}}_{\cdot}(T_s, T_e) = \displaystyle \frac{\int_{T_s}^{T_e}g(T)\sigma(\cdot,T)f(\cdot,T) \d T}{\int_{T_s}^{T_e}f(\cdot,T) \d T}$, such that the exact correlation is a stochastic process, since $f$ is stochastic. However, as shown numerically in Figure \ref{fig:true_model_corr_diff}, that the $L^\infty$-difference between $\mathbb{E}[\rho_{ij}]$ and $\widetilde\rho_{ij}$ lies within two standard deviations of $\rho_{ij}$, which makes the approximated correlation \eqref{eq:inst_correl_model} a valid and tractable deterministic approximation of the exact instantaneous correlation.

In conclusion, one can control the difference error between the instantaneous correlations $\widetilde\rho_{ij}^{\mathrm{hist}}$ obtained by the first calibration step from historical futures' correlation structure and the fully calibrated exact model instantaneous correlations $\rho_{ij}$ using the triangle inequality
$$ 
\|\widetilde\rho_{ij}^{\mathrm{hist}} - \mathbb{E}[\rho_{ij}] \|_\infty = \|(\widetilde\rho_{ij}^{\mathrm{hist}} - \widetilde\rho_{ij}) - (\mathbb{E}[\rho_{ij}] - \widetilde\rho_{ij}) \|_\infty \leq \|\widetilde\rho_{ij}^{\mathrm{hist}} - \widetilde\rho_{ij}\|_\infty + \|\widetilde\rho_{ij} - \mathbb{E}[\rho_{ij}]\|_\infty,
$$
where the right-hand error quantities are respectively given in Figures \ref{fig:contract_corr_mat} and \ref{fig:true_model_corr_diff} across all liquid contracts. Such numerical measurements illustrate the overall consistency of the instantaneous correlations in the approximated model with respect to the exact one.

\begin{figure}[H]
    \begin{center}
    \includegraphics[width=0.8\linewidth]{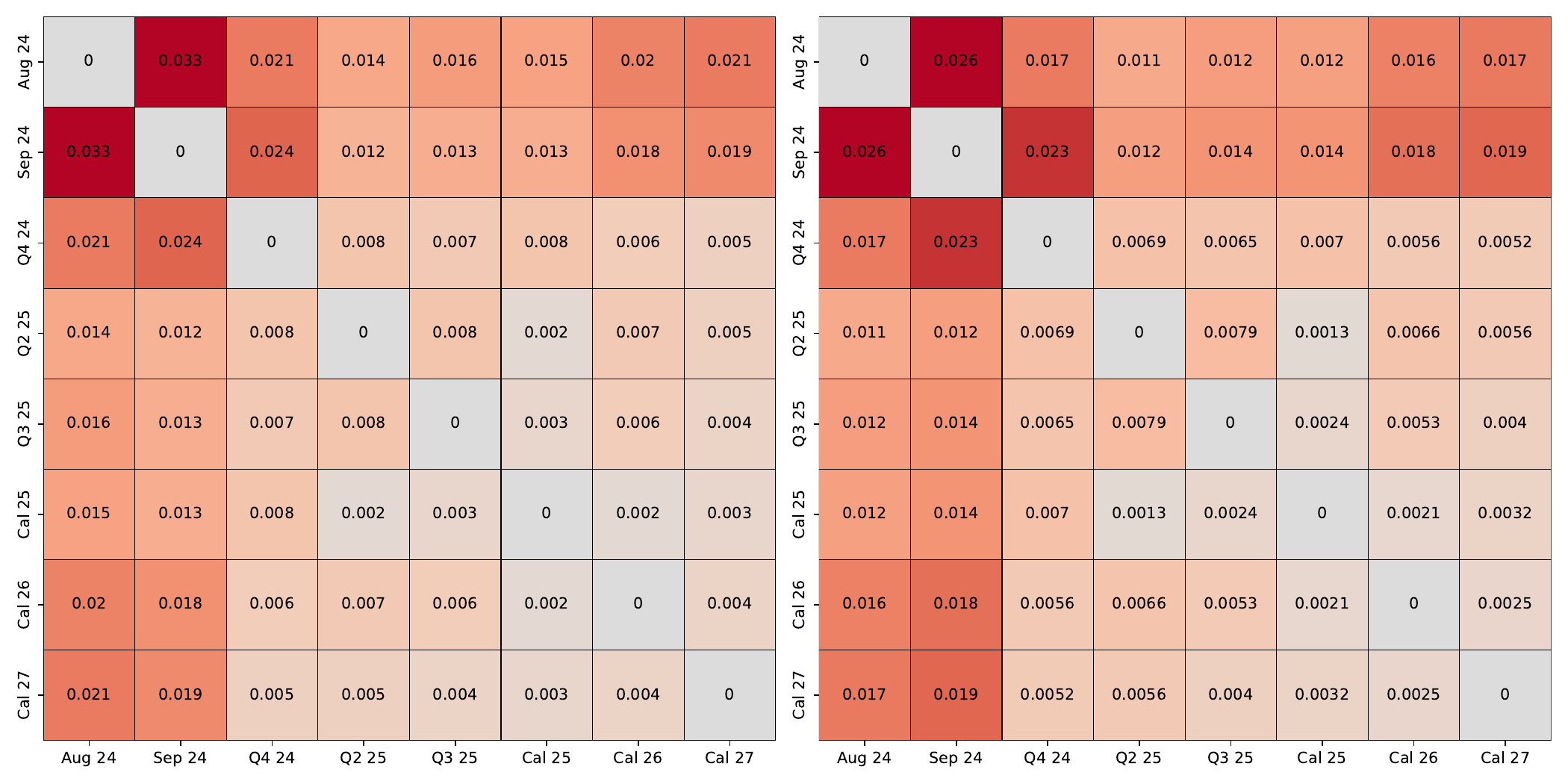}
    \end{center}
    \caption{Distance $\|\widetilde\rho_{ij} - \mathbb{E}[\rho_{ij}]\|_\infty$ (on the left) and its standard deviation (on the right) in the calibrated model.}
    \label{fig:true_model_corr_diff}
\end{figure}

\subsubsection{Quality of the Kemna--Vorst approximation on the smiles} \label{section:kemna_vorst} 

Our next goal is to show that the Kemna--Vorst approximation \eqref{eq:KV_def} provides precise enough pricing results. Since neither closed formula, nor an SDE for the futures price dynamics are available, the option prices can be computed only with the Monte Carlo scheme presented in Appendix \ref{section:monte_carlo_scheme} to generate the trajectories of $F$ and $\widetilde F$.

In order to visualize the approximation error, we benchmark the implied volatility smiles for the contracts Sep 24, Q4 24, and Cal 25 Dec with the calibrated model parameters given in Section \ref{s:implied_calib}. The same random numbers were used for both Monte Carlo, $10^5$ trajectories were simulated and we consider a daily grid for the instantaneous futures used to compute $F$. In the worst case, the difference in implied volatility between the Monte Carlo estimators is less than 0.2\%.

Moreover, we compare the values of the smile calibration loss function introduced further in \eqref{eq:optimization_problem} for the calibrated model. The loss function value in the approximated model is given by $0.001754$, while the loss function in the exact model is $0.001617$. 

It allows us finally to conclude that the Kemna-Vorst approximation is a valid tractable approximation for calibration of the exact model which then can be used for pricing and hedging purposes.

\begin{figure}[H]
\centering   
 \includegraphics[width=1\linewidth]{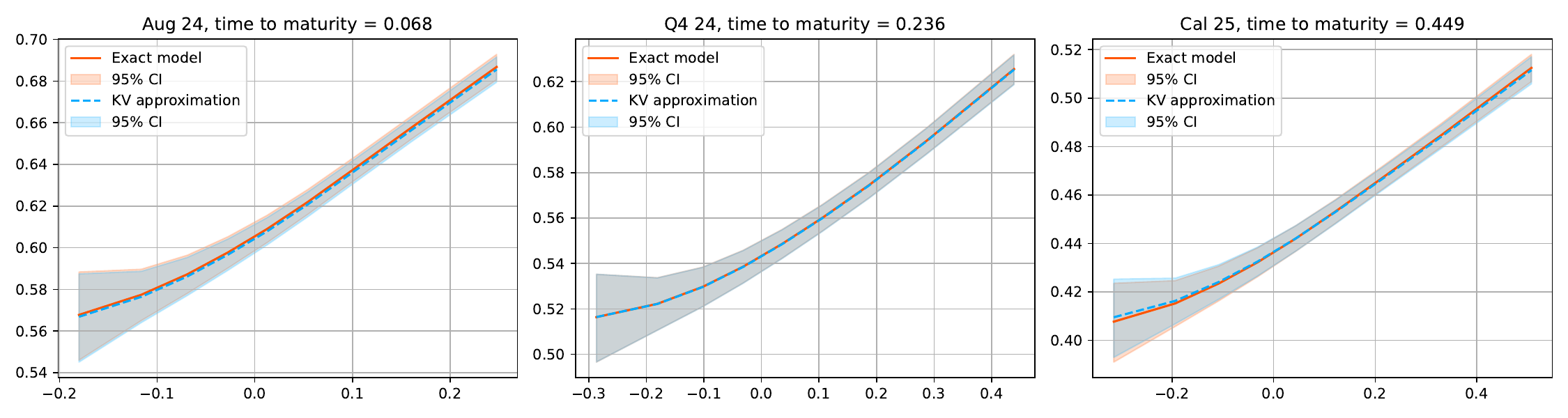}
 \caption{IV smiles generated with the same random numbers ($10^5$ simulations)}
\end{figure}

\subsection{Volatility hypercube extrapolation in the calibrated model}

In this section, we address the question of the implied volatility hypercube extrapolation once the model is fully calibrated and provide the volatility term structure induced by the model.

To illustrate the volatility term structure, we compute the ATM volatilities of options on daily contracts expiring in the day preceding the day of the delivery, as well as the ATM volatilities of (generated) monthly, quarterly and calendar contracts expiring three days before the delivery period start. The range of considered days cover the days from July 3, 2024 to July 3, 2025. The resulting  ATM volatilities of are shown in Figure \ref{fig:calibrated_vol_ts} (a), (b).

We also generate the implied volatility smiles for monthly futures contracts corresponding to the months between August 2024 and December 2026. The options maturity dates are taken equal to the 25-th day of the month preceding the delivery period. The generated smiles are shown in Figure \ref{fig:calibrated_vol_ts} (c).

\begin{figure}[H]
    \centering
    \subfigure[\centering ATM daily volatilities.]{{\includegraphics[width=0.4\linewidth]{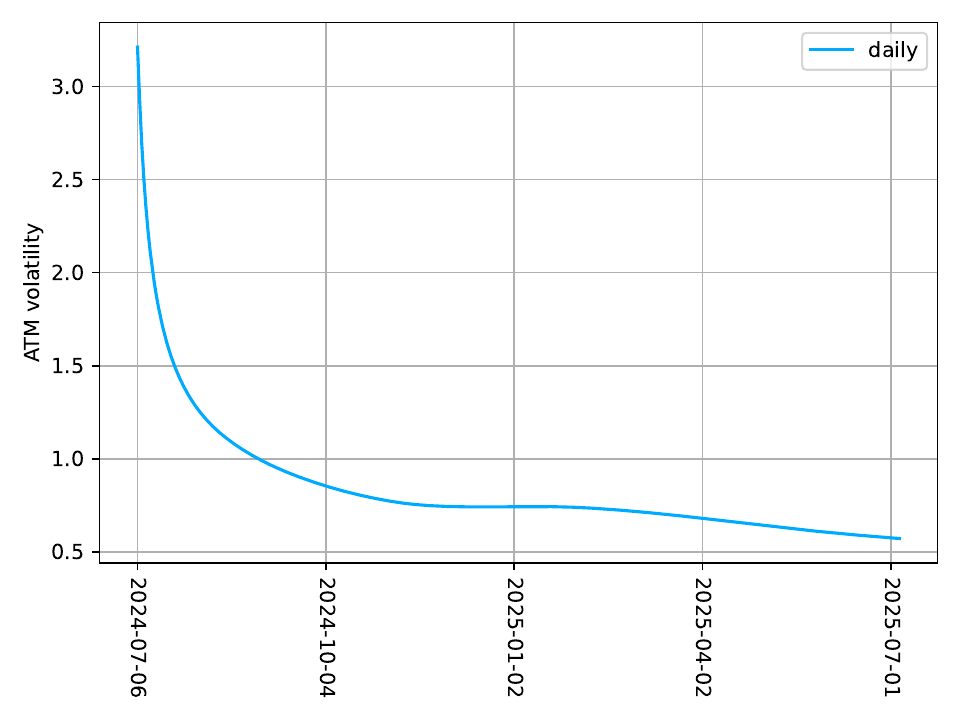}}}%
    \qquad
    \subfigure[\centering ATM monthly, quarterly and cal volatilities.]{{\includegraphics[width=0.4\linewidth]{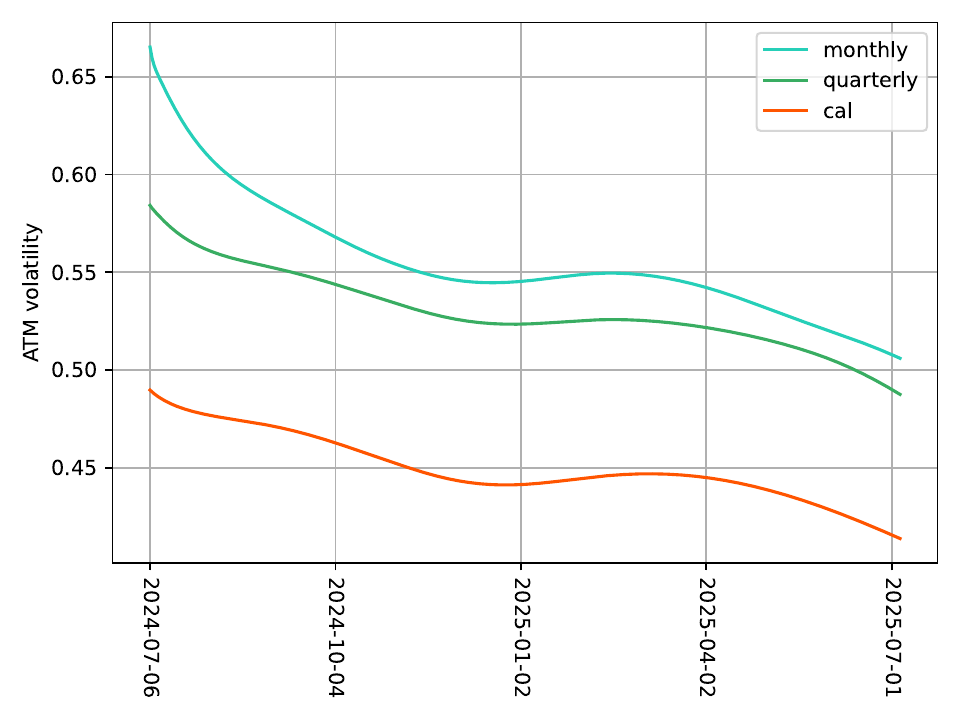}}}%
    \qquad
    \subfigure[\centering IV smiles from August 2024 to December 2026.]{{\includegraphics[width=0.5\linewidth]{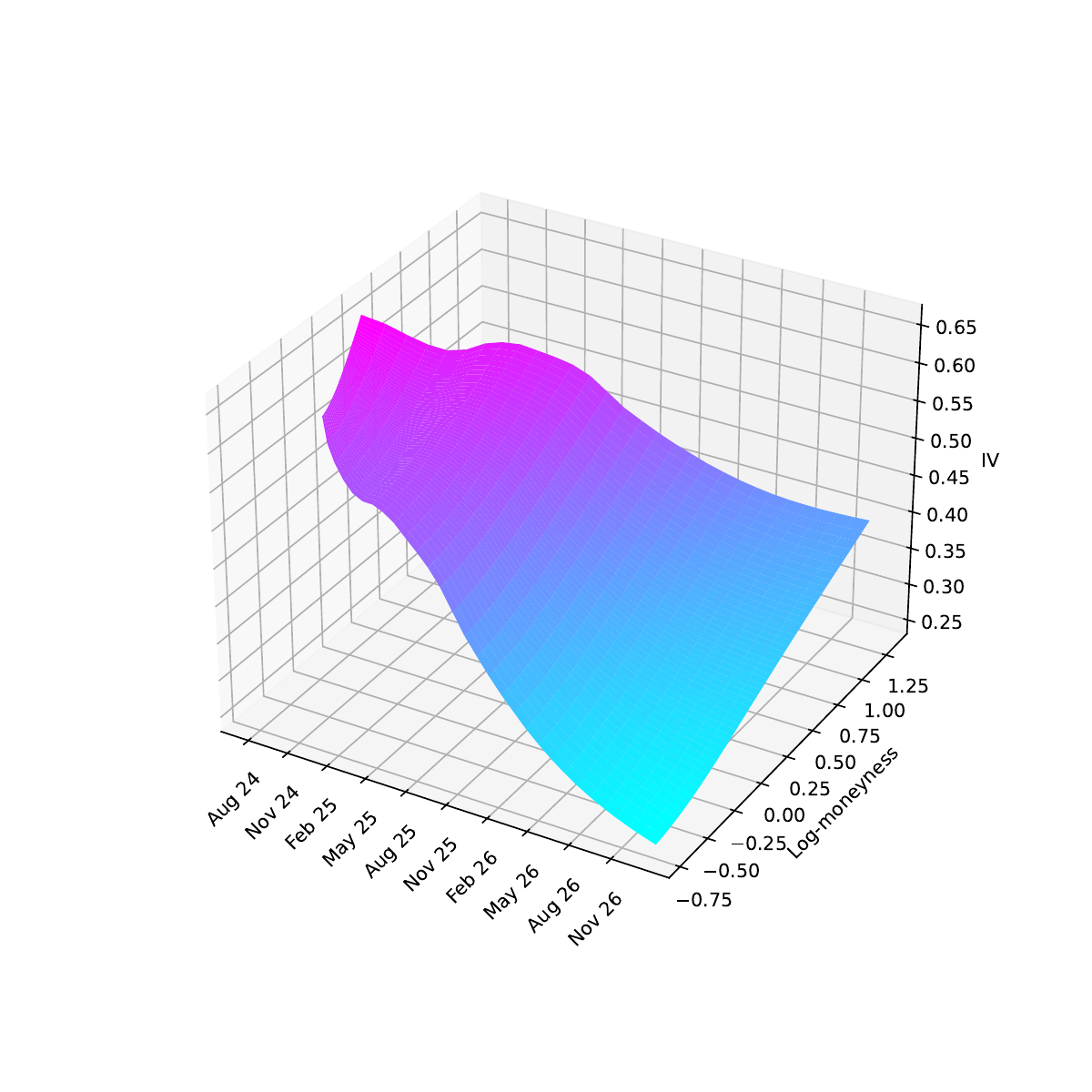} }}%
    \caption{ATM implied volatilities of daily contracts (a) and of monthly, quarterly and calendar contracts (b), and implied volatility smiles for monthly contracts (c) generated by the calibrated model.}%
    \label{fig:calibrated_vol_ts}%
\end{figure}

We highlight that the initial calibration of the variance swap volatilities at the first calibration step prevents the model from over-fitting the VS volatility term structure when calibrating the functions $g$ and $h$ in step 2). This makes possible a reasonable interpolation and extrapolation of the implied volatilities for the contracts and maturities not included in the calibration set. In Figure \ref{fig:extrapolation}, we plot the integrated variance for monthly (green) and quarterly (blue) contracts being calibrated (circles) as well as the ones not present in the calibration set and generated within the calibrated model (crosses). We observe that the calibrated model produces consistent interpolation and extrapolation values, and thus capable of completing the volatility hypercube.

\begin{figure}[H]
\begin{center}    \includegraphics[width=0.5\linewidth]{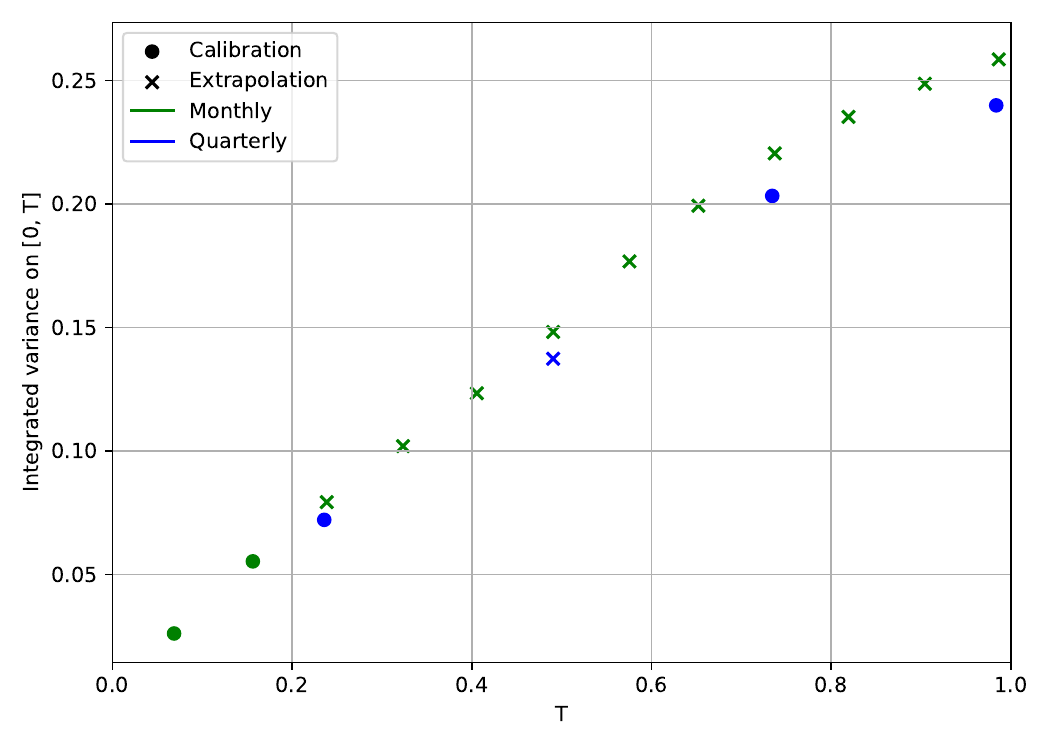}
    \caption{Integrated variance for monthly (green) and quarterly (blue) contracts as function of time to maturity $T$ within one year. Circles correspond to the calibrated contracts  Aug 24, Sep 24, Q4 24, Q2 25, Q3 25, while crosses correspond to the contracts represented by the calibrated HJM model (monthly contracts from Oct 24 to Jul 25 and Q1 25).}
    \label{fig:extrapolation}
\end{center}
\end{figure}

\appendix

\section{More calibration results}
\subsection{Attainable smiles}\label{section:attainable}

Since the Lifted Heston model reduces to the standard Heston model when \( M = 1 \), \( c_1 = 1 \), and \( x_1 = 0 \), the Lifted Heston model can reproduce all the smiles that the standard Heston model can generate. The inverse, however, is not true, as the standard Heston model often struggles to capture steep enough smiles for short maturities. We provide a comparison of the calibrations of both models using the previously introduced calibration set. To highlight the differences, only the smiles with the poorest calibration results are shown in Figure \ref{fig:comparison_heston_lifhes}.

\begin{figure}[H]
    \begin{center}
    \includegraphics[width=1\linewidth]{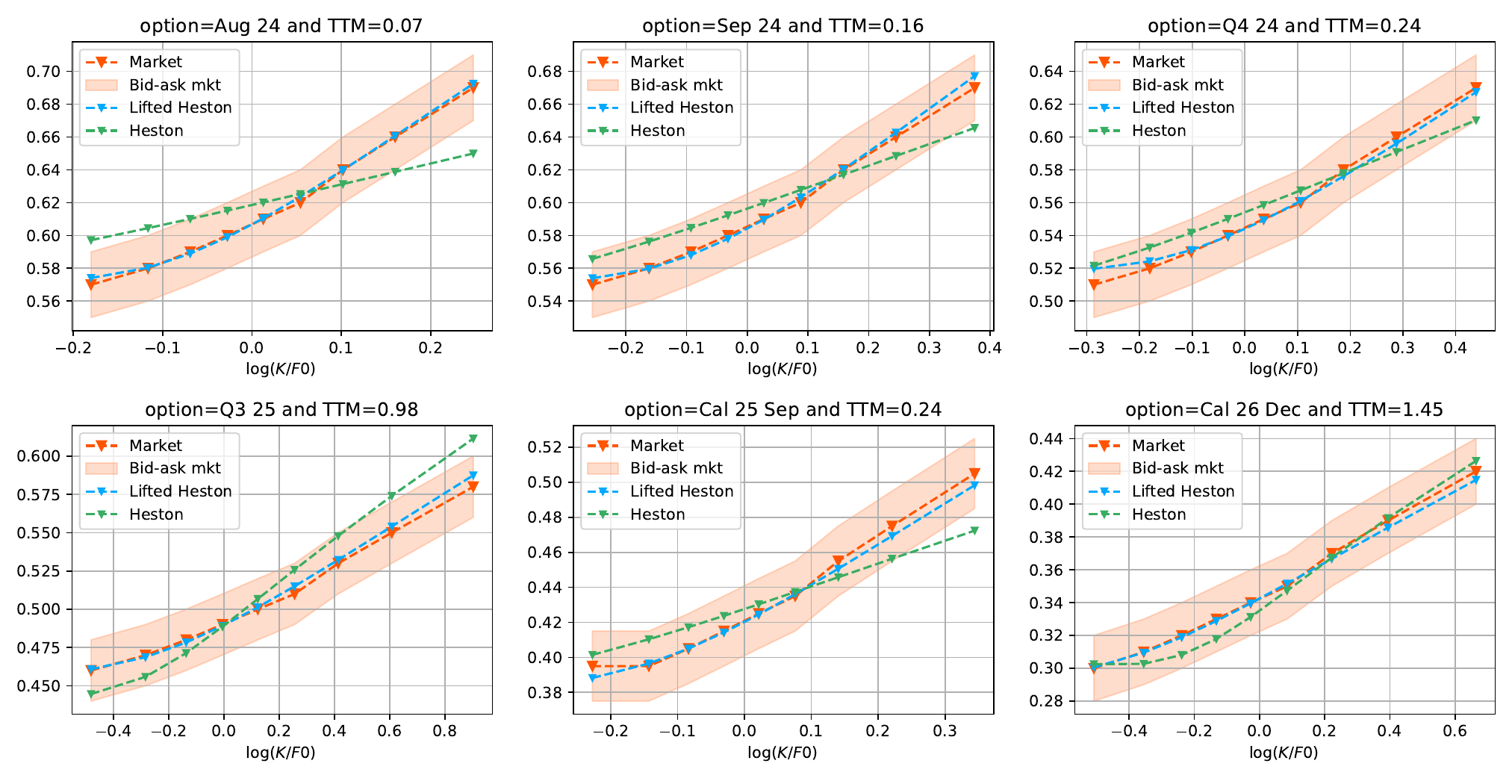}
    \caption{A comparison of the calibration results for the Heston model and the Lifted Heston model with $M = 3$ pseudo-factors on the same calibration set.}
    \label{fig:comparison_heston_lifhes}
    \end{center}
\end{figure}

\subsection{Calibration results for TTF market} \label{ss:ttf_calibration_results}

In this section, we provide the results of calibration performed on the TTF gas market on the $3^{rd}$ July 2024.

Note that options on quarterly and calendar gas futures contracts are quoted as strip options (e.g.~a price of an option on a quarter futures contract is a mean of option prices on monthly contracts forming this quarter). Thus, it is sufficient to calibrate only the smiles corresponding to the monthly futures contracts. It can be achieved by using only the function $h$ at the term structure calibration step, which is shown in Figure \ref{fig:h_TTF}. Thus, the historical correlations between futures remain unchanged and coincide perfectly with the correlations calibrated historically, since the function $g$ remains constant. {The calibrated VS volatility term structure is presented in Figure~\ref{fig:vs_ts_fit_ttf}.}

We display in Figure \ref{F:estimated_correlation_term_structure_TTF} the results of the PC analysis applied to the rolling TTF futures contracts' daily log returns observed from the $3^{rd}$ of July 2024 to the $3^{rd}$ of January 2023, and observe only $2$ PCs allow to reach $95\%$ of explained variance in this case, and $4$ factors to reach $99\%$, consistent again with the quality of fits we obtain in Figure \ref{F:joint_calibration_fit_plots_TTF}, and the observations by \citet[Section 6.2.2]{Andersen2010}. In fact, the high futures correlation structure observed in the upper left of Figure \ref{F:estimated_correlation_term_structure_TTF} stands in sharp contrast to the de-correlation between short-term and long-term observed in the lower left part of Figure \ref{F:estimated_correlation_term_structure} for the German power market, and is mainly due to the large storage capacities of gas in chambers, while electricity still cannot be stored at a large scale.

For the implied calibration at steps 2) and 3), we consider the four-factor model $1L2S1C$. The implied calibration results for six smiles on monthly contracts are presented below in Figure \ref{fig:smiles_TTF}, and the calibrated parameters of the stochastic volatility are provided in Table \ref{calibrated_params_TTF}.

The monthly contracts surface and the ATM volatilities are shown in Figure \ref{fig:surf_TTF}. Note that the term structure of the ATM volatility is not decreasing for short maturities. This is a consequence of the fact that the implied volatility corresponding to the first contract is $5\%$ lower that the volatilities of the other monthly contracts, possibly due to the seasonality. It is also interesting to notice the the difference between daily and monthly volatility is much smaller than in the power market thanks to the very high correlation between all the futures contracts.

\begin{figure}[H]
    \centering
    \begin{minipage}{0.4\textwidth}
        \centering
        \includegraphics[width=\textwidth, trim=0 10 100 10, clip]{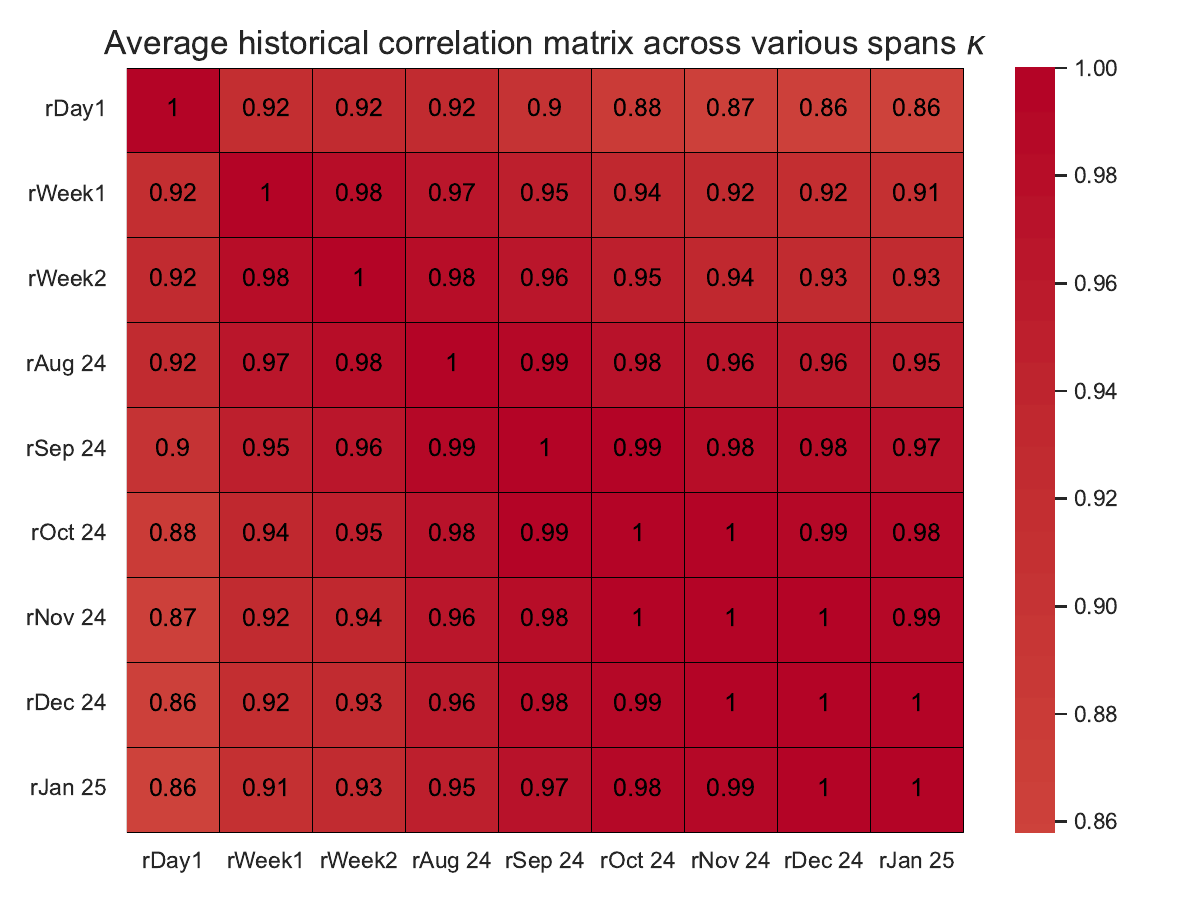} 
    \end{minipage}
    \hspace{0.1cm} 
    \begin{minipage}{0.45\textwidth}
        \centering
        \includegraphics[width=\textwidth, trim=10 0 10 10, clip]{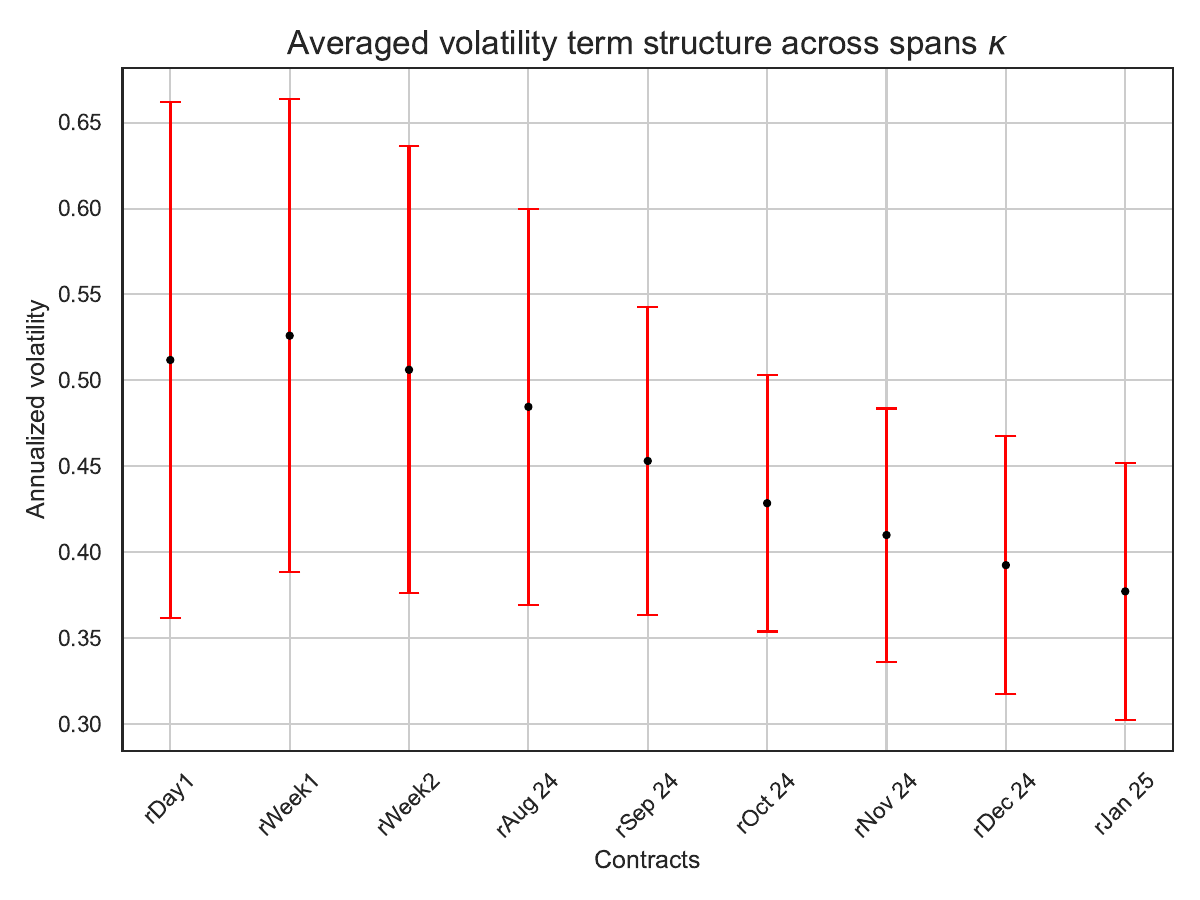} 
    \end{minipage}

    \vspace{0.5cm} 

    \begin{minipage}{0.6\textwidth}
        \centering
        \includegraphics[width=\textwidth, trim=10 0 10 10, clip]{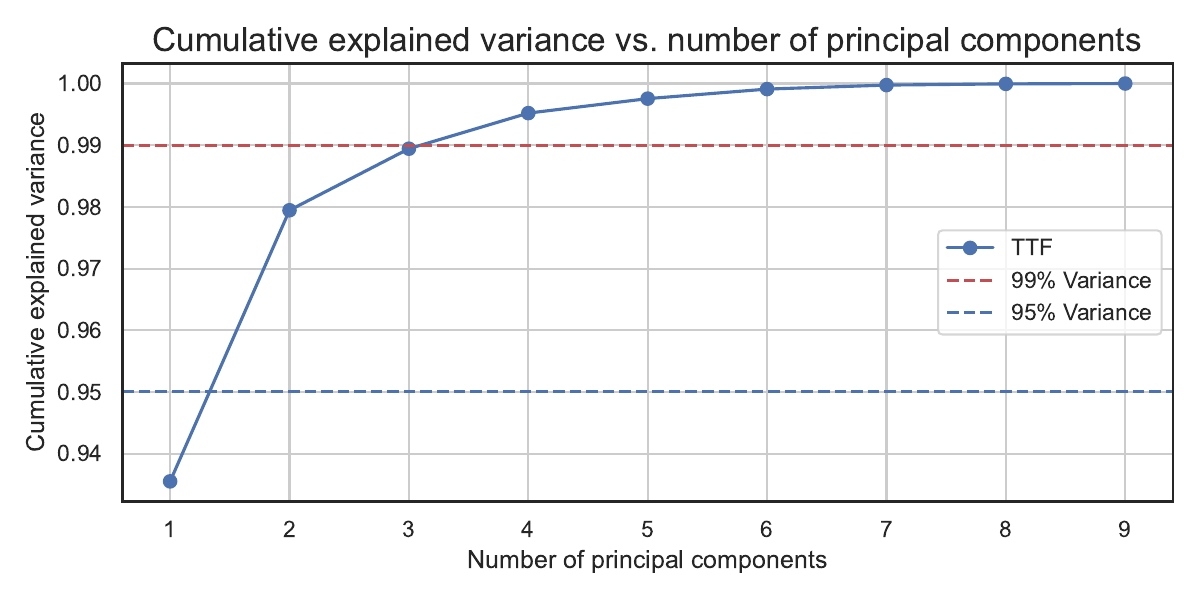} 
    \end{minipage}

    \caption{Rolling futures in TTF gas market: historical (upper left) futures correlation structure and (upper right) volatility term structure (both estimated by averaging across span parameters $\kappa$ going from $30$ to $365$ days, see Appendix \ref{ss:covariance_estimation}); (bottom) principal component analysis}
    \label{F:estimated_correlation_term_structure_TTF}
\end{figure}

\begin{figure}[H]
    \centering
    \adjustbox{trim={0cm 0.2cm 0cm 0.2cm},clip}{%
        \includegraphics[width=0.9\textwidth,angle=0]{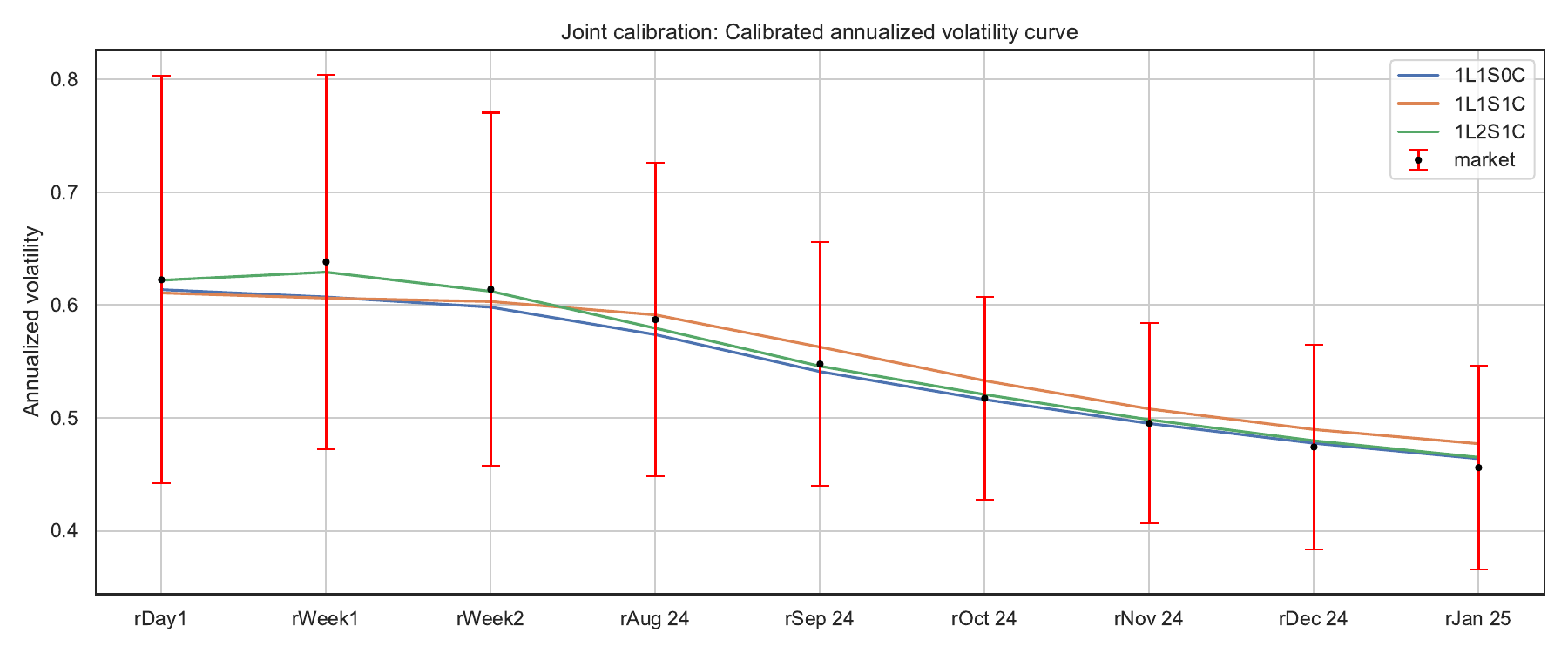}
    }

    \vspace{0.5cm} 

    \adjustbox{trim={0cm 0.2cm 0cm 0.2cm},clip}{%
        \includegraphics[width=0.9\textwidth,angle=0]{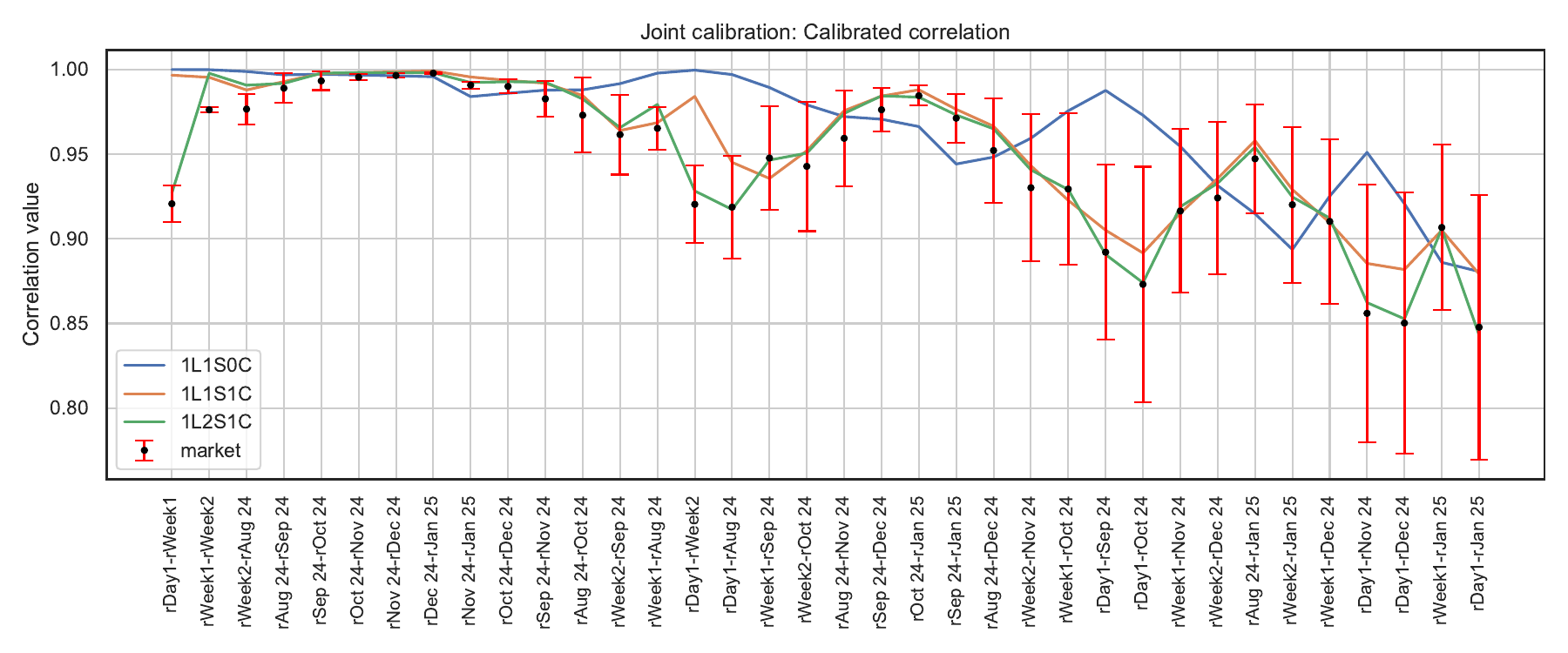}
    }



    \caption{Historical calibration fits for TTF, with $\lambda=0.99$, to (top) the historical realized volatility term structure and (bottom) the historical futures correlation structure.}
    \label{F:joint_calibration_fit_plots_TTF}
\end{figure}

\begin{figure}[H]
    \centering
    \adjustbox{trim={0.1cm 0.1cm 0.1cm 0.1cm},clip}{%
        \includegraphics[width=0.48\textwidth,angle=0]{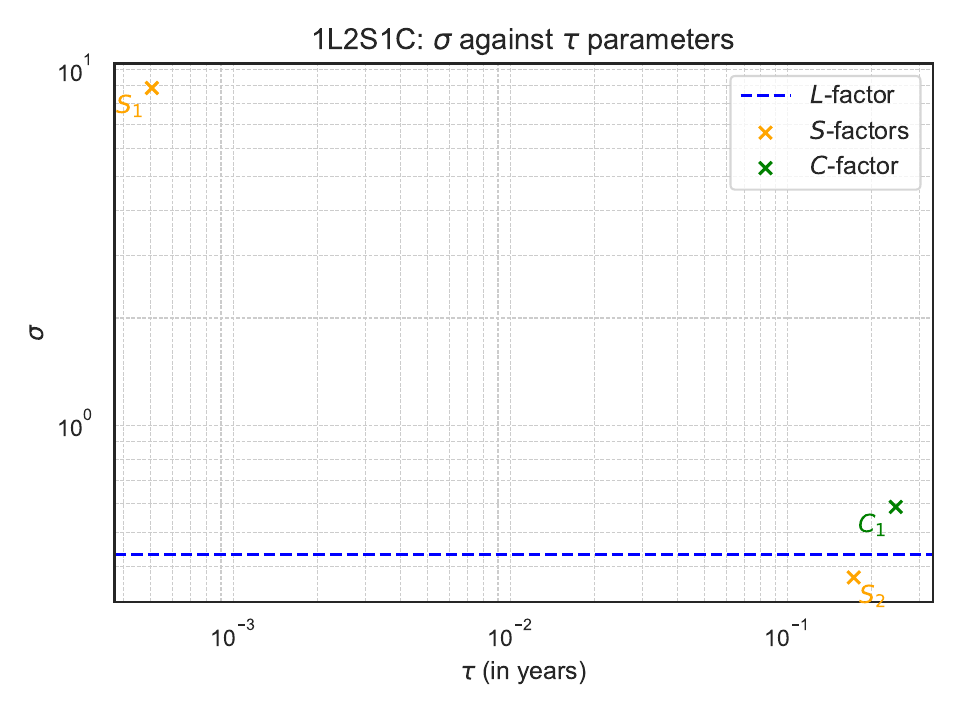}
    }
    \hspace{0.02\textwidth} 
    \adjustbox{trim={0.1cm 0.1cm 0.1cm 0.1cm},clip}{%
        \includegraphics[width=0.45\textwidth,angle=0]{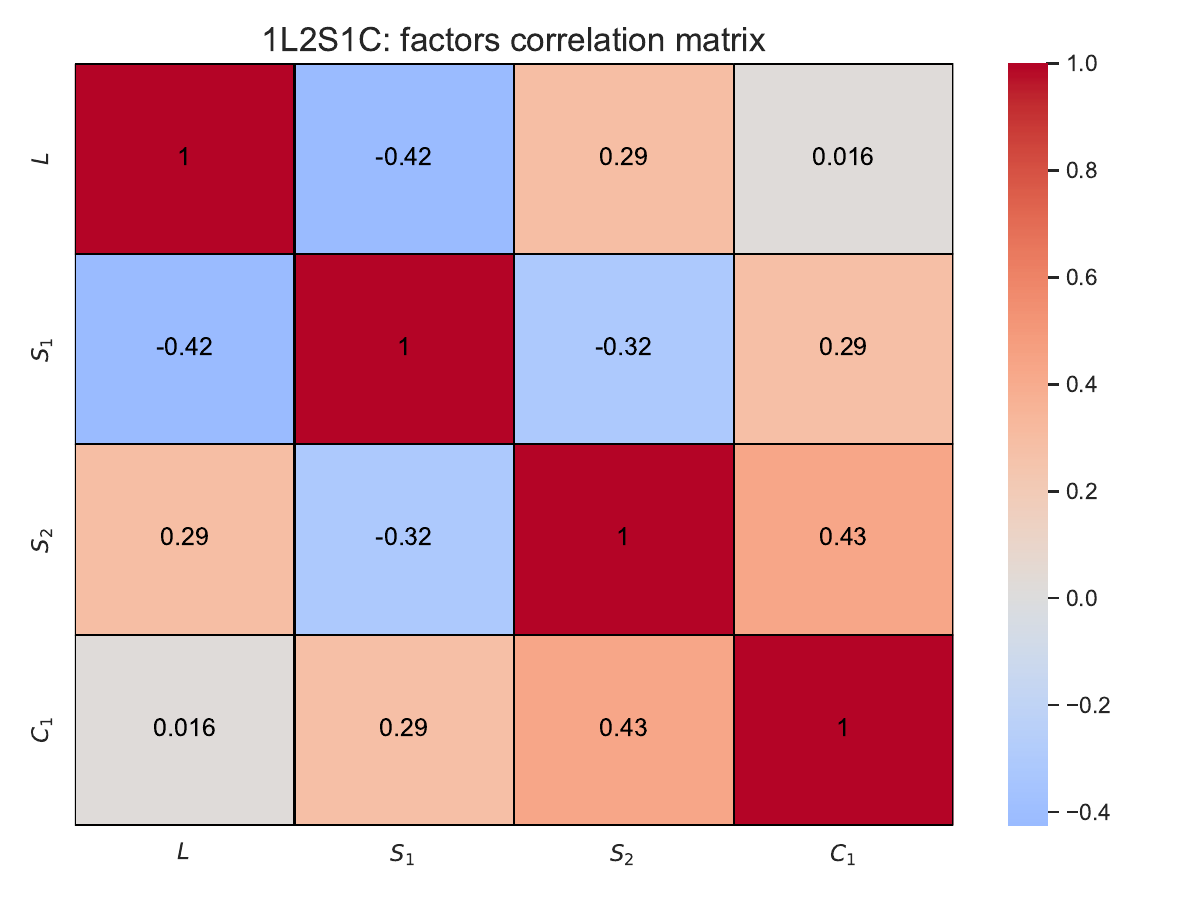}
    }

    \caption{Plots of calibrated $\left(\sigma_{i}\right)_{i}$ against $\left(\tau_{i}\right)_{i}$ parameters (left) and of the factors' correlation matrix (right) for the $1L2S1C$ model on TTF market data, when taking $\lambda=0.99$ in the loss function $J^{\lambda}$ \eqref{eq:def_loss}. Model parameters are, in $L$-$S$-$C$ order: $\left(\sigma_{i}\right)_{i} := [0.4331, 8.8686, 0.3727, 0.5886]$, $\left(\tau_{i}\right)_{i} := [0.0005, 0.1731, 0.2456]$.}
    \label{F:sigma_tau_and_correlation_1L_2S_1C_TTF}
\end{figure}

\begin{figure}[H]
\begin{center}
    \includegraphics[width=0.8\linewidth]{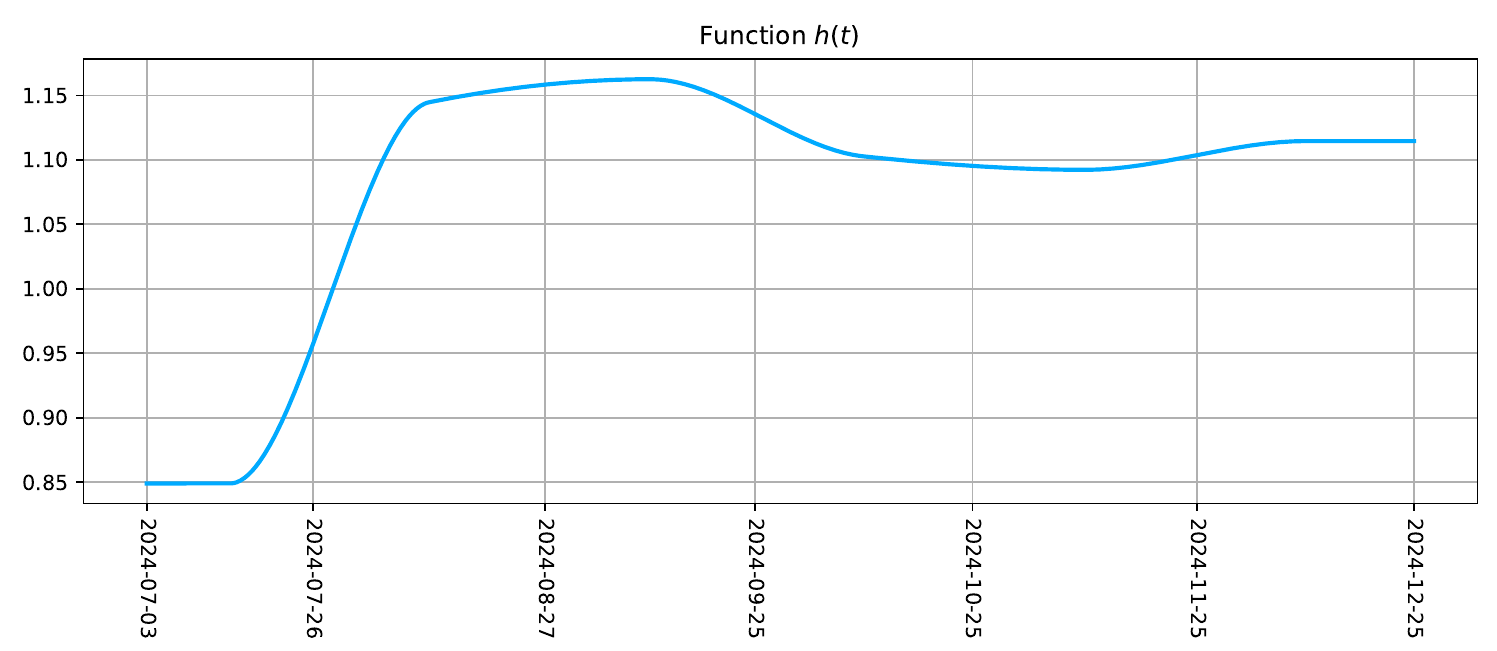}
    \caption{Calibrated function $h$.}
    \label{fig:h_TTF}
\end{center}
\end{figure}

\begin{figure}[H]
\begin{center}    \includegraphics[width=0.6\linewidth]{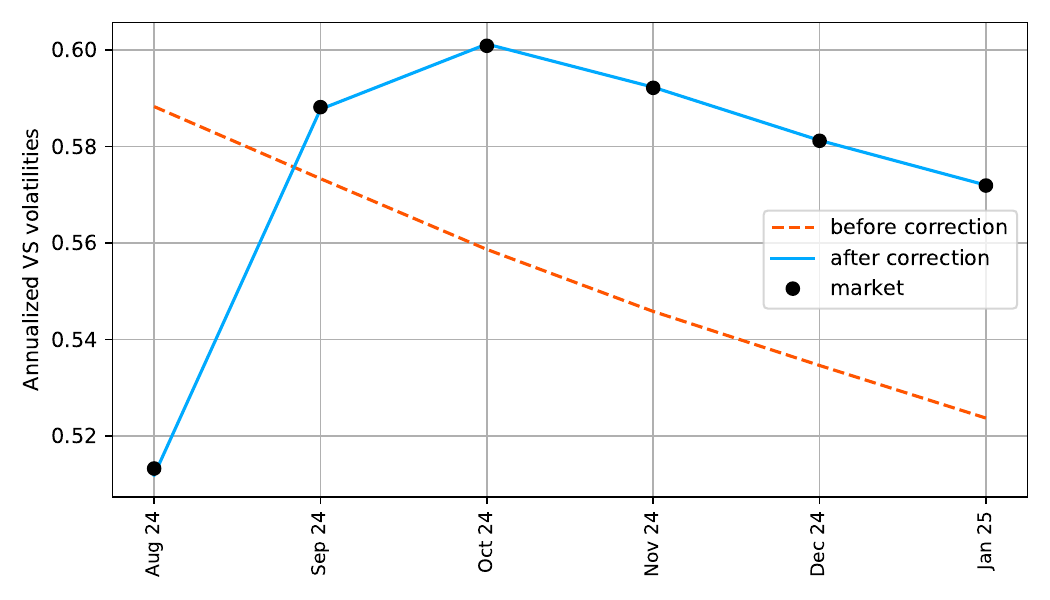}
    \caption{VS term structure fit before the correction (orange) and after the correction (blue).}
    \label{fig:vs_ts_fit_ttf}
\end{center}
\end{figure}

\begin{table}[H]
\begin{center}
\begin{tabular}{ c c} 
 \hline
 Parameter & Calibrated value\\
 \hline \hline
 \rule{0pt}{2ex} $c$ & (1.863, 1.155, 3.747) \\
 
 \rule{0pt}{2ex} $x$ &  (2.586,  4.919, 27.745) \\ 
 
 \rule{0pt}{2ex} $\tilde\rho$ & (0.76 , -0.267, -0.222, -0.272) \\ 
 
 \hline
\end{tabular}
\end{center}
    \caption{Calibrated lifted Heston model parameters}
    \label{calibrated_params_TTF}
\end{table}

\begin{figure}[H]
\begin{center}
    \includegraphics[width=1\linewidth]{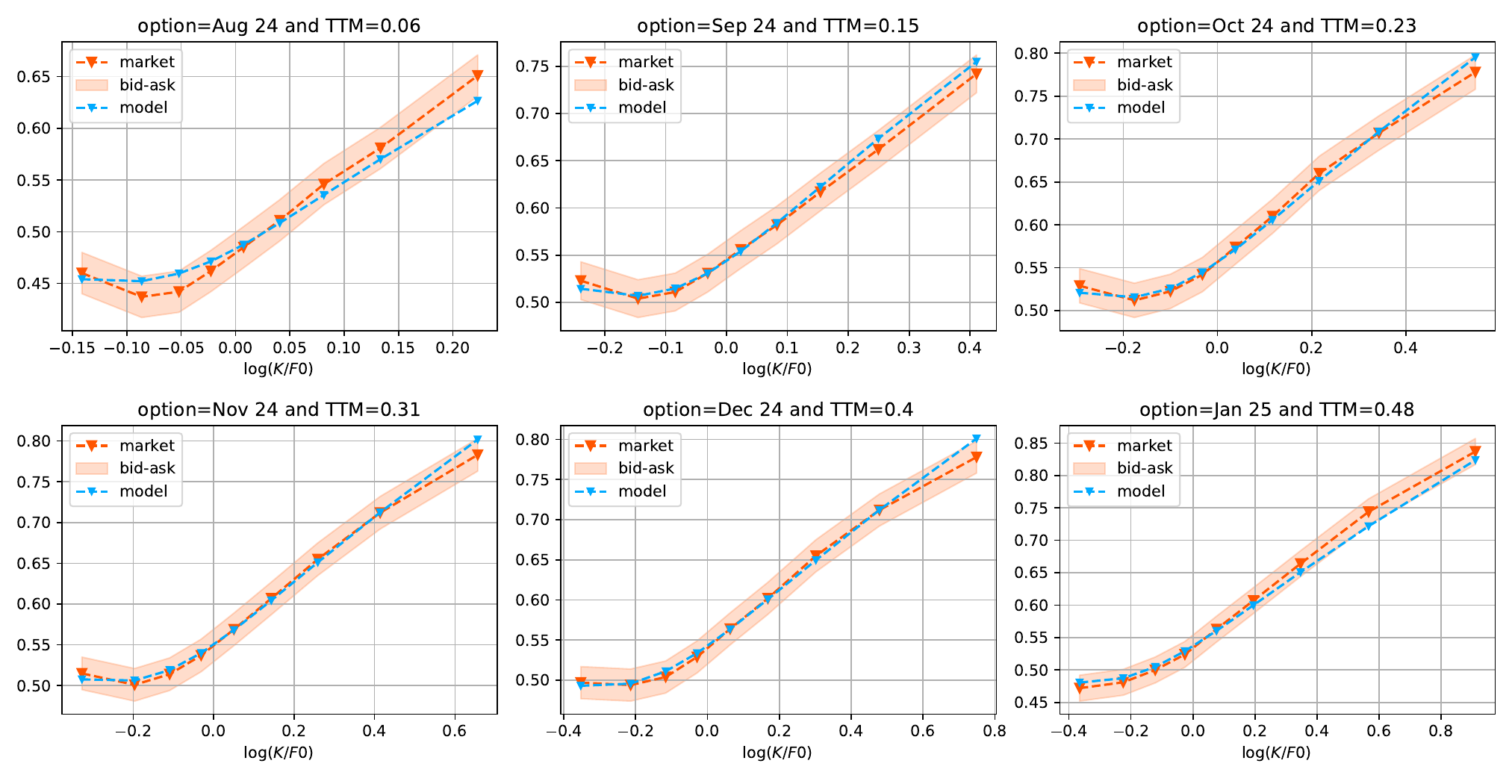}
    \caption{Calibrated IV smiles.}
    \label{fig:smiles_TTF}
\end{center}
\end{figure}

\begin{figure}[H]
    \centering
    \subfigure[\centering IV smiles from August 2024 to February 2026.]{{\includegraphics[width=0.47\linewidth]{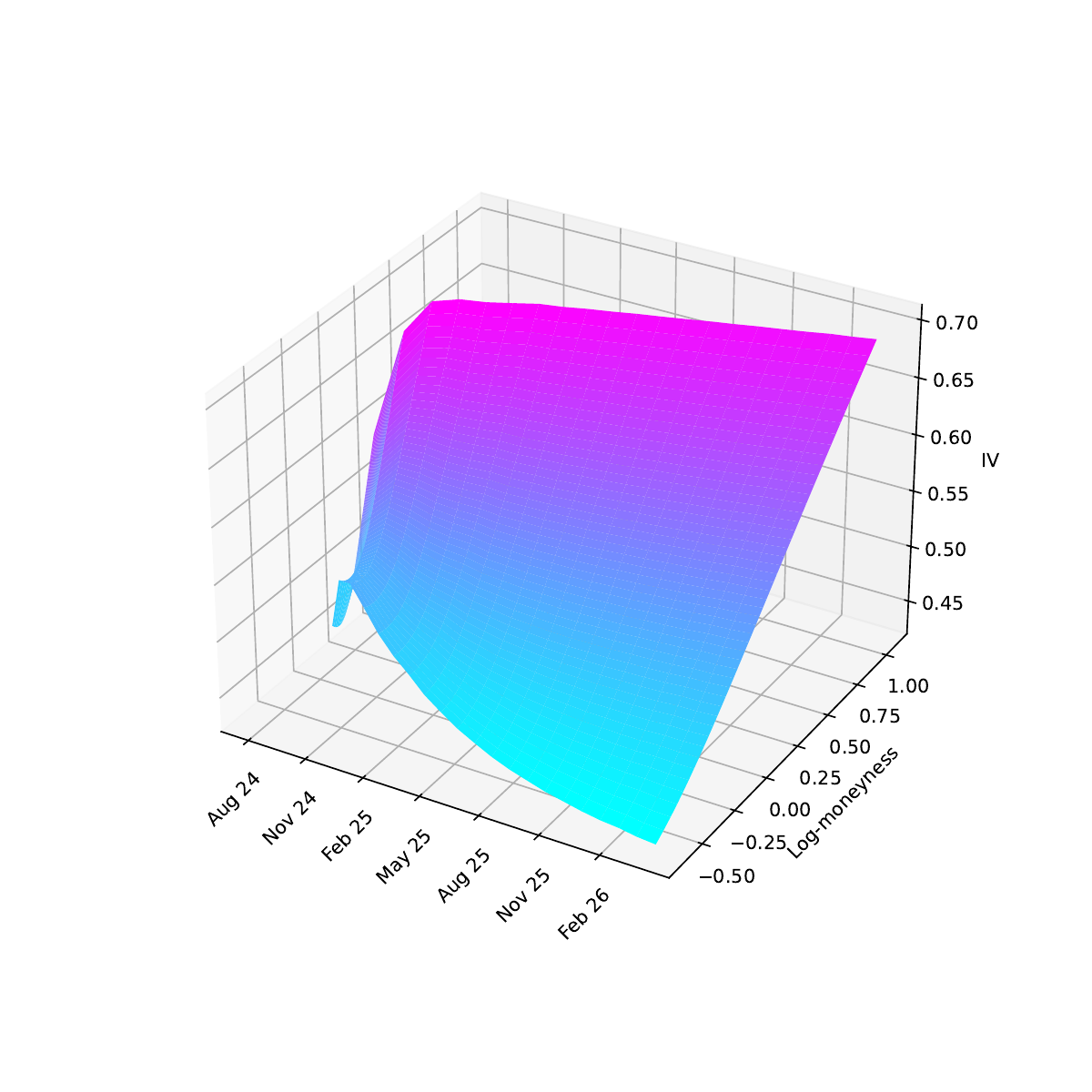} }}%
    \qquad
    \subfigure[\centering ATM volatilities from July 3, 2024 to July 3, 2027.]{{\includegraphics[width=0.47\linewidth]{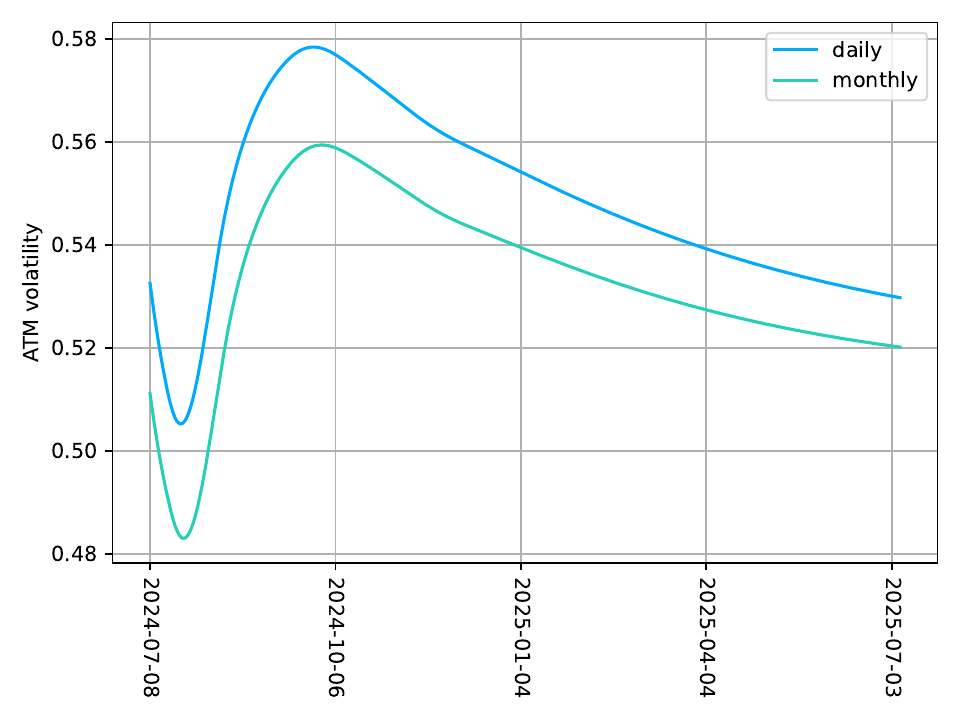}}}%
    \caption{Implied volatility smiles for monthly contracts and implied ATM volatilities of daily and monthly contracts generated by the calibrated model. ((a) and (b) correspondingly).}%
    \label{fig:surf_TTF}
\end{figure}

\subsection{A word on the additive model}\label {section:additive_model}

An additive HJM model with stochastic volatility can be considered as well. In this case, the dynamics is given by
\begin{equation}\label{eq:fwd_with_sto_vol_normal}
     \d f(t,T) = g(T)h(t)\sqrt{V_t}\sum_{i=1}^N \color{black} \sigma_{i}(t,T) \color{black} \d W_t^{i}, \quad t \in [0,\, T].
\end{equation}
However, we found several disadvantages preventing from its efficient usage.
As the market uses the Black-Scholes volatility quotation, instead of the well-known smile flattening effect, the additive model produces negatively skewed smiles for $T \gg 1$, which is hardly acceptable in the commodity market where the skew is typically positive, { i.e.~the ``inverse leverage'' as described in \cite{Andersen2010}}. 

We illustrate it with a simple experiment, calibrating the smile shapes of the contracts Aug 24 and Cal 25 Sep and plotting the smile corresponding to Cal 27 Dec, see Figures \ref{fig:normal_calib}--\ref{fig:log-normal_calib}. The functions $g$ and $h$ were calibrated to all the three smiles. Since the link between the forward variance and the log-contract is absent in the additive model, $g$ and $h$ were calibrated to match the ATM volatility level in this case, while for the multiplicative model we provide both calibration of the ATM volatilities and of the variance swap volatilities (Figures \ref{fig:log-normal_calib_atm} and \ref{fig:log-normal_calib} correspondingly). We observe that the additive model produces a completely inconsistent shape for the third, extrapolated smile. 
The fixed-point ATM volatility calibration algorithm stays the same, but volatility term structure and shape calibration cannot be separated and the calibration routine becomes much more time-consuming.

\begin{figure}[H]
\begin{center}
    \includegraphics[width=1\linewidth]{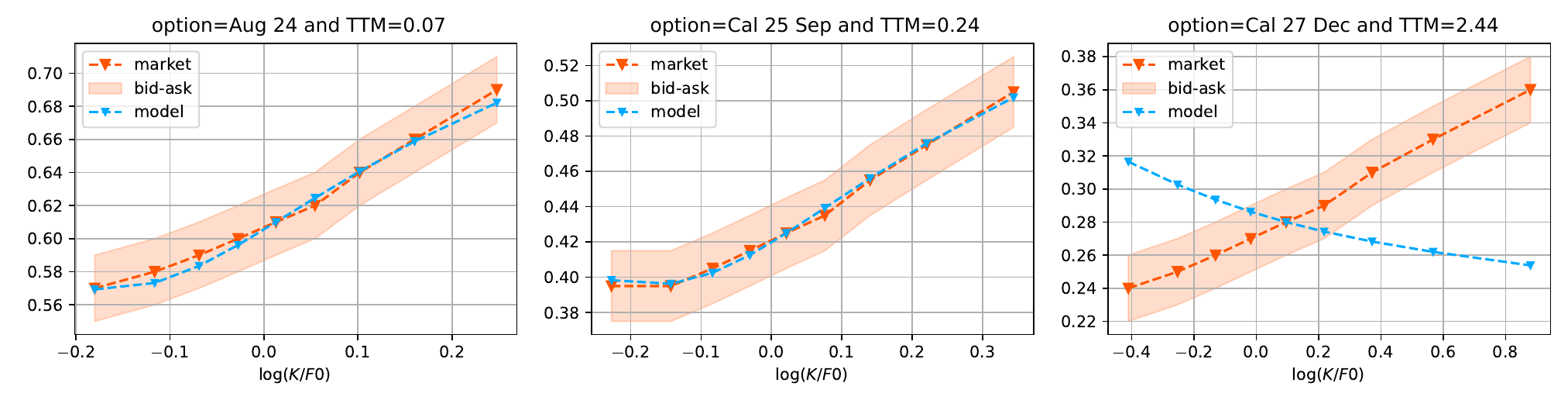}
    \caption{Additive model calibrated to smiles Aug 24 and Cal 25 Sep. \\ Term structure for all the three contracts is calibrated via the ATM volatility.}
    \label{fig:normal_calib}
\end{center}
\end{figure}

\begin{figure}[H]
\begin{center}
    \includegraphics[width=1\linewidth]{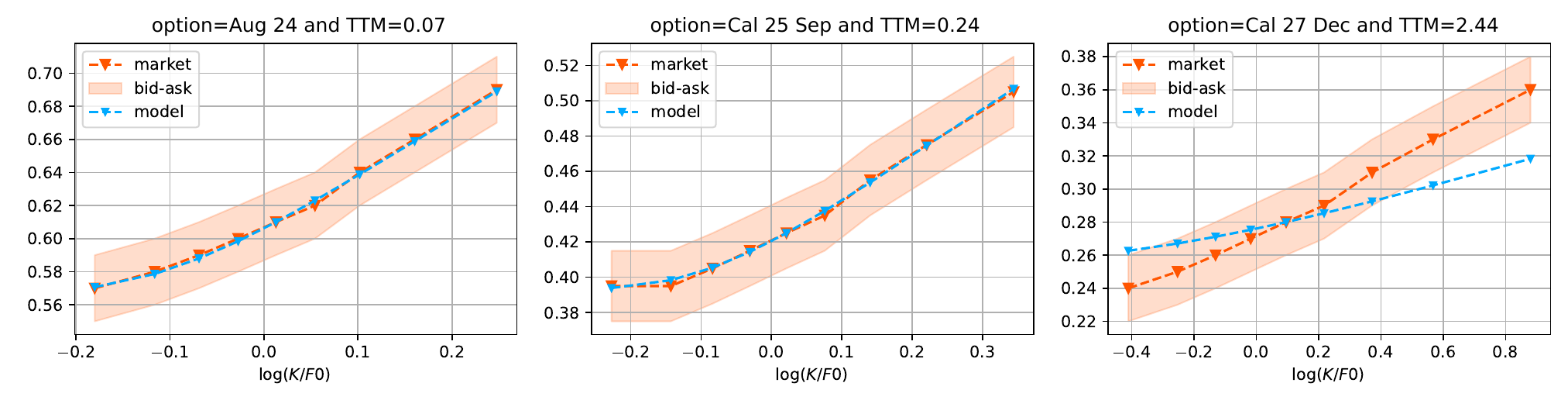}
    \caption{Multiplicative model calibrated to smiles Aug 24 and Cal 25 Sep. \\ Term structure for all the three contracts is calibrated via the ATM volatility.}
    \label{fig:log-normal_calib_atm}
\end{center}
\end{figure}

\begin{figure}[H]
\begin{center}
    \includegraphics[width=1\linewidth]{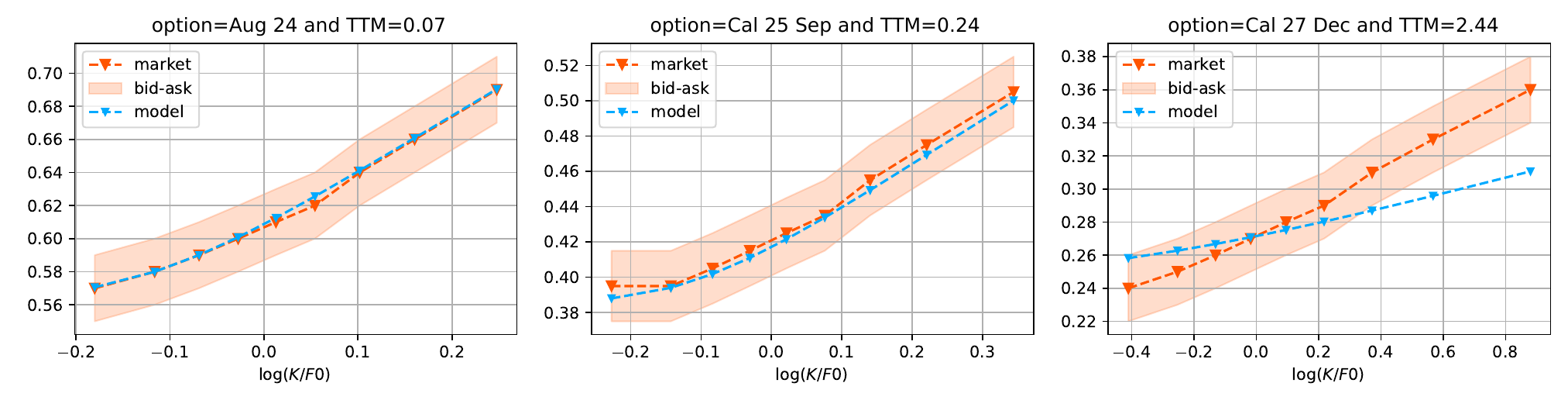}
    \caption{Multiplicative model calibrated to smiles Aug 24 and Cal 25 Sep. \\ Term structure for all the three contracts is calibrated via the variance swap volatility.}
    \label{fig:log-normal_calib}
\end{center}
\end{figure}

\section{Additional details for the step 1) calibration}
\subsection{Estimation of historical covariances} \label{ss:covariance_estimation}

Given two rolling futures contracts' indices $(i,j) \in \{ 1, \ldots, P_{\mathrm{hist}} \}^{2}$, we introduce the following normalized covariance estimator with exponentially decaying weights such that
\begin{equation}
    \widehat{\mathrm{Cov}}_{\kappa} \left( r^{\mathrm{mkt}, i}(\tau_{d}), r^{\mathrm{mkt}, j}(\tau_{d}) \right) := \frac{1}{\tau_{d}} \frac{1}{1-\sum_{h=0}^{H} w_{h}^2(\kappa)} \sum_{h=0}^{H} w_{h}(\kappa) \left( r_{t_{h}}^{\mathrm{mkt}, i}(\tau_{d}) - \bar{r}_{w(\kappa)}^{\mathrm{mkt}, i}(\tau_{d}) \right) \left( r_{t_{h}}^{\mathrm{mkt}, j}(\tau_{d}) - \bar{r}_{w(\kappa)}^{\mathrm{mkt}, j}(\tau_{d}) \right),
\end{equation}
with the exponentially decaying weighted average of log returns given by
\begin{equation}
    \bar{r}_{w(\kappa)}^{\mathrm{mkt}, n}(\tau_{d}) := \sum_{h=0}^{H} w_{h}(\kappa) r_{t_{h}}^{\mathrm{mkt}, n}(\tau_{d}), \quad n \in \{ 1, \ldots, P_{\mathrm{hist}}\},
\end{equation}
and where the weights are computed as follows
\begin{equation} \label{eq:weights_formula}
    w_{h}(\kappa) := \frac{\left( 1 - \alpha_{\kappa} \right)^{H-h}}{\sum_{k=0}^{H} \left( 1 - \alpha_{\kappa} \right)^{H-k}}, \quad h \in \left\{1, \ldots, H\right\}, \quad \alpha_{\kappa} := \frac{2}{\kappa + 1},
\end{equation}
for some time \textit{span} parameter $\kappa \geq 1$ (in number of days) controlling the decay rate of the exponential smoothing. In practice, we consider the covariance estimator of past log returns averaged across a family of $Z \in \mathbb{N}$ time-scale decays $\left( \kappa_{z} \right)_{z \in \{ 1, \ldots, Z \}}$ such that
\begin{equation} \label{eq:historical_covariance_matrix}
    C_{i,j}^{\mathrm{mkt}} := \frac{1}{Z} \sum_{z=1}^{Z} \widehat{\mathrm{Cov}}_{\kappa_{z}} \left( r^{\mathrm{mkt}, i}(\tau_{d}), r^{\mathrm{mkt}, j}(\tau_{d}) \right)
\end{equation}
as well as the element-wise confidence matrix $U$ whose $(i,j)^{th}$ entry is the standard deviation of the $(i,j)^{th}$ entries of the respective covariance matrices $\left( \widehat{\mathrm{Cov}}_{\kappa_{z}} \right)_{z \in \left\{1, \ldots, Z \right\}}$. We then scale $\Gamma \in \mathbb{R}_{+}^{P_{\mathrm{hist}}^{2}}$ weighting the Frobenius norm used in the calibration loss \eqref{eq:def_loss} as
\begin{equation} \label{eq:weight_specification}
    \Gamma \propto \frac{1}{\bar{U}+U},
\end{equation}
where the division is understood element-wise and $\bar{U}$ is the matrix with all entries constant equal to the average entries of $U$. This choice ensures more weights are given to the entries of the covariance matrix of returns which are less volatile across the different decay rates.


\begin{figure}[H]
    \centering
    \begin{minipage}{0.48\textwidth}
        \centering
        \includegraphics[width=\textwidth, trim=0 10 0 0, clip]{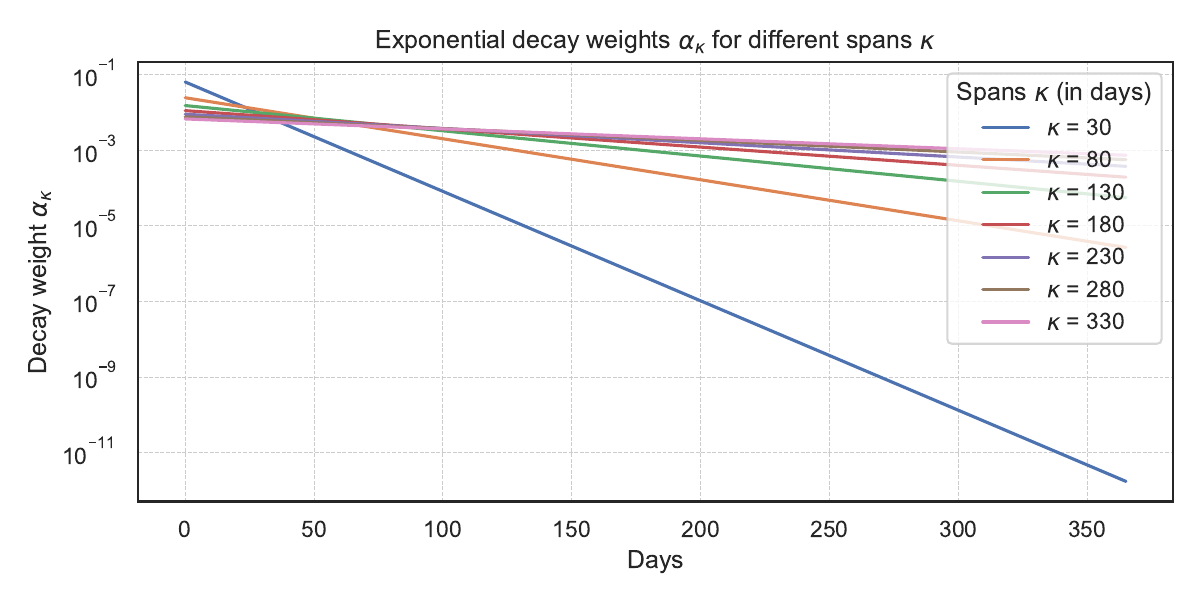} 
    \end{minipage}
    \hspace{0.1cm} 
    \begin{minipage}{0.45\textwidth}
        \centering
        \includegraphics[width=\textwidth, trim=10 0 110 10, clip]{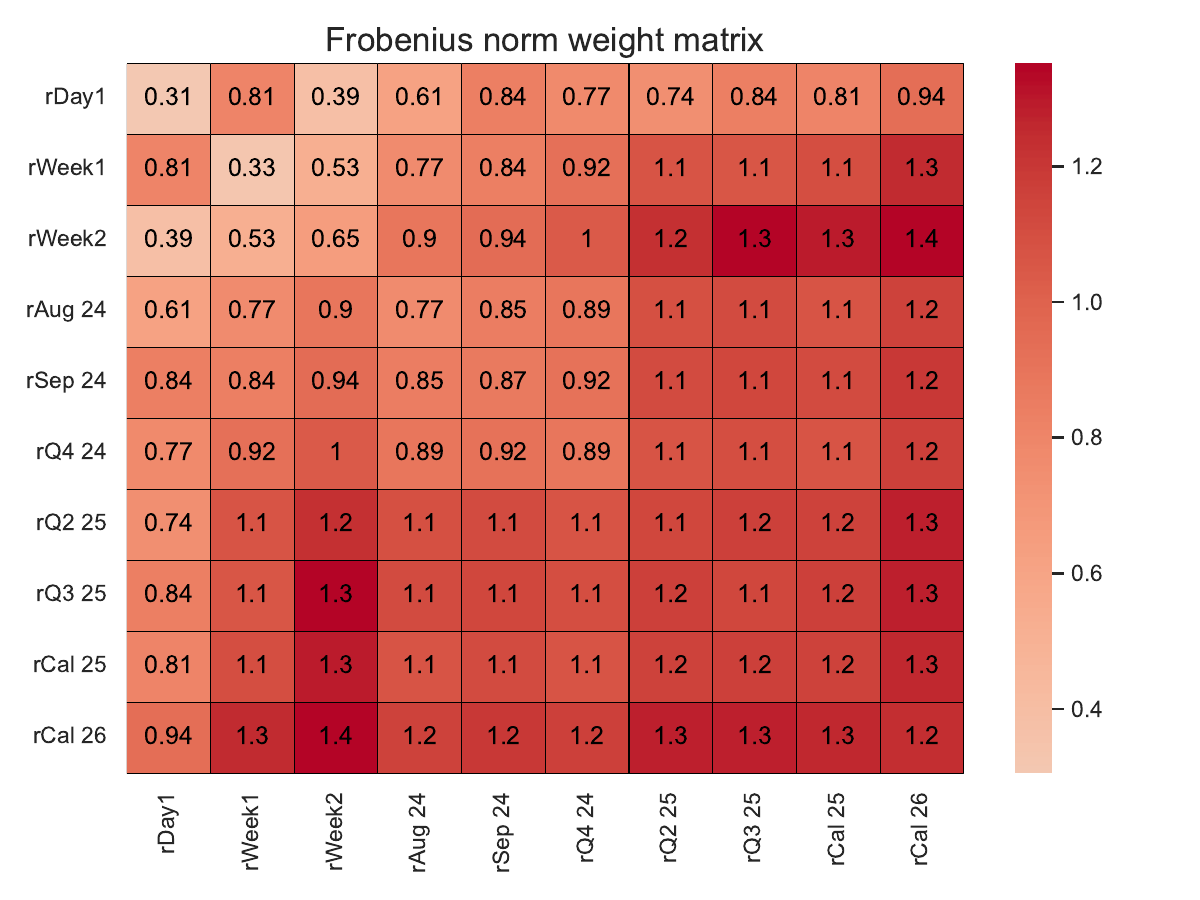} 
    \end{minipage}

    \caption{Left: exponentially decaying weights \eqref{eq:weights_formula} for various span parameters $\kappa \geq 1$; right: resulting weight matrix from \eqref{eq:weight_specification} used in the loss $J_{1}$ \eqref{eq:def_loss_covariance_fit}.}
    \label{F:weights_decay_and_frobenius}
\end{figure}

\subsection{Linear cone programming} \label{s:linear_cone_programming}

The inward minimization problem in \eqref{eq:iterative_minimisation_problem} is formulated as a linear cone program in terms of the variables
\begin{equation} \label{eq:x_seen_as_column_major_vector}
    \bar{x} := \left( \bar{x}_0, \; \left( \bar{x}_{p} \right)_{p \in \left\{ 1, \ldots, \frac{(N+N_{c}+1)(N+N_{c})}{2} \right\}} \right),
\end{equation}
where $\left( \bar{x}_{p} \right)_{p \in \left\{ 1, \ldots, \frac{(N+N_{c}+1)(N+N_{c})}{2} \right\}}$ is identified as a column-major ordered vector corresponding to the lower triangular part of the symmetric matrix variables $\left( x_{p,k}(\sigma, R) \right)_{p, k \in \left\{ 1, \ldots, N+N_{c} \right\}}$ from \eqref{eq:def_quadratic_variables} such that
\begin{align} \label{eq:linear_cone_program}
    & \min_{\bar{x} \in \bar U} c^T \bar{x} \\ \label{eq:linear_cone_program_admissible_set}
    & \bar U := \left\{ \bar{x} \in \mathbb{R}^{1+\frac{(N+N_{c}+1)(N+N_{c})}{2}} \; : \; G^{\lambda}\left( \tau(a) \right)\bar{x} + s = h, \; A\bar{x} = b, \; s \in C \right\},
\end{align}
with $C := C_{0} \times C_{1} \times C_{2}$ such that
\begin{align} \label{eq:non_negative_orthant}
    C_{0} & :=  \left\{ u \in \mathbb{R}^{N} \; : \; u_{i} \geq 0, \; i \in \left\{ 1, \ldots, N \right\} \right\}, \\ \label{eq:def_second_order_cone}
    C_{1} & := \left\{ (u_0, u_1) \in \mathbb{R} \times \mathbb{R}^{\frac{(N+N_{c}+1)(N+N_{c})}{2}} \; : \; u_0 \geq \| u_1 \|_{2} \right\}, \\
    C_{2} & := \left\{ \mathrm{vec}(u) \; | \; u \in \mathbb{S}_{+}^{N} \right\},
\end{align}
and where $\| . \|_{2}$ the standard Euclidean norm, and $\mathrm{vec}(u)$ denotes a symmetric matrix $u$ stored as a vector in column major order.

Indeed, the linear cone program \eqref{eq:linear_cone_program}--\eqref{eq:linear_cone_program_admissible_set} handles the various constraints on $(\sigma, R)$ stated in the inner optimization problem from \eqref{eq:iterative_minimisation_problem} such that:
\begin{itemize}
    \item the non-negative orthant $C_{0}$ \eqref{eq:non_negative_orthant} captures constraints of the form $\sigma_{i}^{\mathrm{upper}} \geq \sigma_{i} \geq 0, \; i \in \left\{ 1, \ldots, N \right\}$, where typically $\sigma_{1}^{\mathrm{upper}}$ is chosen to be a small multiple of the rolling futures contract's annualized volatility associated to the latest delivery to constraint $\sigma_{L}$,
    
    \item we specify $c := \left( 1, 0_{\mathbb{R}^{\frac{(N+N_{c}+1)(N+N_{c})}{2}}} \right)$ so that ``$\bar{x}_{0}$ plays the role of $u_0$'' in the second-order cone $C_{1}$ from \eqref{eq:def_second_order_cone}, and we define accordingly the $\tau(a)$-dependent quantity $G^{\lambda}\left( \tau(a) \right)$ using the weights \eqref{eq:weights_tau_cov}--\eqref{eq:weights_tau_vs_var}, and $h$ with the respective target market values from the loss definition \eqref{eq:loss_functional_quadratic_in_x} so as to ensure both the quadratic loss function $J^{\lambda}$ is minimized and $\left(x_{i,j}\right)_{i,j \in I} \in \mathbb{S}_{+}^{N}$, where $I$ denotes the set of indices relative to $\left\{ L, S, C_{1} \right\}$ state variables, i.e.~$\left(x_{i,j}\right)_{i,j \in I}$ is obtained by withdrawing the columns and rows involving the state variable $C_{2}$ from the matrix $\left( x_{p,k} \right)_{p, k \in \left\{ 1, \ldots, N+N_{c} \right\}}$,

    \item finally, we specify $(A,b)$ accordingly to ensure the equality constraints \eqref{eq:equality_constraints_c_factors} with respect to the state variables of the $C$-factors.
\end{itemize}

In practice, \eqref{eq:linear_cone_program} can be solved numerically by any linear cone program solver such as \textit{conelp} from CVXOPT, see \href{https://cvxopt.org/userguide/coneprog.html}{https://cvxopt.org/userguide/coneprog.html} for additional details, and from which we borrowed the notations.

\section{Proof of Theorem \ref{Thm:fixed_point}} \label{S:thm_fixed-point}
\begin{proof}
Without loss of generality, suppose that $[T_s^j, T_e^j] \subset [T_s^1, T_e^1]$ for $j \in \{2, \ldots {P_{\mathrm{imp}}}\}$, since the delivery periods non included in $[T_s^1, T_e^1]$ will not impact the function $h$. Hence, they can be eliminated from the fixed-point iteration algorithm to be calibrated once at the end.

Given the vector $g = (\bar g_{1}, \ldots, \bar  g_{{P_{\mathrm{imp}}} - 1}) \in \R^{P_{\mathrm{imp}} - 1}$, we construct the vector $h = (h_1, \ldots, h_{N_1}) \in\R^{N_1}$ by \eqref{eq:h_calibration}:
\begin{equation}\label{eq:h_formula}
    h_k^2(g) = \dfrac{\mathrm{VS}^1_{T^1_k} - \mathrm{VS}^1_{T^1_{k-1}}}{\int_{T^1_{k-1}}^{T^1_k} \Sigma_{s}^\top(T_s^1, T_e^1)R\Sigma_{s}(T_s^1, T_e^1)\,\d s}, \quad k \in \{1, \ldots, N_1\}.
\end{equation}
Since $\Sigma_{s}(T_s^1, T_e^1)$ is affine in $g$
\begin{equation}
    \Sigma_{s}(T_s^1, T_e^1) = \sum_{i=2}^{P_{\mathrm{imp}}} \omega_{1}^i \textcolor{black}{\bar g_{i}} \hat\Sigma_{s}(T_s^i, T_e^i) + \left(1-\sum_{i=2}^{P_{\mathrm{imp}}} \omega_{1}^i\right)\hat\Sigma_{s}\left([T_s^1, T_e^1]\setminus \cup_{i=2}^{P_{\mathrm{imp}}} [T_s^i, T_e^i]\right),
\end{equation}
the equation \eqref{eq:h_formula} can be rewritten as
\begin{equation}\label{eq:h_formula_2}
    h_k^2(g) = \dfrac{\mathrm{VS}^1_{T^1_k} - \mathrm{VS}^1_{T^1_{k-1}}}{g^\top Q_k g + 2g^\top p_k + r_k}, \quad k \in \{1, \ldots, N_1\},
\end{equation}
where $Q_k \in\R^{({P_{\mathrm{imp}}} - 1)\times ({P_{\mathrm{imp}}} - 1)}, \quad p_k \in\R^{P_{\mathrm{imp}} - 1},$ and $r_k \in \R$ are defined by
\begin{align}
    (Q_k)_{ij} &= \omega_{1}^i\omega_{1}^j \int_{T^1_{k-1}}^{T^1_k} \hat\Sigma_{s}^\top(T_s^i, T_e^i)R\hat\Sigma_{s}(T_s^j, T_e^j)\,\d s, \quad i, j \in \{2, \ldots, {P_{\mathrm{imp}}}\}, \\
    (p_k)_i &= \omega_{1}^i \Bigl(1-\sum_{l=2}^{P_{\mathrm{imp}}} \omega_{1}^l\Bigl) \int_{T^1_{k-1}}^{T^1_k} \hat\Sigma_{s}^\top(T_s^i, T_e^i)R\hat\Sigma_{s}([T_s^1, T_e^1]\setminus \cup_{l=2}^{P_{\mathrm{imp}}} [T_s^l, T_e^l])\,\d s, \quad i \in \{2, \ldots, {P_{\mathrm{imp}}}\}, \\
    r_k &= \Bigl(1-\sum_{l=2}^{P_{\mathrm{imp}}} \omega_{1}^l\Bigl)^2 \int_{T^1_{k-1}}^{T^1_k} \hat\Sigma_{s}^\top([T_s^1, T_e^1]\setminus \cup_{l=2}^{P_{\mathrm{imp}}} [T_s^l, T_e^l])R\hat\Sigma_{s}([T_s^1, T_e^1]\setminus \cup_{l=2}^{P_{\mathrm{imp}}} [T_s^l, T_e^l])\,\d s
\end{align}
Since we have assumed that the contracts have positive instantaneous correlations, we obtain $(Q_k)_{ij} > 0$ and $(p_k)_i > 0$ for all $i, j \in \{2, \ldots, {P_{\mathrm{imp}}}\}$. Moreover, by the linear independence hypothesis on the division, $[T_s^1, T_e^1]\setminus \cup_{i=2}^{P_{\mathrm{imp}}} [T_s^i, T_e^i] \not= \varnothing$, so that $r_k > 0$.

The new iteration of $g$ calculated by \eqref{eq:g_calibration_1} is then given by
\begin{equation}\label{eq:psi_formula}
    \psi_i(g) = \sqrt{\dfrac{\mathrm{VS}_{T^i}^i}{\sum\limits_{k=1}^{N_1}h_k^2(g)\int\limits_{[0,\,T^i]\cap[T_{k-1},\, T_k]}\hat\Sigma_{s}^\top(T_s^i, T_e^i)R\hat\Sigma_{s}(T_s^i, T_e^i)\,\d s}} = \sqrt{\dfrac{\mathrm{VS}_{T^i}^i}{S_i^\top h^2(g)}}, \quad i \in \{2, \ldots, {P_{\mathrm{imp}}}\}.
\end{equation}
where we have denoted $S_i = \left(\int\limits_{[0,\,T^i]\cap[T_{k-1},\, T_k]}\hat\Sigma_{s}^\top(T_s^i, T_e^i)R\hat\Sigma_{s}(T_s^i, T_e^i)\,\d s\right)_{k \in \{1,\ldots,N_1\}}$. Note all that the elements of $S_i$ are positive.

\textbf{(i) Existence.} We will apply the Brouwer's theorem to establish the existence of the fixed-point. Since all the coefficients are positive, the mapping $\psi: \R_+^{P_{\mathrm{imp}} - 1} \to\R_+^{P_{\mathrm{imp}} - 1}$ is continuous. Thus, it is sufficient to find a convex compact $K \subset \R_+^{P_{\mathrm{imp}} - 1}$, such that $\psi(K) \subset K$. We will look for a compact of the form $K_R = \{x \in \R_+^{P_{\mathrm{imp}} - 1}\colon\ \|x\| \leq R\}$ and we will show that there exists an $R > 0$ big enough, such that $\psi(K_R) \subset K_R$. 

Indeed, if $g \in K_R$, then
\begin{equation}
    h_k^2(g) \geq \dfrac{\mathrm{VS}^1_{T^1_k} - \mathrm{VS}^1_{T^1_{k-1}}}{\|Q_k\| \|g\|^2 + 2\|g\| \|p_k\| + r_k} \geq \dfrac{\mathrm{VS}^1_{T^1_k} - \mathrm{VS}^1_{T^1_{k-1}}}{\|Q_k\| R^2 + 2 \|p_k\| R + r_k}.
\end{equation}
The norm of $\psi$ can be bounded:
\begin{equation}
    \|\psi(g)\|^2 \leq \sum_{i=2}^{P_{\mathrm{imp}}}\dfrac{\mathrm{VS}_{T^i}^i}{\sum\limits_{k=1}^{N_1}S_i^k \dfrac{\mathrm{VS}^1_{T^1_k} - \mathrm{VS}^1_{T^1_{k-1}}}{\|Q_k\| R^2 + 2 \|p_k\| R + r_k}}
\end{equation}
Hence,
\begin{equation}
    \lim_{R\to\infty}\dfrac{\|\psi(g)\|^2}{R^2} \leq \sum_{i=2}^{P_{\mathrm{imp}}}\dfrac{\mathrm{VS}_{T^i}^i}{\sum\limits_{k=1}^{N_1}S_i^k \dfrac{\mathrm{VS}^1_{T^1_k} - \mathrm{VS}^1_{T^1_{k-1}}}{\|Q_k\|}} < 1,
\end{equation}
where we have used \eqref{eq:fixed-point_condition}. Thus, there exists $R$, such that for $g \in K_R$, we have $\|\psi(g)\| < R$, and $\psi(K_R) \subset K_R$.

\textbf{(ii) Stability.} To prove stability of the fixed-point, we linearize $\psi$ and study its Jacobian matrix $\dfrac{\partial\psi}{\partial g}$. Taking the derivative of \eqref{eq:psi_formula}, we obtain
\begin{equation}
    \dfrac{\partial\psi_i(g)}{\partial g_j} = -\dfrac{1}{2} \sqrt{\dfrac{\mathrm{VS}_{T^i}^i}{S_i^\top h^2(g)}}\dfrac{1}{S_i^\top h^2(g)}\sum_{k=1}^{N_1}S_i^k \dfrac{\partial}{\partial g_j} h_k^2(g) = 
    \dfrac{\psi_i(g)}{S_i^\top h^2(g)}\sum_{k=1}^{N_1}S_i^k h_k^2(g)\dfrac{2g^{\top} Q_k^j + 2(p_k)_j}{2g^\top Q_k g + 4g^\top p_k + 2r_k},
\end{equation}
where $Q_k^j$ denotes the $j$-th row of $Q_k$. If $g^*$ is a fixed-point of $\psi$, then $\psi_i(g^*) = g_i^*$, and we have
\begin{equation}\label{eq:A_contraction_ineq}
    \sum_{j=2}^{P_{\mathrm{imp}}} \dfrac{1}{g_i^*}\dfrac{\partial\psi_i(g^*)}{\partial g_j}g_j^* = \dfrac{1}{S_i^\top h^2(g^*)}\sum_{k=1}^{N_1}S_i^k h_k^2(g^*)\underbrace{\dfrac{2g^{*\top} Q_k g^* + 2 g^\top p_k}{2g^{*\top} Q_k g^* + 4g^{*\top} p_k + 2r_k}}_{<1} < 1, \quad i \in \{2, \ldots, {P_{\mathrm{imp}}}\},
\end{equation}
as all the coefficients in this expression are positive. \eqref{eq:A_contraction_ineq} implies that the matrix $A := \mathrm{diag}(g^*)^{-1} \dfrac{\partial\psi}{\partial g}(g^*) \mathrm{diag}(g^*)$ is a contraction in $(\R^{P_{\mathrm{imp}} - 1},\, \|\cdot\|_\infty)$. Thus, there exists a neighborhood $U$ of $g^*$ such that the iterations $g^n$ converge for any $g^0 \in U$. Indeed, defining $e_n = g_n - g^*$, we obtain in $U$
\begin{equation}
    e_{n + 1} = \mathrm{diag}(g^*)A\mathrm{diag}(g^*)^{-1} e_n,
\end{equation}
so that $\|\mathrm{diag}(g^*)^{-1} e_n\|_\infty$ tends to 0 as $n\to\infty$, and $\|e_n\|_\infty\to 0$.

\textbf{(iii) Uniqueness.} We establish uniqueness only for the case $N_1 = 1$. By \eqref{eq:psi_formula}, the fixed-point equation $\psi(g) = g$ reads
\begin{equation}
    \psi_i(g) = \sqrt{\dfrac{1}{S_i^1}\dfrac{\mathrm{VS}_{T^i}^i}{S_i^1 h^2_1(g)}} = \sqrt{\dfrac{1}{S_i^1}\dfrac{\mathrm{VS}_{T^i}^i}{\mathrm{VS}^1_{T^1_1} - \mathrm{VS}^1_{T^1_{0}}}}\sqrt{g^\top Q_1 g + 2g^\top p_1 + r_1}= \alpha(g)e_i,
\end{equation}
where
\begin{equation}
    \mathbf{e} := \left(\sqrt{\dfrac{1}{S_i^1}\dfrac{\mathrm{VS}_{T^i}^i}{\mathrm{VS}^1_{T^1_1} - \mathrm{VS}^1_{T^1_{0}}}}\right)_{i \in \{2, \ldots, {P_{\mathrm{imp}}}\}}, \quad \alpha(g) := \sqrt{g^\top Q_1 g + 2g^\top p_1 + r_1}.
\end{equation}
Is is clear that the fixed point $g$, if exists, should be of the form
\begin{equation}
    g = \beta\mathbf{e}, \quad \beta > 0,
\end{equation}
so that the fixed-point problem in $\R^{P_{\mathrm{imp}}-1}$ is reduced to the one-dimensional fixed-point problem for the mapping
\begin{equation}
    \psi_\beta: \R_+ \to \R_+, \quad \beta \mapsto \alpha(\beta\mathbf{e}) = \sqrt{\mathbf{e}^\top Q_1 \mathbf{e}\beta^2 + 2\mathbf{e}^\top p_1\beta + r_1},
\end{equation}
which is equivalent to the quadratic equation
\begin{equation}\label{eq:fixed-point_qudratic_eq}
    (1 - \mathbf{e}^\top Q_1 \mathbf{e})\beta^2 - 2\mathbf{e}^\top p_1\beta - r_1 = 0.
\end{equation}
The condition \eqref{eq:fixed-point_condition} guarantees that $\mathbf{e}^\top Q_1 \mathbf{e} < 1$, so that the equation \eqref{eq:fixed-point_qudratic_eq} admits a unique positive root $\beta^*$ and the unique fixed-point of $\psi$ is given by $g^* = \beta^*\mathbf{e}$.
\end{proof}

\section{Monte Carlo simulation scheme} \label{section:monte_carlo_scheme}
The European call option prices $C(T,K)$ can also be computed using the Monte Carlo method. At each step $t$ of our discrete-time grid, given $(U_t^i)_{i \in \{1,\ldots, M\}}, V_t,$ and $F_t$ and a time-step $\Delta t$, we simulate $(U_{t+\Delta t}^i)_{i \in \{1,\ldots, M\}}$, then $V_{t+\Delta t}$, and finally $F_{t+\Delta t}$. For the processes $U_t^i$, we start with the following semi-implicit Euler discretization scheme of \eqref{eq:U_def}:
\begin{equation}
     U_{t+\Delta t}^i - U_{t}^i  = - x_i  U_{t+h}^i\, \Delta t - \lambda V_t \,\Delta t + \nu \sqrt{V_t}(B_{t + \Delta t} - B_t),
\end{equation}
   which leads to
\begin{equation}
    U_{t+\Delta t}^i  = \frac{1}{1+x_i \Delta t} \bigl(U_t^i - \lambda V_t \Delta t + \nu \sqrt{V_t}(B_{t + \Delta t} - B_t) \bigr) \\
\end{equation}

We choose this semi-implicit Euler scheme following \cite{lifted2019}, since it gives more stable results than standard Euler scheme which explodes for large mean-reversion coefficients $x_i$. 

We obtain the variance $V_{t+ \Delta t}$ directly by \eqref{eq:V_def} ensuring that the variance process is floored at zero to prevent negative values. For $\log F_{t+ \Delta t}$, standard Euler scheme is used. We sum up the simulation scheme with the following formulae:
\begin{align*}
    U_{t+\Delta t}^i &= \frac{1}{1+x_i \Delta t}(U_t^i - \lambda V_t \Delta t + \nu \sqrt{V_t}(B_{t+\Delta t} - B_t)), \quad i \in \{1, \ldots, M\}, \\
    V_{t+\Delta t} &= \left(m_0(t+\Delta t) + \sum_{i=1}^M c_i U_{t+\Delta t}^i\right)^+, \\
     \log F_{t+\Delta t}  &= \log F_{t} - \frac{1}{2} h(t)^ 2V_t \Sigma_t^\top R \Sigma_t \Delta t + h(t)\sqrt{V_t}\Sigma_t^\top(W_{t+\Delta t} - W_t).
\end{align*}

\bibliographystyle{plainnat}
\bibliography{refs.bib}

\begin{thebibliography}{40}
\providecommand{\natexlab}[1]{#1}
\providecommand{\url}[1]{\texttt{#1}}
\expandafter\ifx\csname urlstyle\endcsname\relax
  \providecommand{\doi}[1]{doi: #1}\else
  \providecommand{\doi}{doi: \begingroup \urlstyle{rm}\Url}\fi

\bibitem[Abi~Jaber(2019)]{lifted2019}
Eduardo Abi~Jaber.
\newblock Lifting the {H}eston model.
\newblock \emph{Quantitative Finance}, 19\penalty0 (12):\penalty0 1995--2013, 2019.

\bibitem[Abi~Jaber et~al.(2019)Abi~Jaber, Larsson, and Pulido]{jaber2019affinevolterraprocesses}
Eduardo Abi~Jaber, Martin Larsson, and Sergio Pulido.
\newblock Affine {V}olterra processes.
\newblock \emph{The Annals of Applied Probability}, 29\penalty0 (5):\penalty0 3155--3200, 2019.

\bibitem[Abi~Jaber et~al.(2024)Abi~Jaber, Bayer, and Breneis]{abijaber2024state}
Eduardo Abi~Jaber, Christian Bayer, and Simon Breneis.
\newblock State spaces of multifactor approximations of nonnegative {V}olterra processes.
\newblock \emph{arXiv:2412.17526}, 2024.

\bibitem[Andersen(2010)]{Andersen2010}
Leif B.~G. Andersen.
\newblock Markov models for commodity futures: Theory and practice.
\newblock \emph{Quantitative Finance}, 10\penalty0 (8), 2010.

\bibitem[Benth and Harang(2021)]{benth2021infinite}
Fred~Espen Benth and Fabian~A Harang.
\newblock Infinite dimensional pathwise volterra processes driven by {G}aussian noise--probabilistic properties and applications.
\newblock \emph{Electronic Journal of Probability}, 26:\penalty0 1--42, 2021.

\bibitem[Benth and Koekebakker(2008)]{benth2008modeling}
Fred~Espen Benth and Steen Koekebakker.
\newblock Stochastic modeling of financial electricity contracts.
\newblock \emph{Energy Economics}, 30\penalty0 (3):\penalty0 1116--1157, 2008.

\bibitem[Benth and Kr{\"u}hner(2023)]{benth2023stochastic}
Fred~Espen Benth and Paul Kr{\"u}hner.
\newblock \emph{Stochastic Models for Prices Dynamics in Energy and Commodity Markets}.
\newblock Springer, 2023.

\bibitem[Benth and Paraschiv(2018)]{benth2018space}
Fred~Espen Benth and Florentina Paraschiv.
\newblock A space-time random field model for electricity forward prices.
\newblock \emph{Journal of Banking \& Finance}, 95:\penalty0 203--216, 2018.

\bibitem[Benth and Simonsen(2018)]{benth2018heston}
Fred~Espen Benth and Iben~Cathrine Simonsen.
\newblock The {H}eston stochastic volatility model in {H}ilbert space.
\newblock \emph{Stochastic Analysis and Applications}, 36\penalty0 (4):\penalty0 733--750, 2018.

\bibitem[Benth et~al.(2008)Benth, Benth, and Koekebakker]{benth2008stochastic}
Fred~Espen Benth, Jurate~Saltyte Benth, and Steen Koekebakker.
\newblock \emph{Stochastic modelling of electricity and related markets}, volume~11.
\newblock World Scientific, 2008.

\bibitem[Benth et~al.(2019)Benth, Piccirilli, and Vargiolu]{benth2017additive}
Fred~Espen Benth, Marco Piccirilli, and Tiziano Vargiolu.
\newblock Mean-reverting additive energy forward curves in a {H}eath--{J}arrow--{M}orton framework.
\newblock \emph{Mathematics and Financial Economics}, 13\penalty0 (4):\penalty0 543--577, 2019.

\bibitem[Bergomi and Guyon(2012)]{BergomiGuyon2011}
Lorenzo Bergomi and Julien Guyon.
\newblock Stochastic volatility's orderly smiles.
\newblock \emph{Risk}, 25\penalty0 (5):\penalty0 60, 2012.

\bibitem[Boyd and Vandenberghe(2004)]{boyd2004convex}
Stephen Boyd and Lieven Vandenberghe.
\newblock \emph{Convex optimization}.
\newblock Cambridge university press, 2004.

\bibitem[Brigo and Mercurio(2006)]{Mercurio2006}
Damiano Brigo and Fabio Mercurio.
\newblock \emph{Interest Rate Models — Theory and Practice: With Smile, Inflation and Credit}.
\newblock Springer Verlag, 2006.

\bibitem[Carmona and Durrleman(2003)]{carmona2003pricing}
Ren{\'e} Carmona and Valdo Durrleman.
\newblock Pricing and hedging spread options.
\newblock \emph{Siam Review}, 45\penalty0 (4):\penalty0 627--685, 2003.

\bibitem[Carr and Madan(1998)]{carr1998towards}
Peter Carr and Dilip Madan.
\newblock Towards a theory of volatility trading.
\newblock \emph{Volatility: New estimation techniques for pricing derivatives}, 29:\penalty0 417--427, 1998.

\bibitem[Cartea and Figueroa(2005)]{cartea2005pricing}
Alvaro Cartea and Marcelo~G Figueroa.
\newblock Pricing in electricity markets: a mean reverting jump diffusion model with seasonality.
\newblock \emph{Applied Mathematical Finance}, 12\penalty0 (4):\penalty0 313--335, 2005.

\bibitem[Clewlow and Strickland(1999)]{Clewlow1999}
Les Clewlow and Chris Strickland.
\newblock Valuing energy options in a one factor model fitted to forward prices. research paper series 10.
\newblock \emph{Quantitative Finance Research Centre, University of Technology, Sydney}, 10, 1999.

\bibitem[Cortazar et~al.(2017)Cortazar, Lopez, and Naranjo]{cortazar2017multifactor}
Gonzalo Cortazar, Matias Lopez, and Lorenzo Naranjo.
\newblock A multifactor stochastic volatility model of commodity prices.
\newblock \emph{Energy Economics}, 67:\penalty0 182--201, 2017.

\bibitem[Cox et~al.(2022)Cox, Karbach, and Khedher]{cox2022infinite}
Sonja Cox, Sven Karbach, and Asma Khedher.
\newblock An infinite-dimensional affine stochastic volatility model.
\newblock \emph{Mathematical Finance}, 32\penalty0 (3):\penalty0 878--906, 2022.

\bibitem[Demeterfi et~al.(1999)Demeterfi, Derman, Kamal, and Zou]{demeterfi1999more}
Kresimir Demeterfi, Emanuel Derman, Michael Kamal, and Joseph Zou.
\newblock More than you ever wanted to know about volatility swaps.
\newblock \emph{Goldman Sachs quantitative strategies research notes}, 41:\penalty0 1--56, 1999.

\bibitem[Deschatre et~al.(2021)Deschatre, F{\'e}ron, and Gruet]{deschatre2021survey}
Thomas Deschatre, Olivier F{\'e}ron, and Pierre Gruet.
\newblock A survey of electricity spot and futures price models for risk management applications.
\newblock \emph{Energy Economics}, 102:\penalty0 105504, 2021.

\bibitem[Edoli et~al.(2013)Edoli, Tasinato, and Vargiolu]{edoli2013calibration}
Enrico Edoli, Davide Tasinato, and Tiziano Vargiolu.
\newblock Calibration of a multifactor model for the forward markets of several commodities.
\newblock \emph{Optimization}, 62\penalty0 (11):\penalty0 1553--1574, 2013.

\bibitem[F{\'e}ron and Gruet(2024)]{feron2024estimation}
Olivier F{\'e}ron and Pierre Gruet.
\newblock Estimation of the number of factors in a multi-factorial {H}eath-{J}arrow-{M}orton model in power markets.
\newblock In \emph{Quantitative Energy Finance: Recent Trends and Developments}, pages 3--39. Springer, 2024.

\bibitem[Fritsch and Butland(1984)]{fritsch1984method}
Frederick~N Fritsch and Judy Butland.
\newblock A method for constructing local monotone piecewise cubic interpolants.
\newblock \emph{SIAM journal on scientific and statistical computing}, 5\penalty0 (2):\penalty0 300--304, 1984.

\bibitem[Gardini and Santilli(2024)]{gardini2023heath}
Matteo Gardini and Edoardo Santilli.
\newblock A {H}eath--{J}arrow--{M}orton framework for energy markets: review and applications for practitioners.
\newblock \emph{Decisions in Economics and Finance}, 2024.
\newblock ISSN 1129-6569.

\bibitem[Gatheral and Jacquier(2014)]{GatheralSSVI}
Jim Gatheral and Antoine Jacquier.
\newblock Arbitrage-free {SVI} volatility surfaces.
\newblock \emph{Quantitative Finance}, 14\penalty0 (1):\penalty0 59--71, 2014.

\bibitem[Heath et~al.(1992)Heath, Jarrow, and Morton]{HJM1992}
David Heath, Robert Jarrow, and Andrew Morton.
\newblock Bond pricing and the term structure of interest rates: A new methodology for contingent claims valuation.
\newblock \emph{Econometrica}, 60\penalty0 (1):\penalty0 77--105, 1992.
\newblock ISSN 00129682, 14680262.

\bibitem[Horn and Johnson(2013)]{horn2013matrix}
Roger~A. Horn and Charles~R. Johnson.
\newblock \emph{Matrix Analysis}.
\newblock Matrix Analysis. Cambridge University Press, 2013.
\newblock ISBN 9780521839402.

\bibitem[Kemna and Vorst(1990)]{KEMNA1990113}
Angelien Kemna and A.C.F. Vorst.
\newblock A pricing method for options based on average asset values.
\newblock \emph{Journal of Banking \& Finance}, 14\penalty0 (1):\penalty0 113--129, 1990.

\bibitem[Kiesel et~al.(2009)Kiesel, Schindlmayr, and B{\"o}rger]{Kiesel2009}
R{\"u}diger Kiesel, Gero Schindlmayr, and Reik~H B{\"o}rger.
\newblock A two-factor model for the electricity forward market.
\newblock \emph{Quantitative Finance}, 9\penalty0 (3):\penalty0 279--287, 2009.

\bibitem[Koekebakker and Ollmar(2005)]{koekebakker2005forward}
Steen Koekebakker and Fridthjof Ollmar.
\newblock Forward curve dynamics in the nordic electricity market.
\newblock \emph{Managerial Finance}, 31\penalty0 (6):\penalty0 73--94, 2005.

\bibitem[Lewis(2001)]{Lewis2001}
Alan Lewis.
\newblock A simple option formula for general jump-diffusion and other exponential {L}evy processes.
\newblock \emph{SSRN Electronic Journal}, 2001.

\bibitem[Mishura et~al.(2023)Mishura, Ottaviano, and Vargiolu]{Mishura2023GaussianVP}
Yuliya Mishura, Stefania Ottaviano, and Tiziano Vargiolu.
\newblock Gaussian {V}olterra processes as models of electricity markets.
\newblock \emph{SSRN Electronic Journal}, 2023.

\bibitem[Musti et~al.(2016)Musti, Fanelli, and Maddalena]{Fanelli2016path-dep}
Silvana Musti, Viviana Fanelli, and Lucia Maddalena.
\newblock Modelling electricity futures prices using seasonal path-dependent volatility.
\newblock \emph{Applied Energy}, 173:\penalty0 92--102, 2016.

\bibitem[Nelson and Siegel(1987)]{Nelson1987}
Charles Nelson and Andrew Siegel.
\newblock Parsimonious modeling of yield curves.
\newblock \emph{The Journal of Business}, 60:\penalty0 473--89, 1987.

\bibitem[Piccirilli et~al.(2021)Piccirilli, Schmeck, and Vargiolu]{piccirilli2021capturing}
Marco Piccirilli, Maren~Diane Schmeck, and Tiziano Vargiolu.
\newblock Capturing the power options smile by an additive two-factor model for overlapping futures prices.
\newblock \emph{Energy Economics}, 95:\penalty0 105006, 2021.

\bibitem[Samuelson(2016)]{samuelson2016proof}
Paul~A Samuelson.
\newblock Proof that properly anticipated prices fluctuate randomly.
\newblock In \emph{The world scientific handbook of futures markets}, pages 25--38. World Scientific, 2016.

\bibitem[Schmeck and Schwerin(2021)]{Schmeck2021TheEO}
Maren~Diane Schmeck and Stefan Schwerin.
\newblock The effect of mean-reverting processes in the pricing of options in the energy market: An arithmetic approach.
\newblock \emph{Risks}, 2021.

\bibitem[Sepp and Rakhmonov(2023)]{Sepp2023}
Artur Sepp and Parviz Rakhmonov.
\newblock Stochastic volatility for factor {HJM} framework.
\newblock \emph{SSRN Electronic Journal}, 01 2023.
\newblock \doi{10.2139/ssrn.4646925}.

\end{thebibliography}

\end{document}